\documentclass{article}
\usepackage[T1]{fontenc}
\usepackage{authblk}
\usepackage{geometry}
\usepackage{amsfonts}
\usepackage{graphicx,amssymb,mathrsfs,amsmath,color,fancyhdr}
\usepackage{subcaption}
\usepackage{amsthm}
\usepackage{manfnt}
\usepackage{cases}
\usepackage{mathbbol}
\usepackage{bm,bbm}
\usepackage{setspace}
\usepackage{hyperref}
\hypersetup{colorlinks = True,linkcolor = blue,anchorcolor =red,citecolor = blue,filecolor = red,urlcolor = red, pdfauthor=author}
\newtheorem{theorem}{Theorem}
\newtheorem{definition}{Definition}

\newtheorem{remark}{Remark}
\newtheorem{lemma}{Lemma}

\title{Mean Field Game of High-Frequency Anticipatory Trading}
\author[1]{Xue Cheng \thanks{chengxue@pku.edu.cn}}
\author[1]{Meng Wang}
\author[1]{Ziyi Xu}
\affil[1]{\footnotesize LMEQF, Department of Financial Mathematics, \authorcr School of Mathematical Sciences, \authorcr
Peking University, Beijing 100871, China.}
\date{}
\begin{document}
\maketitle
\begin{abstract}
The interactions between a large population of high-frequency traders (HFTs) and a large trader (LT) who executes a certain amount of assets at discrete time points are studied. 
HFTs are faster in the sense that they trade continuously and predict the transactions of LT. A jump process is applied to model the transition of HFTs' attitudes towards inventories and the equilibrium is solved through the mean field game approach. When the crowd of HFTs is averse to running (ending) inventories, they first take then supply liquidity at each transaction of LT (throughout the whole execution period). Inventory-averse HFTs lower LT's costs if the market temporary impact is relatively large to the permanent one. What's more, the repeated liquidity consuming-supplying behavior of HFTs makes LT's optimal strategy close to uniform trading.
\end{abstract}




\newpage

\section{Introduction}
Institutional investors and high-frequency traders (HFTs) are active and important participants in today's financial markets. HFTs trade at a high speed and are equipped with advanced methods to process information; institutional investors split orders to finish the execution task with lower costs, but may leave footprints for HFTs to predict their future transactions. As found in Sa{\u{g}}lam (2020) \cite{sauglam2020order} and Hirschey (2021) \cite{hirschey2021high}, the interactions between institutional investors and HFTs are mainly completed through the channel of anticipatory trading. Actually, HFTs' trading pattern within the execution period of large orders, the optimal execution strategies of large traders and HFTs' influences on large traders' profits are highly concerned in both research and industry.

Motivated by this, we consider a market with a large trader (LT) and a large population of HFTs, which is approximated by a continuum of infinitesimal agents via mean field game approach. LT trades at discrete time points, while HFTs predict LT's orders and trade continuously, which shows the difference in both the trading speed and the population size. LT aims to maximize the revenue, while HFTs maximize the revenue together with the inventory costs. HFTs with different levels of inventory aversion (terminal aversion $\Gamma$ and running aversion $\phi$) are investigated and a jump process is applied to characterize the transition of these levels. That is to say, HFTs may not only be at different levels of inventory aversion, but this level may also change over time.

Two kinds of equilibria are investigated. If the strategy of LT is fixed, the game among HFTs is solved and the reached Nash equilibrium is called ``partial Nash equilibrium''; if LT also participates in the game, the game between LT and HFTs and the game among HFTs are both solved, the reached Nash equilibrium is called ``overall Nash equilibrium''.

In both equilibria, HFTs implement anticipatory trading towards LT's orders. When the crowd of HFTs is averse to ending position, they will play the role of Round-Tripper throughout the whole trading period. Round-Tripper is an anticipatory strategy defined in Xu and Cheng (2023) \cite{xu2023high}: after predicting other investors' future order flow, Round-Tripper trade in the same direction ahead of the coming order; when the order arrives, she trades in the opposite direction to supply liquidity back.
For example, if LT keeps buying, HFTs will buy in the first half of the time and sell in the second half of the time. When the crowd of HFTs is averse to running position, they will play the role of Round-Tripper at each transaction of LT. For example, if LT keeps buying, HFTs will buy before and sell after each purchase of LT. 

Consequently, the behavior of HFTs who hate both ending and running positions exhibits (1) short time-frames of establishing and liquidating positions; (2) ending the trading day in as close to a flat position as possible, which corresponed to the characteristics of HFTs mentioned in SEC (2010) \cite{sec2010}. Otherwise, if the crowd of HFTs is not sensitive to inventory, they will trade in the same direction as LT and accumulate positions.

It is interesting to find that between two transactions of LT, HFTs first trade in the opposite direction as the large order and then trade in the same direction as it, as empirically found in van Kervel and Menkveld (2019) \cite{van2019high}.

LT is benefited by inventory-averse HFTs if the temporary price impact is relatively large to the permanent one. It stands for a resilient market, most of the liquidity consumed by investors will soon recover. As a result, the majority of adverse impact caused by HFTs' same-direction trading will quickly disappear.
The liquidity provided by HFTs' opposite-direction trading dominates the remaining adverse impact and thus LT profits more.

In the overall Nash equilibrium, if the crowd of HFTs is more averse to running positions, LT's optimal strategy will be more close to uniform trading, which actually is her optimal strategy without HFTs. In this case, HFTs act as Round-Tripper and the repeated liquidity consuming-supplying behavior of HFTs smooths out LT's transactions. Therefore, the existence of anticipatory HFTs do not necessarily change the execution mode of LT. If the crowd of HFTs is more averse to terminal positions, LT will delay more shares to trade in the latter half of execution period, waiting for HFTs to offer liquidity.

The mean field game approach is applied to solve the curse of dimensionality. It is proved that both partial and overall Nash equilibrium have the $\varepsilon$-Nash property. In other words, the reached equilibria with infinitely many HFTs are good approximations to the equilibria with finitely many HFTs.

\section{Literature Review}
Mean field games (MFGs) were first introduced simultaneously and independently by Lasry and Lions (2007) \cite{lasry2007mean} and Huang, Malhamé and Caines (2006) \cite{huang2006large}. When the number of agents increases in the system, the problem of solving the Nash equilibrium becomes intractable, and thus we can focus on a representative player interacting with the distribution or average of the population. In the first MFG papers, the analytic approach was employed with two coupled partial differential equations: (1) the backward Hamilton-Jacobi-Bellman (HJB) equation takes care of the optimization part and represents the value function of the representative player; (2) the forward Kolmogorov-Fokker-Planck (KFP) equation represents the evolution of the population distribution. Later on, Carmona and Delarue (2013) \cite{carmona2013probabilistic} introduces the probabilistic approach for solving mean field games, which calls for the solution of systems of forward-backward stochastic differential equations (FBSDE) of Mckean-Vlasov type. 

The assumptions of identical and insignificant agents and symmetric interactions are required to focus on one representative agent, however, in reality there may exist heterogeneous players that affect the system disproportionally. For example, Huang (2010) \cite{huang2010large} studies linear-quadratic-Gaussian (LQG) games with a major player and a large number of minor players. Nourian and Caines (2013) \cite{nourian2013epsilon} studies a dynamic game with a major agent and a population of minor agents which are coupled via both their individual nonlinear stochastic dynamics and their individual finite time horizon nonlinear cost functions, the interactions between LT and HFTs are also achieved in these two ways. 

While MFGs are most often studied in settings with a continuous state space and deterministic or diffusive dynamics, there are also works on MFGs with finite state spaces, including Gomes, Mohr and Souza (2013) \cite{gomes2013continuous} and Belak, Hoffmann and Seifried (2021) \cite{belak2021continuous}. Moreover, the individual players may not only interact via their states but also via their controls, which is referred to as \textit{extended mean field games}, see Chapter I.4.6 in \cite{carmona2018probabilistic} as an example, and this is particularly relevant in finance and economics. This paper considers all the aspects mentioned above, with LT as the major agent and HFTs as the population of minor agents, where the states of HFTs include both the inventory level in the continuous space and the inventory averse level in the finite space, and they interacts through both the states (inventory levels) and the controls (trading speeds). To the best of our knowledge, this is the first paper considering both continuous-time and discrete-time games in the same framework, especially in the application of finance.

This paper contributes to literature on optimal liquidation problem in the mean field game framework. Two most related works are
Cardaliaguet and Lehalle (2018) \cite{cardaliaguet2018mean} and Jaimungal and Nourian (2019) \cite{huang2019mean}. In \cite{cardaliaguet2018mean}, traders execute orders in the presence of other similar market participants. HFTs in this paper are those ``similar participants'', the difference is that (1) there is also a large trader and HFTs will implement corresponding strategies when predicting the large trader's order; (2) the inventory aversion of HFTs in \cite{cardaliaguet2018mean} does not change over time. \cite{huang2019mean} considers the game between a major agent who is liquidating a large number of shares,
and a number of HFTs who detect and trade against the liquidator. The difference is that HFTs in this paper is faster than the large trader in the sense that the large trader can only trade at certain time points but HFTs can lay out continuously. However, in \cite{huang2019mean}, the major agent and HFTs both trade continuously and in a same time period. Besides, Casgrain and Jaimungal \cite{casgrain2020mean} (2020) solves a mean-filed game with heterogeneous investors, who optimize different trading objectives and have different beliefs about price dynamics.

This paper contributes to literature on interactions between HFT and large trader.
 To our knowledge, this paper is the first to model the anticipatory trading of infinitely many fast HFTs towards a slow large trader, with the transition of HFTs' inventory aversion.
Cont, Micheli and Neuman (2023) \cite{cont2022fast} models the Stackelberg game between a slow institutional investor and an HFT. 
HFT exploits price information more frequently and is subject to
periodic inventory constraints. The difference is that HFTs in this paper have higher trading speed than the large trader and  a 
 market with multiple HFTs is studied, solving the game in the mean field game approach. Li (2018) \cite{li2018high}, Yang and Zhu (2020) \cite{yang2020back}, Xu and Cheng (2023) \cite{xu2023high} and Xu and Cheng (2023) \cite{xu2023large} model the anticipatory trading in extended Kyle models and in discrete time. Ro{\c{s}}u (2019) \cite{rocsu2019fast} studies fast traders who use immediate information and slow traders who use lagged one. So the former's order flow predicts the latter's. 
 Empirical works include but is not limited to van Kervel ad Menkveld (2019) \cite{van2019high}, Korajczyk and Murphy (2019) \cite{korajczyk2019high}, Sa{\u{g}}lam (2020) \cite{sauglam2020order} and Hirschey (2021) \cite{hirschey2021high}.

\section{The Model}
We work on the filtered probability space $(\Omega, \mathcal{F},\{\mathcal{F}_t\}_{t\in[0,T]}, P)$. During a finite time horizon $[0,T],$ we consider a market with a large trader (LT) and a large population of high-frequency traders (HFTs) trading the same asset. We formulate the game with both finitely many HFTs and infinitely many HFTs, where in the latter, the mean field game approach is applied. The Nash equilibrium in the mean field game can be solved explicitly, and we demonstrate that it provides an approximate Nash equilibrium to the former one when the number of HFTs $M$ goes to infinity.

\textbf{Investors' trading activities.} LT has to liquidate the position $\xi_0\in\mathbbm{R}$ at some specific time points 
$$0=t_0< t_1<t_2<...<t_K<t_{K+1}=T.$$ 
The trading strategy of LT is denoted by $\{\xi_k\}_{k=1}^K$, where $\xi_k\in\mathbbm{R}$ represents the quantity traded by LT at $t_k.$

HFTs trade continuously. Their initial inventories are independently and identically distributed (i.i.d.) random variables $\{X_j(0)\}_{j\in\mathbbm{N}^+}$.
For the $j$-th HFT in the population, let $\mathcal{F}_t-$adapted processes $X_j(t)$ and $v_j(t)$ denote her position and trading speed at time $t\in[0,T]$,
\begin{equation*}
    dX_j(t)=v_j(t)dt.
\end{equation*}

\textbf{HFTs' inventory aversion.} Different HFTs may have different inventory-averse levels and these levels form a finite set $\mathbb{S}$ with $N$ elements. Each element $i\in\mathbb{S}$ corresponds to a pair of non-negative real numbers $(\Gamma(i), \phi(i))$ representing the terminal and running inventory aversion, respectively. For the $j$-th HFT, let $\mathcal{F}_t-$adapted process $Y_j(t)\in\mathbb{S}$ denote the level at time $t$. We assume that $\{Y_j(t)\}_{t\in[0,T]}^{j\in\mathbbm{N}^+}$ (1) are independent processes; (2) are identically distributed at $t=0$; (3) have same and constant transition rate matrix $\boldsymbol{Q}\in\mathbbm{R}^{N\times N}$. As a result, $\{Y_j(t)\}_{t\in[0,T]}^{j\in\mathbbm{N}^+}$ are identically distributed processes with
$$p_i(t)=P(Y_1(t)=i),\ i\in\mathbb{S},\ t\in[0,T].$$ 
In the following, let $Q^{ij}$ denote the element of $\boldsymbol{Q}$ at the $i$-th row and the $j$-th column. And we further assume that the choice of initial distribution $\{p_i(0)\}_{i\in\mathbb{S}}$ and transition rate matrix $\boldsymbol{Q}$ guarantee that $$p_i(t)>0,i\in\mathbb{S},t\in[0,T].$$

\textbf{Market information.} A standard one-dimensional Brownian motion $\{W(t)\}_{t\in[0,T]}$ drives the price dynamic.

\subsection{The game with finitely many HFTs}\label{finite_HFT_formulation}
Assume that there are $M\in\mathbbm{N}^+$ HFTs in the market and index $j\in\{1,...,M\}$.
Let $P^M(t)$ denote the fair price at time $t$, with permanent price impacts from LT's trading and HFTs' aggregate trading:
\begin{equation*}
dP^M(t)=\sigma dW(t)+\sum_{j=1}^M \gamma^{H,M}v_j(t)dt+\gamma\sum_{k=1}^K\xi_k 1_{[t=t_k]},
\end{equation*}
where $\gamma^{H,M}>0$ and $\gamma>0$ are the permanent price impact coefficients for HFTs and LT, respectively; $\sigma>0$ is the volatility of the asset.

As $M$ increases, the overall number of investors increases, so that each of them will have a milder impact on the market. And in order to get rid of explosive impact, we add a non-explosive condition $\gamma^{H,M}=O(\frac{1}{M})$. It is assumed that
\begin{equation*}
\gamma^{H,M}=\frac{1}{M}\gamma^H.
\end{equation*}
Define
\begin{equation*}
\overline{v}^M(t)=\frac{1}{M}\sum_{j=1}^M v_j(t),
\end{equation*}
then the fair price $P^M(t)$ follows:
\begin{equation*}
dP^M(t)=\sigma dW(t)+\gamma^H\overline{v}^M(t)dt+\gamma\sum_{k=1}^K\xi_k 1_{[t=t_k]}.
\end{equation*}

The transaction price also includes a temporary impact from other participants in the market and the trader's own trading fees. Let $\Hat{P}^M(t)$ and $\Tilde{P}^M_j(t)$ denote the actual transaction price of LT and the $j$-th HFT respectively,
\begin{equation*}
\begin{aligned}
&\Hat{P}^M(t_k) = P^M(t_k) + \lambda^H\overline{v}^M(t_k) + \lambda \xi_k + \eta_0 \xi_k,\ k=1,...,K;\\
&\Tilde{P}_j^M(t)=P^M(t)+\lambda^H\overline{v}^M(t)+\lambda\sum_{k=1}^K\xi_k 1_{[t=t_k]}+\eta v_j(t),
\end{aligned}
\end{equation*}
where $\lambda^H>0$ and $\lambda>0$ are the temporary price impact coefficients for HFTs and LT, $\eta_0>0$ and $\eta>0$ represent the trading fees.

LT maximizes her expected revenue of liquidation, with an objective function
\begin{equation}
\label{LTobj-finite}
\begin{aligned}
J_{LT}^M(\xi)&=\mathbbm{E}\left[\sum_{k=1}^K (-\xi_k)\Hat{P}^M(t_k) \right]\\
&=\mathbbm{E}\left\{\sum_{k=1}^K(-\xi_k)\left[P^M(0)+\sigma W(t_k) + \gamma\sum_{j=1}^k\xi_j+\gamma^H \int_0^{t_k} \overline{v}^M(t) dt +\lambda^H\overline{v}^M(t_k)+(\lambda+\eta_0)\xi_k \right] \right\}.
\end{aligned}
\end{equation}

The LT's trading strategy is supposed to be deterministic, since LT is slow not only in terms of the trading speed but also the frequency of exploiting the market information. The set of admissible controls for LT is
\begin{equation}
\label{LT-admissible-finite}
\mathcal{A}_{LT}^M=\left\{\{\xi_k\}_{k=1}^K: \xi_k\in\mathbbm{R} \text{ is deterministic and } \xi_0+\sum_{k=1}^K\xi_k=0\right\}.
\end{equation} 

HFT is concerned with the revenue and the inventory risk as well. For simplicity, we omit the subscript $j$ below. For each HFT, the objective function is
\begin{equation}
\label{HFTobj-finite}
J_{HFT}^M (v)=\mathbbm{E}\left[\left. X(T)P^M(T)-\int_0^{T} v(t)\Tilde{P}^M(t) dt-\Gamma(Y(T)) X(T)^2-\int_0^T\phi(Y(t)) X(t)^2dt \right\vert \mathcal{F}_0 \right].
\end{equation}
The first term measures the marked-to-market value of HFT's final position; the second term is the total trading revenue; the third and last terms refer to the  penalties for final and running positions, where $\Gamma(\cdot)\geq0$ and $\phi(\cdot)\geq0$ are functions of $Y(T)$ and $Y(t)$, respectively. 

The set of admissible controls for HFT is 
\begin{equation}
\label{HFT-admissible-finite}
\mathcal{A}_{HFT}^M=\left\{v: \mathcal{F}_t\text {-progressively measurable, c\`{a}dl\`{a}g and }  \mathbbm{E} \int_0^T v(t)^2 d t<\infty\right\}.
\end{equation}
The left-limit property is enabled by the condition of square integrability; the right-continuous property reflects HFT's timely reaction to LT's order and the jump of inventory aversion. 

\subsection{The mean field game}
Due to the fact that different HFTs are coupled through the price dynamic, the solution becomes intractable when the number of HFTs becomes large. To address the curse of dimensionality, we adopt the idea from mean field game (MFG) theory, where HFTs' 
``aggregate'' trading activities affect the market through both permanent and temporary price impact, while each individual HFT only affects her own transaction price through trading fees.

We define $\{E_i(t)\}_{t\in[0,T]}^{i\in\mathbb{S}}$ and $\{\mu_i(t)\}_{t\in[0,T]}^{i\in\mathbb{S}}$ as deterministic processes representing the average inventory and the average trading speed of the HFT population at inventory-averse level $i\in\mathbb{S}$. Additionally, we define the mean field of the whole HFT population as
\begin{equation*}
\begin{aligned}
E(t) = \sum_{i\in\mathbb{S}}p_i(t)E_i(t),\\
\mu(t) = \sum_{i\in\mathbb{S}}p_i(t)\mu_i(t).
\end{aligned}
\end{equation*}

Inheriting the definitions of parameters from Section \ref{finite_HFT_formulation}, let $P(t)$ denote the fair price at time $t$, the permanent price impact is modeled as below:
\begin{equation}
\label{MFGprice}
dP(t)=\sigma dW(t)+\gamma^H\mu(t)dt+\gamma\sum_{k=1}^K\xi_k 1_{[t=t_k]},
\end{equation}

By adding the temporary impact and a trader's own trading fees, let $\Hat{P}(t)$ and $\Tilde{P}(t)$ denote the actual transaction price of LT and an HFT respectively. By omitting the subscript $j$, we denote the inventory, trading speed, and inventory-averse level of the HFT by $X(t)$, $v(t)$, $Y(t)$, respectively. Then
\begin{equation*}
\begin{aligned}
&\Hat{P}(t_k) = P(t_k)  + \lambda^H\mu(t_k)+ \lambda \xi_k + \eta_0 \xi_k,\ k=1,...,K;\\
&\Tilde{P}(t)=P(t)+\lambda^H\mu(t)+\lambda\sum_{k=1}^K\xi_k 1_{[t=t_k]}+\eta v(t).
\end{aligned}
\end{equation*}

The LT's objective function is
\begin{equation}
\label{LTobj}
\begin{aligned}
J_{LT}(\xi)&=\mathbbm{E}\left[\sum_{k=1}^K (-\xi_k)\Hat{P}(t_k) \right]\\
&=\mathbbm{E}\left\{\sum_{k=1}^K(-\xi_k)\left[P(0)+\sigma W(t_k) + \gamma\sum_{j=1}^k\xi_j+\gamma^H \int_0^{t_k} \mu(t) dt +\lambda^H\mu(t_k)+(\lambda+\eta_0)\xi_k \right] \right\},
\end{aligned}
\end{equation}
with the set of admissible controls
\begin{equation}
\label{LT-admissible}
\mathcal{A}_{LT}=\left\{\{\xi_k\}_{k=1}^K: \xi_k\in\mathbbm{R} \text{ is deterministic and }\xi_0+\sum_{k=1}^K\xi_k=0\right\}.
\end{equation}

The objective function of HFT is
\begin{equation}
\label{HFTobj}
J_{HFT} (v)=\mathbbm{E}\left[\left. X(T)P(T)-\int_0^{T} v(t)\Tilde{P}(t) dt-\Gamma(Y(T)) X(T)^2-\int_0^T\phi(Y(t)) X(t)^2dt \right\vert \mathcal{F}_0\right],
\end{equation}
with the set of admissible controls
\begin{equation}
\label{HFT-admissible}
\mathcal{A}_{HFT}=\left\{v: \mathcal{F}_t\text {-progressively measurable, c\`{a}dl\`{a}g and }  \mathbbm{E} \int_0^T v(t)^2 d t<\infty\right\}.
\end{equation}
As a result, the average trading speeds
$\{\mu_i(t)\}_{t\in[0,T]}^{i\in\mathbb{S}}$ are also c\`{a}dl\`{a}g.

As we will see in Section \ref{sec_partial_equilibrium}, given $\xi$ and $\mu$, the optimal control $v$ is in a Markovian feedback sense, i.e.,
\begin{equation*}
v(t)=v(t,X(t),Y(t);\xi,\mu),
\end{equation*}
where $v(\cdot,\cdot,\cdot;\xi,\mu)$ is a deterministic function on $t$, $X(t)$ and $Y(t)$.


Two kinds of mean field Nash equilibria are defined.
\begin{definition}[Mean field partial Nash equilibrium] Given LT's strategy $\{\xi_k\}_{k=1}^K$, we call $(v^*,\{E_i^*\}_{i\in\mathbb{S}},\{\mu_i^*\}_{i\in\mathbb{S}})$ a partial Nash equilibrium if the following conditions are satisfied:\\
(i) HFT's optimizing condition: given $\mu^*$, $v^*$ maximizes \eqref{HFTobj} over $\mathcal{A}_{HFT}$,
\begin{equation*}
J_{HFT} (v^*;\xi,\mu^*) =\max_{v\in \mathcal{A}_{HFT}} J_{HFT} (v;\xi,\mu^*);
\end{equation*}
(ii) Fixed-point condition: HFT's optimal inventory process $\{X^*(t)\}_{t\in[0,T]}$ and optimal control $v^*(t)=v^*(t,X^*(t),Y(t);\xi,\mu^*)$ actually form the mean field $\{E_i^*\}_{i\in\mathbb{S}}$ and $\{\mu_i^*\}_{i\in\mathbb{S}}$,
\begin{equation*}
\begin{aligned}
\mathbbm{E}\left[X^*(t) \vert Y(t)=i  \right] &= E_i^*(t),\\
\mathbbm{E}\left[v^*(t,X^*(t),Y(t);\xi,\mu^*)\vert Y(t)=i  \right] &= \mu_i^*(t).
\end{aligned}
\end{equation*}
\end{definition}

\begin{definition}[Mean field overall Nash equilibrium] We call $(\{\xi_k^*\}_{k=1}^K, v^*,\{E_i^*\}_{i\in\mathbb{S}},\{\mu_i^*\}_{i\in\mathbb{S}})$ an overall Nash equilibrium if:\\
(i) given $\{\mu_i^*\}_{i\in\mathbb{S}}$, $\{\xi_k^*\}_{k=1}^K$ maximizes \eqref{LTobj} over $\mathcal{A}_{LT}$,
\begin{equation*}
J_{LT}(\xi^*;\mu^*) = \max_{\xi\in \mathcal{A}_{LT}}J_{LT}(\xi;\mu^*);
\end{equation*}
(ii) given $\{\xi_k^*\}_{k=1}^K$, mean field partial Nash equilibrium is reached.
\end{definition}

\begin{remark}
The assumption that 
$\{E_i^*\}_{i\in\mathbb{S}}$ and 
$\{\mu_i^*\}_{i\in\mathbb{S}}$ are deterministic is common in mean field game literature, as in Cardaliaguet and Lehalle (2018) \cite{cardaliaguet2018mean}, Lacker and Zariphopoulou (2019) \cite{lacker2019mean} and Lacker and Soret (2020) \cite{lacker2020many}.    
\end{remark}

It will be demonstrated that the solution to the mean field game is a good approximation to the game with finitely many HFTs, and we define the approximate Nash equilibria as below:
\begin{definition}[$M$-HFT partial $\varepsilon$-Nash equilibrium] Fix $\varepsilon>0$. Given LT's strategy $\{\xi_k\}_{k=1}^K$, we call $(v_1^*,...,v_M^*)$ a partial $\varepsilon$-Nash equilibrium if for each $j\in\{1,...,M\}$, given $v_{-j}^*=\{v_k^*\}_{k=1,k\neq j}^M$,
\begin{equation*}
J_{HFT}^M (v_j^*;\xi,v_{-j}^*) > J_{HFT}^M (v_j;\xi,v_{-j}^*)-\varepsilon,\ \text{a.s.},\ \forall v_j \in \mathcal{A}_{HFT}^M.
\end{equation*}
\end{definition}

\begin{definition}[$M$-HFT overall $\varepsilon$-Nash equilibrium] Fix $\varepsilon>0$. We call $(\{\xi_k^*\}_{k=1}^K, v_1^*,...,v_M^*)$ an overall $\varepsilon$-Nash equilibrium if:\\
(i) given $(v_1^*,...,v_M^*)$,
\begin{equation*}
J_{LT}^M(\xi^*;v_1^*,...,v_M^*) > J_{LT}^M(\xi;v_1^*,...,v_M^*) - \varepsilon,\, \forall \xi \in \mathcal{A}_{LT}^M;
\end{equation*}
(ii) given $\{\xi_k^*\}_{k=1}^K$, $M$-HFT partial $\varepsilon$-Nash equilibrium is reached.
\end{definition}

It will be shown that for any $\varepsilon>0$, when $M$ is large enough, mean field partial and overall equilibrium satisfy $M$-HFT partial and overall $\varepsilon$-Nash equilibrium, respectively.

\section{Equilibrium}
We solve the mean field game first and show the $\varepsilon$-Nash equilibrium property.

\subsection{Partial Nash equilibrium}\label{sec_partial_equilibrium}
Define the value function $H$ as
\begin{equation*}
 \begin{aligned}
H(t,x,P,i)=\max_{v\in\mathcal{A}_{HFT}}\mathbbm{E}\Bigg[ &X(T)P(T)-\int_t^{T} v(s)\Tilde{P}(s) ds-\Gamma(Y(T)) X(T)^2\\
&-\int_t^T\phi(Y(s))X(s)^2ds\Bigg\vert X(t)=x,P(t)=P,Y(t)=i\Bigg].
 \end{aligned}   
\end{equation*}
By It\^{o}'s lemma and price dynamic \eqref{MFGprice},
\begin{equation*}
\begin{aligned}
H(t,x,P,i)=Px+\max_{v\in\mathcal{A}_{HFT}}\mathbbm{E}\Bigg[&\gamma\sum_{k=1}^K\xi_kX(t_k)1_{[t<t_k]}+\int_t^{T}\gamma^H\mu(s) X(s)ds-\int_t^Tv(s)(\lambda^H\mu(s)+\eta v(s))ds\\
 &-\Gamma(Y(T)) X(T)^2-\int_t^T\phi(Y(s))X(s)^2ds\Bigg\vert X(t)=x,Y(t)=i\Bigg].\\
\end{aligned}
\end{equation*}
Define
\begin{equation}
\label{def-h}
\begin{aligned}
h(t,x,i)=\max_{v\in\mathcal{A}_{HFT}}\mathbbm{E}&\Bigg[\gamma\sum_{k=1}^K\xi_kX(t_k)1_{[t<t_k]}+\int_t^{T}\gamma^H\mu(s) X(s)ds-\int_t^Tv(s)(\lambda^H\mu(s)+\eta v(s))ds\\
&-\Gamma(Y(T)) X(T)^2-\int_t^T\phi(Y(s))X(s)^2ds\Bigg\vert X(t)=x,Y(t)=i\Bigg],
\end{aligned}
\end{equation}
we have $H(t,x,P,i)=Px+h(t,x,i).$

Compared with the traditional continuous-time optimization, the existence of LT introduces ``discontinuities'' in the price dynamics. Special handling is required at times $\{t_k\}_{k=1}^K$ when LT engages in trading.

\textbf{HFT's individual optimization.} Within each time interval $t\in[t_k,t_{k+1}),k=0,1,...,K,$ assume $Y(t)=i\in\mathbb{S},$
\begin{equation*}
\begin{aligned}
h(t,X(t),i)\approx \max_{v}\Big \{&\gamma^H\mu(t)X(t)\Delta t-[\lambda^H\mu(t)+\eta v]v\Delta t-\phi(i)X(t)^2\Delta t \\
& +\mathbbm{E}[h(t+\Delta t,X(t+\Delta t),Y(t+\Delta t))|\mathcal{F}_t] \Big\}\\
\approx \max_{v}\Big\{&\gamma^H\mu(t)X(t)\Delta t-[\lambda^H\mu(t)+\eta v]v\Delta t-\phi(i)X(t)^2\Delta t \\
&+h(t,X(t),i)+\frac{\partial h}{\partial
 t}\Delta t+\frac{\partial h}{\partial x}v\Delta t+\sum_{j}Q^{ij}h(t,X(t),j)\Delta t \Big\}.\\
\end{aligned}
\end{equation*}
It implies that
\begin{equation*}
\frac{\partial h(t,x,i)}{\partial t}+\gamma^H\mu(t) x-\phi(i)x^2+\sum_j Q^{ij}h(t,x,j)+\max_v\left\{-\eta v^2+v\left[\frac{\partial h(t,x,i)}{\partial x}-\lambda^H\mu(t)\right]\right\}=0.
\end{equation*}
According to the definition of $h$ in \eqref{def-h}, we additionally have following boundary conditions:
\begin{equation*}
\begin{cases}
&h(t_{k-},x,i)=h(t_k,x,i)+\gamma\xi_kx,\ k=1,...,K,\\
&h(T,x,i)=-\Gamma(i)x^2.    
\end{cases}
\end{equation*}
Denote $h_i(t,x)=h(t,x,i),$ for each each $i\in\mathbb{S}$, the complete \textbf{HJB equation} follows:
\begin{equation}
\label{HJB}
\begin{cases}
&\frac{\partial h_i}{\partial t}+\gamma^H\mu(t) x-\phi(i)x^2+\sum_j Q^{ij}h_j+\frac{1}{4\eta}\left[\frac{\partial h_i}{\partial x}-\lambda^H\mu(t)\right]^2=0,\ t\in[t_k,t_{k+1}),\ k=1,...,K,\\
&h_i(t_{k-},x)=h_i(t_k,x)+\gamma\xi_kx,\ k=1,...,K,\\
&h_i(T,x)=-\Gamma(i)x^2.    
\end{cases}
\end{equation}
The optimal feedback control follows:
\begin{equation*}
 v(t,x,i)=\frac{1}{2\eta}\left[\frac{\partial h_i}{\partial x}-\lambda^H\mu(t)\right],\ t\in[0,T].   
\end{equation*}

Assume that $h_i(t,x)=h_i^0(t)+h_i^1(t)X+h_i^2(t)x^2,$ substitute it into \eqref{HJB}, the Riccati equations of $\{h_i^0,h_i^1,h_i^2\}_{i\in\mathbb{S}}$ follow
\begin{equation}
\label{hs}
\begin{cases}
&\frac{dh_i^2}{dt}+\frac{(h_i^2)^2}{\eta}-\phi(i)+\sum_jQ^{ij}h_j^2=0,\ \\
&\frac{dh_i^1}{dt}+\frac{h_i^2(h_i^1-\lambda^H\mu(t))}{\eta}+\gamma^H\mu(t)+\sum_jQ^{ij}h_j^1=0,\\
&\frac{dh_i^0}{dt}+\frac{(h_i^1-\lambda^H\mu(t))^2}{4\eta}+\sum_jQ^{ij}h_j^0=0,\
\end{cases}
t\in[t_k,t_{k+1}),\ k=1,...,K,
\end{equation}
with boundary conditions
\begin{equation}
\label{hs-boundary}
\begin{cases}
& h_i^2(t_{k-})=h_i^2(t_k),h_i^1(t_{k-})=h_i^1(t_k)+\gamma\xi_k,h_i^0(t_{k-})=h_i^0(t_k),\ k=1,...,K,\\
&h_i^2(T)=-\Gamma(i),h_i^1(T)=h_i^0(T)=0. \\
\end{cases}
\end{equation}
The optimal feedback control follows:
\begin{equation}
\label{optcontrol2}
 v(t,x,i)=\frac{1}{2\eta}\left[h_i^1+2h_i^2x-\lambda^H\mu(t)\right], t\in[0,T]. 
\end{equation}

In Theorem \ref{thm-verification}, the verification theorem is proved.
\begin{theorem}
\label{thm-verification}
The solution of \eqref{HJB} satisfies
\begin{equation*}
\begin{aligned}
h_i(t,x) = \max_{v\in\mathcal{A}_{HFT}}\mathbbm{E}&\Bigg[\gamma\sum_{k=1}^K\xi_kX(t_k)1_{[t<t_k]}+\int_t^{T}\gamma^H\mu(s) X(s)ds-\int_t^Tv(s)(\lambda^H\mu(s)+\eta v(s))ds\\
&-\Gamma(Y(T)) X(T)^2-\int_t^T\phi(Y(s))X(s)^2ds\Bigg\vert X(t)=x,Y(t)=i\Bigg].
\end{aligned}
\end{equation*}
Consequently, $H(t,x,P,i)=Px+h_i(t,x)$ is actually the value function that satisfies
\begin{equation*}
\begin{aligned}
H(t,x,P,i)=\max_{v\in\mathcal{A}_{HFT}}\mathbbm{E}\Bigg[ &X(T)P(T)-\int_t^{T} v(s)\Tilde{P}(s) ds-\Gamma(Y(T)) X(T)^2\\
&-\int_t^T\phi(Y(s))X(s)^2ds\Bigg\vert X(t)=x,P(t)=P,Y(t)=i\Bigg].
\end{aligned}
\end{equation*}
\end{theorem}
\begin{proof}
See \ref{proof-thm-verification} of \nameref{appendix}.
\end{proof}

Solving the coupled ODEs in \eqref{hs} is rather complicated. Instead of studying the behavior of a single HFT, we turn to the evolution of the average position $\{E_i\}_{i\in\mathbb{S}}$ and $E$ of the HFT population, which does not require explicit solutions of \eqref{hs}.
To this end, we next consider the forward equation for the distribution of $Y$ and the distribution of $X$ conditioned on $Y$.

\textbf{Forward equation.} The Kolmogorov forward equation for the jump process $Y$ during $[0,T]$ is
\begin{equation*}
\begin{cases}
&\frac{d}{dt}p_i(t)=\sum_j p_j(t)Q^{ji},\\
&p_i(0)=P(Y(0)=i).
\end{cases}
\end{equation*}

For HFT's inventory, denote the conditional density of $X(t)$ given $Y(t)=i$ by $m_t^i(x)$. The following lemma indicates the evolution of $m_t^i(x)$.
\begin{lemma}\label{lemma-forward-eqn}
The conditional density $\{m_t^i(x)\}_{i\in\mathbb{S}}$ satisfies the following forward equation:
\begin{equation*}
\sum_jp_j(t) Q^{ji}m_t^i(x)+p_i(t)\frac{\partial}{\partial t}m_t^i(x)=-p_i(t)\frac{\partial}{\partial x}[v(t, x,i)m_t^i(x)]+\sum_j p_j(t) Q^{ji} m_t^j(x).
\end{equation*}
\end{lemma}

\begin{proof}
See \ref{proof-lemma-forward-eqn} of \nameref{appendix}.
\end{proof}

\textbf{Mean field equation.} 
In this part, we first derive the evolution of average positions $\{E_i\}_{i\in\mathbb{S}}$ of different HFT populations, then calculate the average position $E$ of the whole HFT population. 

Given conditional densities $\{m_t^i(x)\}_{i\in\mathbb{S}},$ for each $i\in\mathbb{S}$:
\begin{equation*}
E_i(t)=\int_{\mathbbm{R}} x m_t^i(x)dx,\ \mu_i(t)=\int_{\mathbbm{R}} v(t,x,i) m_t^i(x)dx.
\end{equation*}
We denote
\begin{equation*}
\boldsymbol{E}=
\begin{bmatrix}
E_1\\
E_2\\
\vdots\\
E_N
\end{bmatrix},\ 
\boldsymbol{\mu}=
\begin{bmatrix}
\mu_1\\
\mu_2\\
\vdots\\
\mu_N
\end{bmatrix}.
\end{equation*}
For simplicity, the following notations will be used in the paragraph below:
\begin{equation}
\label{notation1}
\begin{aligned}
\boldsymbol{e}=
\begin{bmatrix}
1\\
\vdots\\
1
\end{bmatrix},\ 
\boldsymbol{p(t)}=
\begin{bmatrix}
p_1(t)\\
\vdots\\
p_N(t)
\end{bmatrix},\ 
\boldsymbol{H(t)}=
\begin{bmatrix}
h_1^2(t)& &\\
&\ddots&\\
& &h_N^2(t)
\end{bmatrix},\ 
\boldsymbol{Q}=
\begin{bmatrix}
Q^{11}&\cdots &Q^{1N}\\
&\ddots&\\
Q^{N1}&\cdots &Q^{NN}
\end{bmatrix},\ 
\end{aligned}
\end{equation}

\begin{equation}
\label{notation2}
\begin{aligned}
\boldsymbol{\Gamma}=
\begin{bmatrix}
\Gamma(1)& &\\
&\ddots&\\
& &\Gamma(N)
\end{bmatrix},\ 
\boldsymbol{\Phi(t)}=
\begin{bmatrix}
\phi(1)-\sum_j Q^{1j}h_j^2& &\\
&\ddots&\\
& &\phi(N)-\sum_j Q^{Nj}h_j^2
\end{bmatrix},
\end{aligned}
\end{equation}

\begin{equation}
\label{notation3}
\boldsymbol{p_Q(t)}=\begin{bmatrix}
\frac{1}{p_1(t)}& &\\
&\ddots&\\
& &\frac{1}{p_N(t)}
\end{bmatrix}\boldsymbol{Q^\top}\begin{bmatrix}
p_1(t)& &\\
&\ddots&\\
& &p_N(t)
\end{bmatrix}
-\text{diag}\left\{
\begin{bmatrix}
\frac{1}{p_1(t)}& &\\
&\ddots&\\
& &\frac{1}{p_N(t)}
\end{bmatrix}\boldsymbol{Q^\top}\begin{bmatrix}
p_1(t)\\
\vdots\\
p_N(t)
\end{bmatrix}
\right\}.
\end{equation}

Lemma \ref{lemma-dE} gives the relationship between the derivative of the average inventory $\boldsymbol{E}$ of the HFT population at inventory aversion $i$ and their average trading speed $\boldsymbol{\mu}$, which is a direct result of Lemma \ref{lemma-forward-eqn}.
\begin{lemma}
\label{lemma-dE}
\begin{equation}
\label{dE}
\begin{aligned}
&\boldsymbol{\mu(t)}=\boldsymbol{\dot{E}(t)}-\boldsymbol{p_Q(t)}\boldsymbol{E(t)}.
\end{aligned}
\end{equation}
\end{lemma}
\begin{proof}
See \ref{proof-lemma-dE} of \nameref{appendix}.
\end{proof}

Dot multiply each side of \eqref{dE} by $\boldsymbol{p(t)}$, we have
\begin{equation*}
    \mu(t)=\dot{E}(t),
\end{equation*}
which implies that the derivative of average inventory $E$ of the whole HFT population is their average trading speed $\mu.$ However, \eqref{dE} implies that derivative of $E_i$ is not $\mu_i$. It is because the change of $E_i$ is from not only the trading activity, but also the change of population composition with inventory aversion $i$.

Before we give the whole evolution of the average inventory $\boldsymbol{E}$, the following boundary conditions at $\{t_k\}_{k=1}^K$ are studied, which come from discontinuities of $h_i^1$. It offers some intuition about HFTs' anticipatory trading towards LT. 
\begin{lemma}
\label{lemma-Eboundary}
For $k=1,...,K,$
\begin{equation*}
\begin{aligned}
&\boldsymbol{E(t_{k-})}=\boldsymbol{E(t_k)},\ \boldsymbol{\mu(t_{k-})}=\boldsymbol{\mu(t_k)}+\frac{\gamma\xi_k}{\lambda^H+2\eta}\boldsymbol{e},\\
&E(t_k)=E(t_{k-}),\ \mu(t_{k-})=\mu(t_{k})+\frac{\gamma\xi_k}{\lambda^H+2\eta}.    
\end{aligned}
\end{equation*}
\end{lemma}
\begin{proof}
See \ref{proof-lemma-Eboundary} of \nameref{appendix}.
\end{proof}
HFTs' anticipation towards the large order brings the jump in their trading speed. For example, if LT is going to buy, i.e., $\xi_k>0,$ HFTs will decrease the trading speed to exploit the market impact caused by LT's purchase.

Combining the Riccati equations \eqref{hs} and Lemma \ref{lemma-dE}, the evolution of $\boldsymbol{E}$ within each time interval can be derived, together with the boundary conditions in Lemma \ref{lemma-Eboundary}, the whole equilibrium system is displayed in Theorem \ref{thm-E}.
\begin{theorem}
The equilibrium system of $\boldsymbol{E}$ is:
\label{thm-E}
{\footnotesize\begin{equation}\label{mean_field_equation}
\begin{cases}
&[2\eta\boldsymbol{I}+\lambda^H\boldsymbol{e}\boldsymbol{p}^\top]\boldsymbol{\dot{\mu}(t)} +[\gamma^H\boldsymbol{e}\boldsymbol{p}^\top+2\eta\boldsymbol{Q}+\lambda^H\boldsymbol{e}\boldsymbol{p}^\top\boldsymbol{Q}]\boldsymbol{\mu(t)}
-2[\boldsymbol{H(t)}\boldsymbol{p_Q(t)}+\boldsymbol{\Phi(t)}+\boldsymbol{Q}\boldsymbol{H(t)}]\boldsymbol{E(t)}=\boldsymbol{0},\ t\in[t_k,t_{k+1}),k=0,...,K\\
&\boldsymbol{\dot{E}(t)}=\boldsymbol{\mu(t)}+\boldsymbol{p_Q(t)}\boldsymbol{E(t)},\ t\in[t_k,t_{k+1}),k=0,...,K\\
&\boldsymbol{E(t_{k-})}=\boldsymbol{E(t_k)},\ \boldsymbol{\mu(t_{k-})}=\boldsymbol{\mu(t_k)}+\frac{\gamma\xi_k}{\lambda^H+2\eta}\boldsymbol{e},\ k=1,...,K\\
&[2\eta\boldsymbol{I}+\lambda^H\boldsymbol{ep}^\top]\boldsymbol{\mu(T)}+2\boldsymbol{\Gamma}\boldsymbol{E(T)}=\boldsymbol{0},\\
&\boldsymbol{E(0)}=\boldsymbol{E}_0.
\end{cases}
\end{equation}}
\end{theorem}
\begin{proof}
    See \ref{proof-thm-E} of \nameref{appendix}.
\end{proof}
Once $\boldsymbol{E}$ is solved, the average position of the total HFT population $E=\boldsymbol{p}^\top\boldsymbol{E}$ is easily calculated.

\subsubsection{The special case of \texorpdfstring{$N=1$}{}}
When $|\mathbb{S}|=1$, i.e., all HFTs have same inventory aversion, system \eqref{mean_field_equation} becomes:
\begin{equation}
\label{mean_field_equation_nojump}
\begin{cases}
&(\lambda^H+2\eta)\ddot{E}(t)+\gamma^H \dot{E}(t)-2\phi E(t)=0,\ t\in[t_k,t_{k+1}),\ k=0,...,K.\\
&E(t_{k-})=E(t_{k}),\dot{E}(t_{k-})=\dot{E}(t_{k})+\frac{\gamma\xi_k}{\lambda^H+2\eta},\ k=1,...,K.\\
&(\lambda^H+2\eta)\dot{E}(T)+2\Gamma E(T)=0,\\
&E(0)=E_0.\\
\end{cases}
\end{equation}
This system can be solved explicitly,
$$
E(t)=A_k e^{\theta_1 t}+B_ke^{\theta_2 t},t\in[t_k,t_{k+1}),\ k=0,...,K,   
$$
where
\begin{equation*}
\label{LTnogamenojumpHFT}
\begin{cases}
&\theta_1=\frac{-\gamma^H+\sqrt{(\gamma^{H})^2+8\phi(\lambda^H+2\eta)}}{2(\lambda^H+2\eta)},\ \theta_2=\frac{-\gamma^H-\sqrt{(\gamma^{H})^2+8\phi(\lambda^H+2\eta)}}{2(\lambda^H+2\eta)},\\
&x_k=\frac{1}{\theta_2-\theta_1}\frac{\gamma}{\lambda^H+2\eta}\sum_{j=1}^k\xi_je^{-\theta_1t_j},\ y_k=\frac{1}{\theta_2-\theta_1}\frac{\gamma}{\lambda^H+2\eta}\sum_{j=1}^k\xi_je^{-\theta_2t_j},\ k=1,...,K.\\
&A_0=\frac{-(\lambda^H+2\eta)[x_K\theta_1e^{\theta_1 T}+(E_0-y_K)\theta_2e^{\theta_2 T}]-2\Gamma[x_Ke^{\theta_1 T}+(E_0-y_K)e^{\theta_2 T}]}{(\lambda^H+2\eta)(\theta_1e^{\theta_1 T}-\theta_2e^{\theta_2 T})+2\Gamma(e^{\theta_1 T}-e^{\theta_2 T})},\\
&B_0=E_0-A_0,\\
&A_k=A_0+x_k,\ B_k=B_0-y_k, k=1,...,K.\\
\end{cases}
\end{equation*}

\subsubsection{The solving scheme of general \texorpdfstring{$N$}{}}
\label{subsec-solvingscheme}
\begin{itemize}
    \item[Step 1] Given $\boldsymbol{Q}$, solve $\boldsymbol{p}$:
    \begin{equation*}
    \frac{d}{dt}\boldsymbol{p(t)}^\top = \boldsymbol{p(t)}^\top\boldsymbol{Q}.
    \end{equation*}
    For example, if $\boldsymbol{Q}=\begin{bmatrix}
        -x&x\\
        y&-y
    \end{bmatrix},x^2+y^2\neq0$ and $\boldsymbol{p(0)}=\begin{bmatrix}
        1/2\\
        1/2
    \end{bmatrix},$
    \begin{equation*}
    p_1(t)=\frac{y+\frac{x-y}{2}e^{-(x+y)t}}{x+y},\ p_2(t)=\frac{x-\frac{x-y}{2}e^{-(x+y)t}}{x+y}.
    \end{equation*}
    $\boldsymbol{p_Q}$ in \eqref{notation3} is also solved.
\item[Step 2] Solve $\{h_i^2\}_{i=1}^N:$
$$\frac{dh_i^2}{dt}+\frac{(h_i^2)^2}{\eta}-\phi(i)+\sum_jQ^{ij}h_j^2=0,\ h_i^2(T)=-\Gamma(i).$$
This can be numerically solved by Euler method, then we get $\boldsymbol{H}$ in \eqref{notation1} and $\boldsymbol{\Phi}$ in \eqref{notation2}.

\item[Step 3] Solve $\boldsymbol{E}$:
note that $$|2\eta\boldsymbol{I}+\lambda^H \boldsymbol{ep}^{\top}|=(2\eta)^{N-1}(\lambda^H+2\eta)>0,$$ so $2\eta\boldsymbol{I}+\lambda^H \boldsymbol{ep}^{\top}$ is invertible, and the mean field equilibrium system \eqref{mean_field_equation} can be written as
\begin{equation*}
d\begin{bmatrix}
\boldsymbol{\mu(t)}\\
\boldsymbol{E(t)}
\end{bmatrix}=
\boldsymbol{A(t)}\begin{bmatrix}
\boldsymbol{\mu(t)}\\
\boldsymbol{E(t)}
\end{bmatrix} dt
-\frac{\gamma}{\lambda^H+2\eta} \sum_{k=1}^K \xi_k 1_{[t=t_k]} \begin{bmatrix}
\boldsymbol{e}\\
\boldsymbol{0}
\end{bmatrix},
\end{equation*}
where
\begin{equation}\label{mean_field_equation1}
\begin{aligned}
&\boldsymbol{A(t)}=\begin{bmatrix}
\boldsymbol{A_1(t)} & \boldsymbol{A_2(t)} \\
\boldsymbol{I} & \boldsymbol{p_Q(t)}
\end{bmatrix},\\
&\boldsymbol{A_1(t)} = -[2\eta\boldsymbol{I}+\lambda^H\boldsymbol{e}\boldsymbol{p}^\top]^{-1} [\gamma^H\boldsymbol{e}\boldsymbol{p}^\top+2\eta\boldsymbol{Q}+\lambda^H\boldsymbol{e}\boldsymbol{p}^\top\boldsymbol{Q}],\\
&\boldsymbol{A_2(t)} = 2[2\eta\boldsymbol{I}+\lambda^H\boldsymbol{e}\boldsymbol{p}^\top]^{-1} [\boldsymbol{H(t)}\boldsymbol{p_Q(t)}+\boldsymbol{\Phi(t)}+\boldsymbol{Q}\boldsymbol{H(t)}].
\end{aligned}
\end{equation}
We first use Euler's method to find the fundamental solution matrix $\boldsymbol{U(t)}$ of 
\begin{equation*}
\begin{cases}
&d\boldsymbol{U(t)}=\boldsymbol{A(t)}\boldsymbol{U(t)}dt,\\
&\boldsymbol{U(0)}=\boldsymbol{I}_{2N}.    
\end{cases}
\end{equation*}
Next, we let
\begin{equation*}
\begin{bmatrix}
\boldsymbol{\mu(t)}\\
\boldsymbol{E(t)}
\end{bmatrix}
=\boldsymbol{U(t)}\boldsymbol{c(t)},
\end{equation*}
and solve $\boldsymbol{c(t)}$.

By \eqref{mean_field_equation1},
\begin{equation*}
d\boldsymbol{c(t)}=-\frac{\gamma}{\lambda^H+2\eta} \sum_{k=1}^K \xi_k 1_{[t_k= t]} \boldsymbol{U(t_k)}^{-1} \begin{bmatrix} \boldsymbol{e}\\ \boldsymbol{0} \end{bmatrix},
\end{equation*}
so
\begin{equation*}
\boldsymbol{c(t)}=\boldsymbol{c(0)}-\frac{\gamma}{\lambda^H+2\eta} \sum_{k=1}^K \xi_k 1_{[t_k\leq t]} \boldsymbol{U(t_k)}^{-1} \begin{bmatrix} \boldsymbol{e}\\ \boldsymbol{0} \end{bmatrix},
\end{equation*}
and the only unknown is $\boldsymbol{c(0)}$.

Since $\boldsymbol{U(0)}=\boldsymbol{I}_{2N}$,
\begin{equation*}
\begin{bmatrix}
\boldsymbol{\mu(0)}\\
\boldsymbol{E(0)}
\end{bmatrix} = \boldsymbol{c(0)},
\end{equation*}
by the initial condition $\boldsymbol{E(0)}=\boldsymbol{E}_0$, we get the last $N$ entries of $\boldsymbol{c(0)}$.

Using the terminal condition $[2\eta\boldsymbol{I}+\lambda^H\boldsymbol{ep}^\top]\boldsymbol{\mu(T)}+2\boldsymbol{\Gamma}\boldsymbol{E(T)}=\boldsymbol{0}$, we get
\begin{equation*}
\begin{bmatrix}
2\eta\boldsymbol{I}+\lambda^H\boldsymbol{ep}^\top & 2\boldsymbol{\Gamma}
\end{bmatrix}
\boldsymbol{U(T)}\left(\boldsymbol{c(0)}-\frac{\gamma}{\lambda^H+2\eta} \sum_{k=1}^K \xi_k \boldsymbol{U(t_k)}^{-1} \begin{bmatrix} \boldsymbol{e}\\ \boldsymbol{0} \end{bmatrix}   \right)=\boldsymbol{0},
\end{equation*}
and the first $N$ entries of $\boldsymbol{c(0)}$ can be found by solving the above linear system.
\item[Step 4] Solve $E=\boldsymbol{p}^\top\boldsymbol{E}.$
\end{itemize}

\subsubsection{\texorpdfstring{$\varepsilon-$}{}Nash equilibrium}
In Theorem \ref{thm-epsilonNash} we prove that, the partial Nash equilibrium with infinite HFTs well approximates the Nash equilibrium among enough HFTs.
\begin{theorem}[Partial $\varepsilon-$Nash equilibrium]
\label{thm-epsilonNash}
Assume that there exists $m>0$ such that
$$
P(\sup_{j\in\mathbbm{N}^+} \vert X_j(0)\vert\leq m)=1.
$$
Given LT's strategy $\{\xi_k\}_{k=1}^K$, suppose there are $M$ HFTs in the market, let $(v_1^*, ..., v_M^*)$ be the MFG strategy of HFTs.
Then $\forall \varepsilon>0$, there exists $\Tilde{M}:=\Tilde{M}(\omega,\varepsilon)\in\mathbbm{N}^+$ such that when $M\geq \Tilde{M}$,
$(v_1^*, ..., v_M^*)$ forms an $M$-HFT partial $\varepsilon$-Nash equilibrium.
\end{theorem}
\begin{proof}
    See \ref{proof-thm-epsilonNash} of \nameref{appendix}.
\end{proof}

\subsection{Overall Nash equilibrium}
In this section, we investigate the overall Nash equilibrium, where LT also participates in the game.
Given HFTs' strategy, LT maximizes her expected profit \eqref{LTobj} and from the first-order condition, her optimal strategy follows:
\begin{equation}
\label{LTopt}
\begin{aligned}
\xi_k=&-\frac{\xi_0}{K}+\frac{1}{\gamma+2(\lambda+\eta_0)}\{\gamma^H(E(t_K)-E(t_k))+\lambda^H(\mu(t_K)-\mu(t_k))\\
&-\frac{1}{K}\sum_{j=1}^{K-1}[\gamma^H(E(t_K)-E(t_j))+\lambda^H(\mu(t_K)-\mu(t_j))]\},\ k=1,...,K-1,\\
\xi_K=&-\xi_0-\sum_{j=1}^{K-1}\xi_j.
\end{aligned}
\end{equation}

Observe that the solution to \eqref{mean_field_equation}, $\boldsymbol{E}(\boldsymbol{E_0},\boldsymbol{\xi})$, is linear in $\boldsymbol{E}_0=\begin{bmatrix}
    E_{01}\\
    \vdots\\
    E_{0N}
\end{bmatrix}$ and $\boldsymbol{\xi}=\begin{bmatrix}
    \xi_{1}\\
    \vdots\\
    \xi_{K}
\end{bmatrix}$ in the sense that 
\begin{equation*}
\boldsymbol{E}(\boldsymbol{E_0},\boldsymbol{\xi})=\sum_{i=1}^N E_{0i}\boldsymbol{E}(\boldsymbol{E_0}^{(i)},\boldsymbol{0})+\sum_{k=1}^K\xi_{k}\boldsymbol{E}(\boldsymbol{0},\boldsymbol{\xi}^{(k)}),
\end{equation*}
where
\begin{equation*}
\begin{aligned}
 &\boldsymbol{E_0}^{(i)}=\begin{bmatrix}
    0\\
    \vdots\\
    0\\
    1_i\\
    0\\
    \vdots\\
    0
\end{bmatrix},\ i=1,...,N;   \ \boldsymbol{\xi}^{(k)}=\begin{bmatrix}
    0\\
    \vdots\\
    0\\
    1_k\\
    0\\
    \vdots\\
    0
\end{bmatrix},\ k=1,...,K.   
\end{aligned}
\end{equation*}
Consequently, $E$ and $\mu$ are also linear in $\boldsymbol{E_0}$ and $\boldsymbol{\xi}$:
\begin{equation}\label{mean_field_linear}
\begin{aligned}
E &= \sum_{i=1}^N E_0^i E(\boldsymbol{E}_0^{(i)},\boldsymbol{0})+ \sum_{k=1}^K \xi_k E(\boldsymbol{0},\boldsymbol{\xi}^{(k)}):= C_E + \sum_{k=1}^K \xi_k E(\boldsymbol{0},\boldsymbol{\xi}^{(k)}), \\
\mu &= \sum_{i=1}^N E_0^i \mu(\boldsymbol{E}_0^{(i)},\boldsymbol{0})+ \sum_{k=1}^K \xi_k \mu(\boldsymbol{0},\boldsymbol{\xi}^{(k)}):= C_{\mu}+ \sum_{k=1}^K \xi_k \mu(\boldsymbol{0},\boldsymbol{\xi}^{(k)}).
\end{aligned}
\end{equation}
Substitute \eqref{mean_field_linear} into \eqref{LTopt}, we get linear equations of $\{\xi_k\}_{k=1}^{K-1}$ in equilibrium. The solving scheme is:
\begin{itemize}
    \item[Step 1] Solve bases: applying the solving scheme in Section \ref{subsec-solvingscheme}, solve $\left\{E(\boldsymbol{E}_0^{(i)},\boldsymbol{0})\right\}_{i=1}^N$, $\left\{\mu(\boldsymbol{0},\boldsymbol{\xi}^{(k)})\right\}_{i=1}^N$, $\left\{E(\boldsymbol{E}_0^{(i)},\boldsymbol{0})\right\}_{k=1}^K$, $\left\{\mu(\boldsymbol{0},\boldsymbol{\xi}^{(k)})\right\}_{k=1}^K$.
    \item[Step 2] Solve $\{\xi_k\}_{k=1}^K:$ substitute above bases into \eqref{LTopt} and \eqref{mean_field_linear}, solve the linear system.
    \item[Step 3] Solve HFTs' average position: substitute $\{\xi_k\}_{k=1}^K$ into \eqref{mean_field_linear}, HFTs' average position $E$ is calculated.
\end{itemize}

Actually, the overall Nash equilibrium is also a good approximation for an overall Nash equilibrium with enough HFTs.
\begin{theorem}[Overall $\varepsilon-$Nash equilibrium]
\label{thm-epsilonNash2}
Assume that there exists $m>0$ such that
$$
P(\sup_{j\in\mathbbm{N}^+} \vert X_j(0)\vert\leq m)=1.
$$
    Suppose there are $M$ HFTs in the market, let $\xi^*$ and $(v_1^*, ..., v_M^*)$ be the MFG strategy of LT and HFTs, respectively. Then $\forall \varepsilon>0$, there exists $\Tilde{M}:=\Tilde{M}(\omega,\varepsilon)\in\mathbbm{N}^+$ such that when $M\geq \Tilde{M}$,
     $(\xi^*, v_1^*,...,v_M^*)$ forms an $M$-HFT overall $\varepsilon$-Nash equilibrium.
\end{theorem}
\begin{proof}
See \ref{proof-thm-epsilonNash2} of \nameref{appendix}.
\end{proof}

\newpage
\section{Numerical Results}
\subsection{Partial Nash equilibrium}
In this section, the strategy of LT is fixed, HFT may take different anticipatory strategies and will affect LT differently. Trading period starts from $0$ and ends at $T=1.$ LT trades $K=9$ times and $\xi_k=1,t_k=\frac{k}{10},k=1,...,9.$ The mean initial position of HFTs $E_0=0.$
We first show HFTs' strategies and how LT is affected when HFTs are identically inventory averse and the aversion does not jump. The values of market impact coefficients do not affect conclusions about HFTs' strategies and we set
$\gamma=1,\gamma^H=0.7,\lambda=0.4,\lambda^H=0.1,\eta=\eta_0=0.05.$ 

From Figure \ref{fignojumpLTnogamephi=0}, given $\phi,$ with the increase of $\Gamma$, HFTs play the role of Round-Tripper throughout the trading period. From Figure \ref{fignojumpLTnogamegam=0}, given $\Gamma$, with the increase of $\phi,$ HFTs play the role of Round-Tripper each time the large order is executed. Interestingly, if we focus on the time period between two transactions of LT, it can be seen that HFTs first trade in the opposite direction of LT then trade in the same direction, it verifies the finding in  van Kervel and Menkveld (2019) \cite{van2019high} that HFTs initially lean against institutional orders but later change direction.

\begin{figure}[!htbp]
    \centering
\subcaptionbox{$\Gamma=0$}{
    \includegraphics[width = 0.27\textwidth]{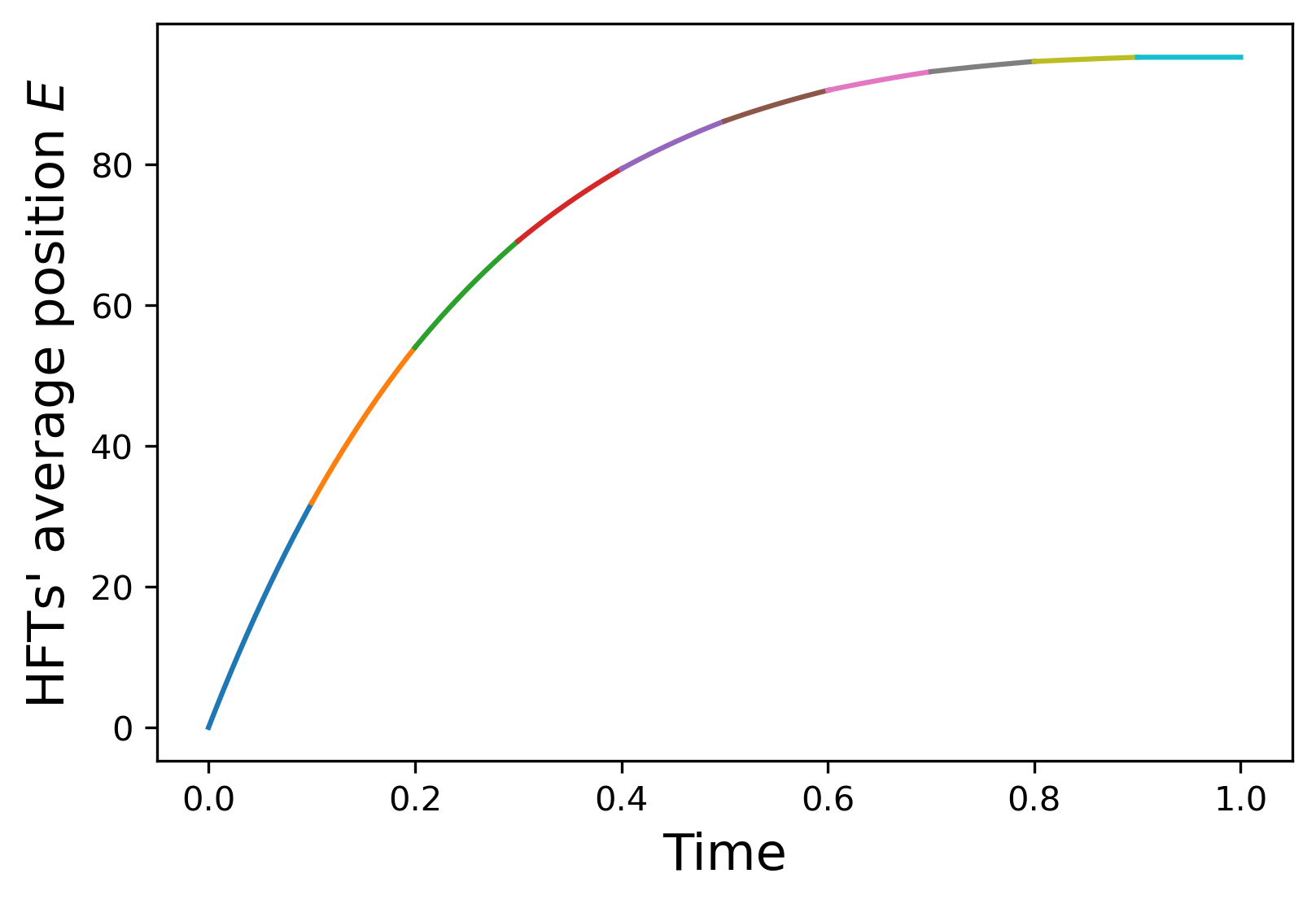}
    }
\subcaptionbox{$\Gamma=0.1$}{
    \includegraphics[width = 0.27\textwidth]{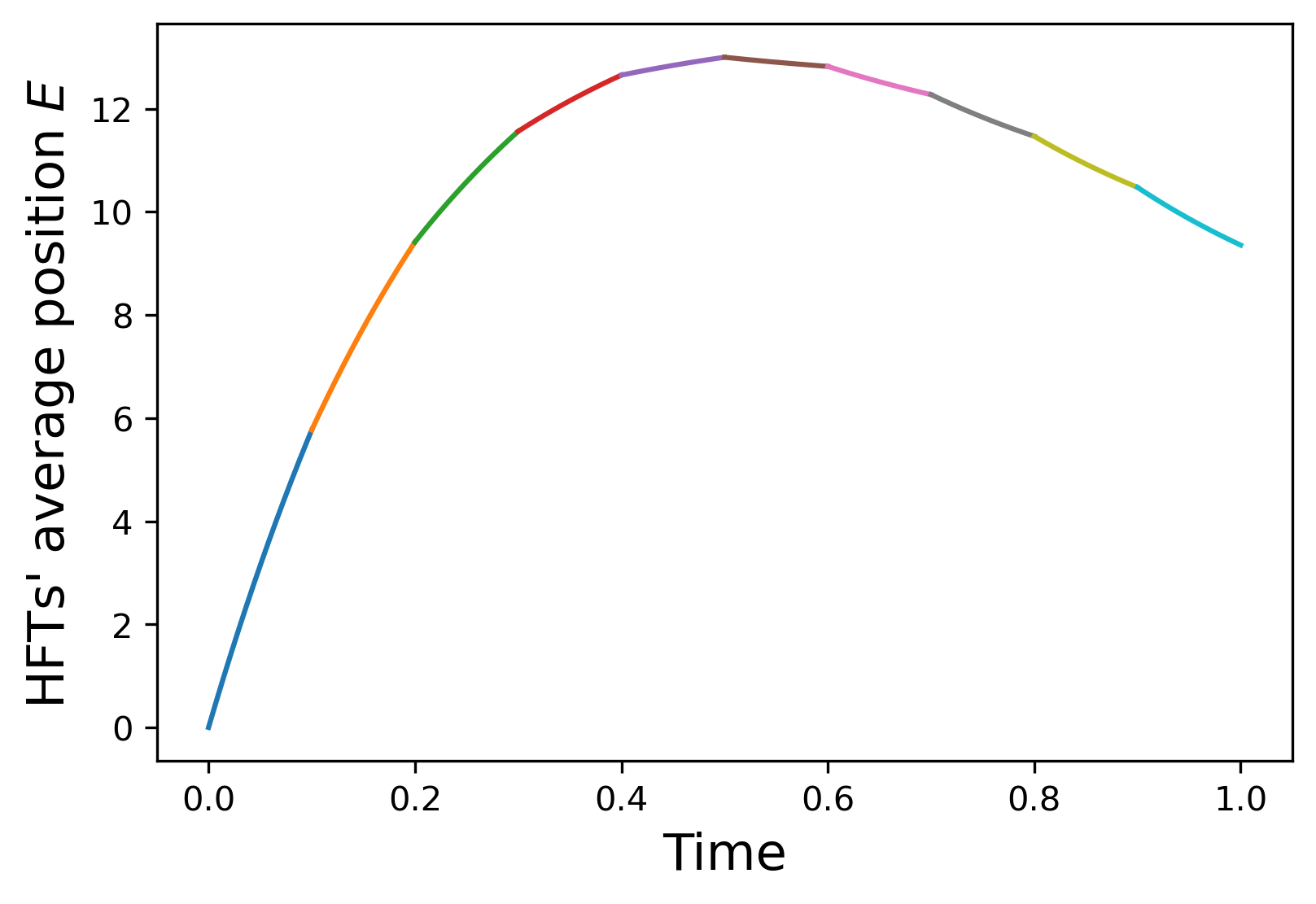}
    }
 \subcaptionbox{$\Gamma=2$}{
    \includegraphics[width = 0.27\textwidth]{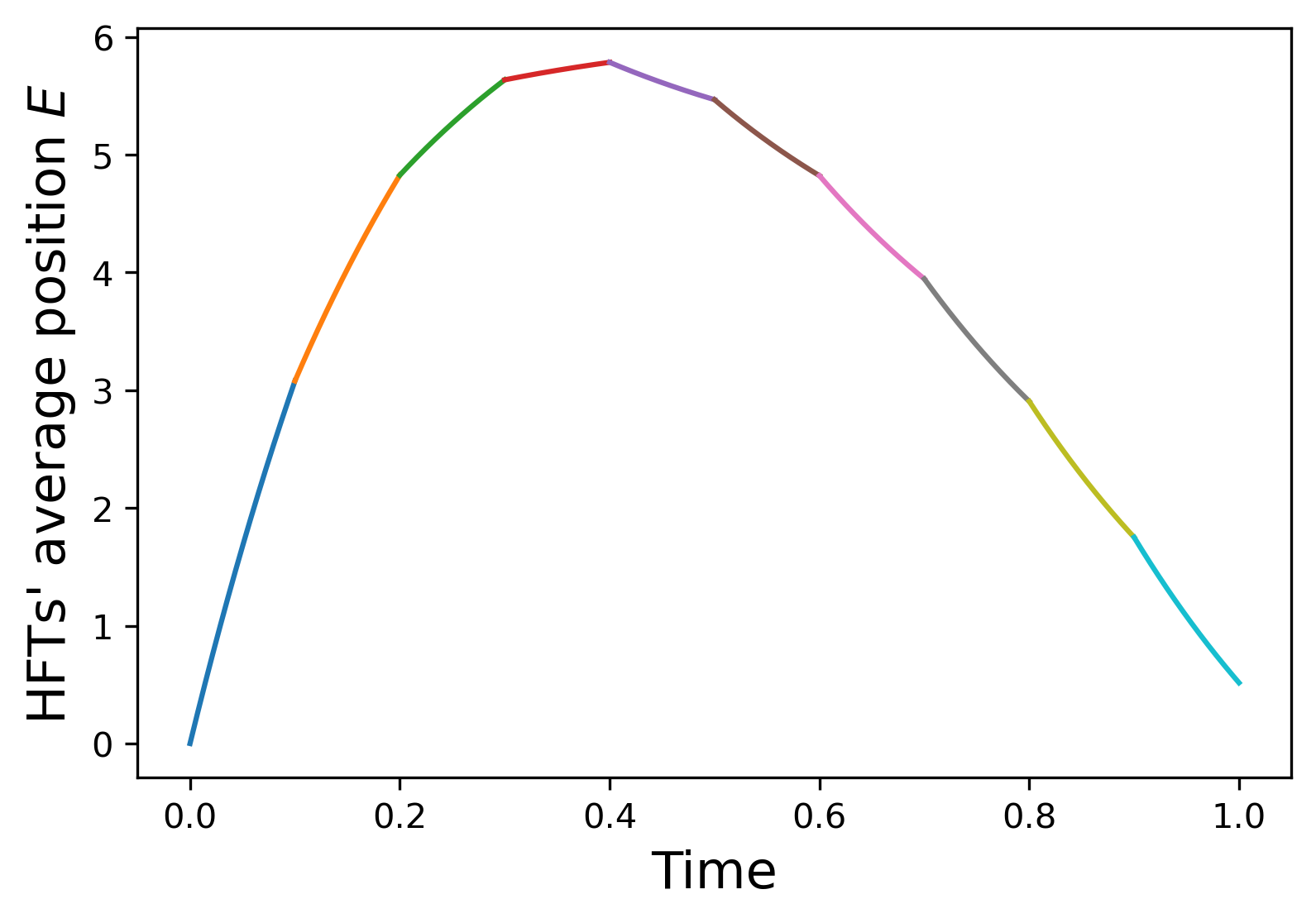}
    }   
    \caption{$\phi=0,$ HFTs' average position $E$.}
    \label{fignojumpLTnogamephi=0}
\end{figure}  

\begin{figure}[!htbp]
    \centering
\subcaptionbox{$\phi=0$}{
    \includegraphics[width = 0.27\textwidth]{fig/nojump/LTnogame/phi0Gam0.png}
    }
\subcaptionbox{$\phi=5$}{
    \includegraphics[width = 0.27\textwidth]{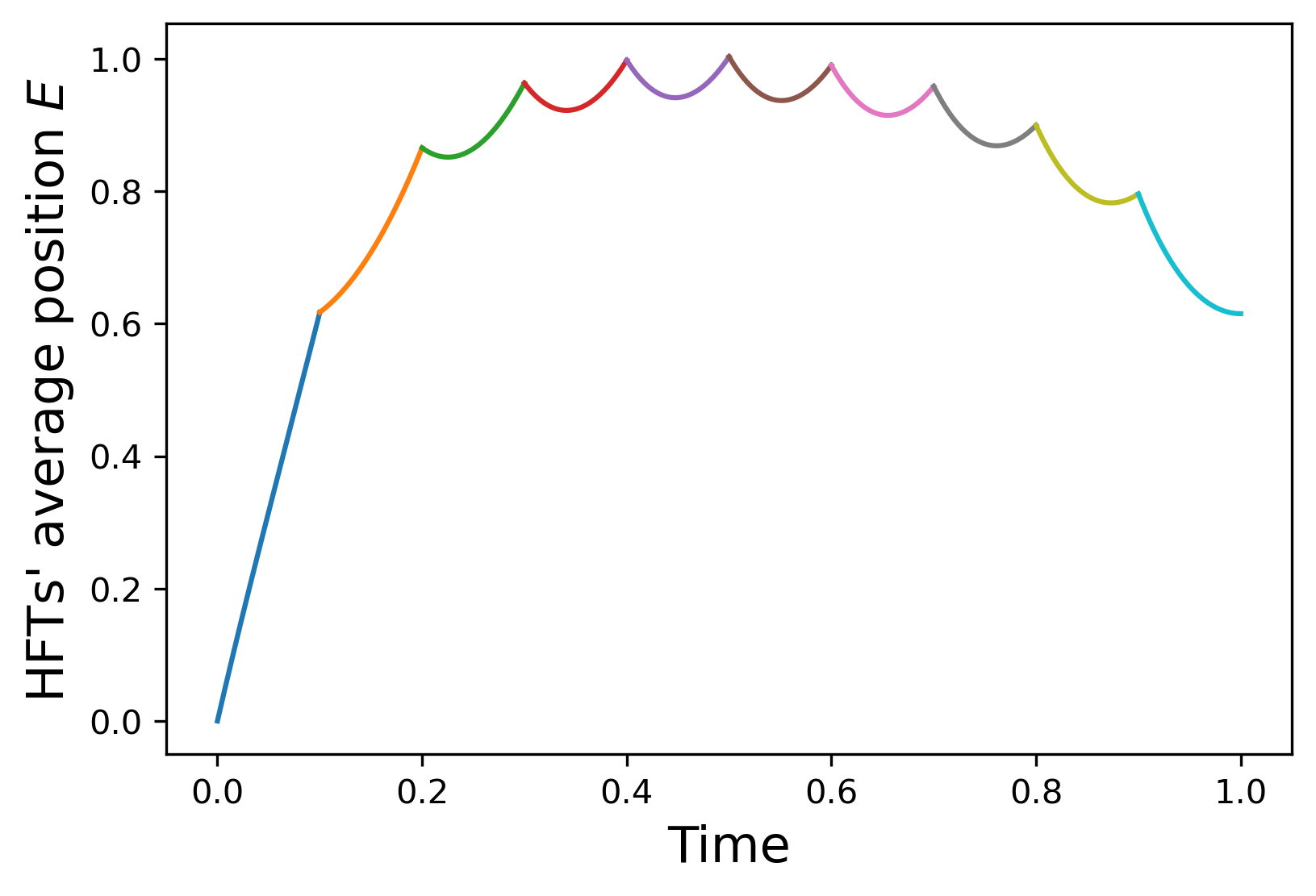}
    }
 \subcaptionbox{$\phi=10$}{
    \includegraphics[width = 0.27\textwidth]{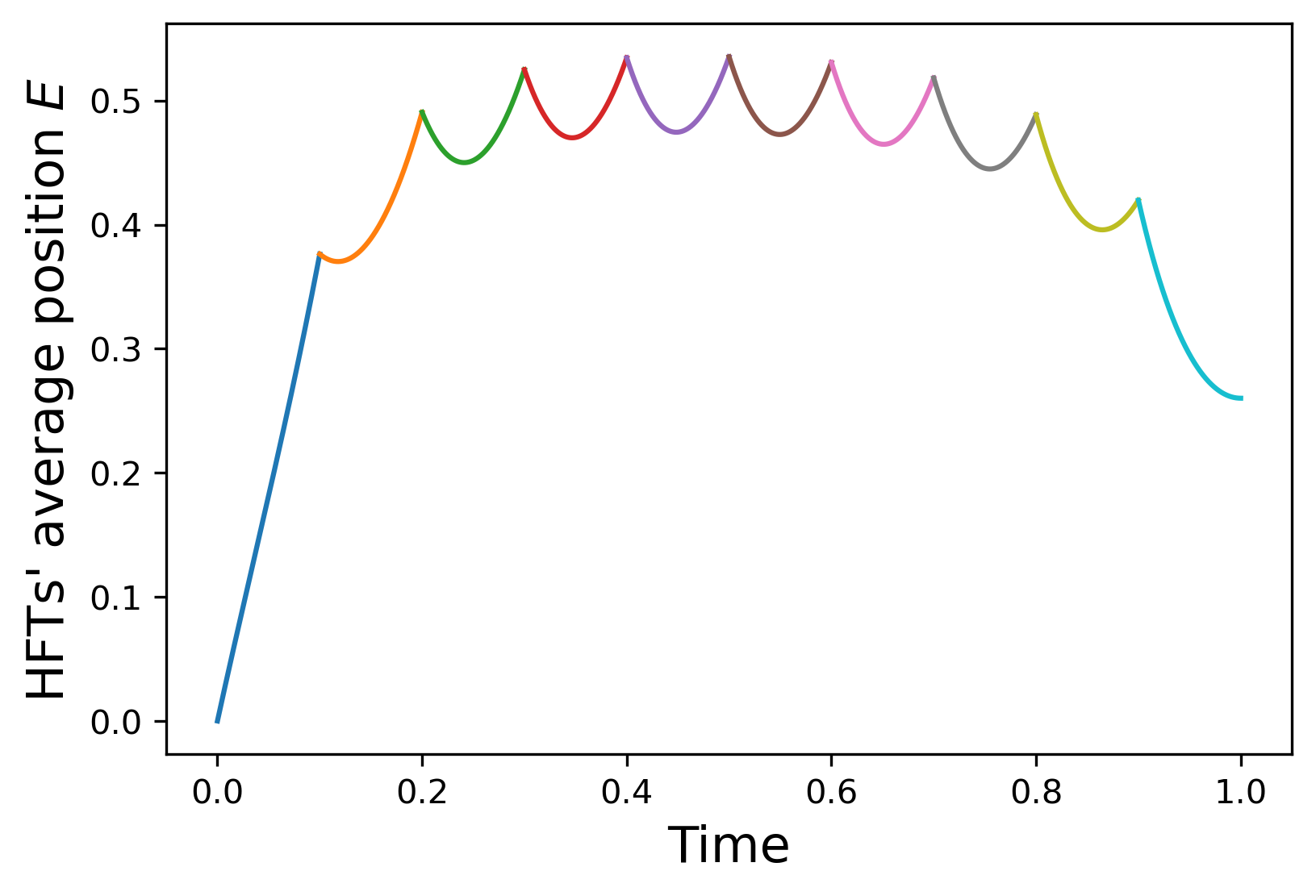}
    }   
    \caption{$\Gamma=0,$ HFTs' average position $E$.}
        \label{fignojumpLTnogamegam=0}
\end{figure}  



Next, we investigate HFTs' influences on LT's profit. Without HFTs, LT's expected profit is:
\begin{equation*}
\begin{aligned}
\mathbbm{E}(\pi^{LT}_0)
=P_0\xi_0-\gamma\sum_{k=1}^K\xi_k\sum_{j=1}^k\xi_j-(\lambda+\eta_0)\sum_{k=1}^K\xi_k^2.  
\end{aligned}
\end{equation*}
In presence of HFTs, LT's expected profit is:
\begin{equation*}
\begin{aligned}
\mathbbm{E}(\pi^{LT})=
\mathbbm{E}(\pi^{LT}_0)+\sum_{k=1}^K(-\xi_k)[\gamma^H(E(t_k)-E(0))+\lambda^H\mu(t_k)].
\end{aligned}
\end{equation*}
In other words, the extra cost (or surplus) that HFTs add to LT is the extra price impact that they cause.

In Figure \ref{fignojumpLTnogameLTprofitdiff}, we analyze how an inventory-averse HFT impact LT when the temporary impact $\lambda^H$ changes. Since the permanent impact $\gamma^H$ is given, it actually represents the relative size of $\lambda^H$ to $\gamma^H$. LT makes more profits with Round-Trippers if the temporary impact $\lambda^H$ is relatively large to the permanent impact $\gamma^H$. A large temporary impact implies that the majority of adverse impact of HFTs' preemitive same-direction trading will disappear and the opposite-direction trading decreases the cost of LT's order.
\begin{figure}[!htbp]
    \centering
    \includegraphics[width = 0.35\textwidth]{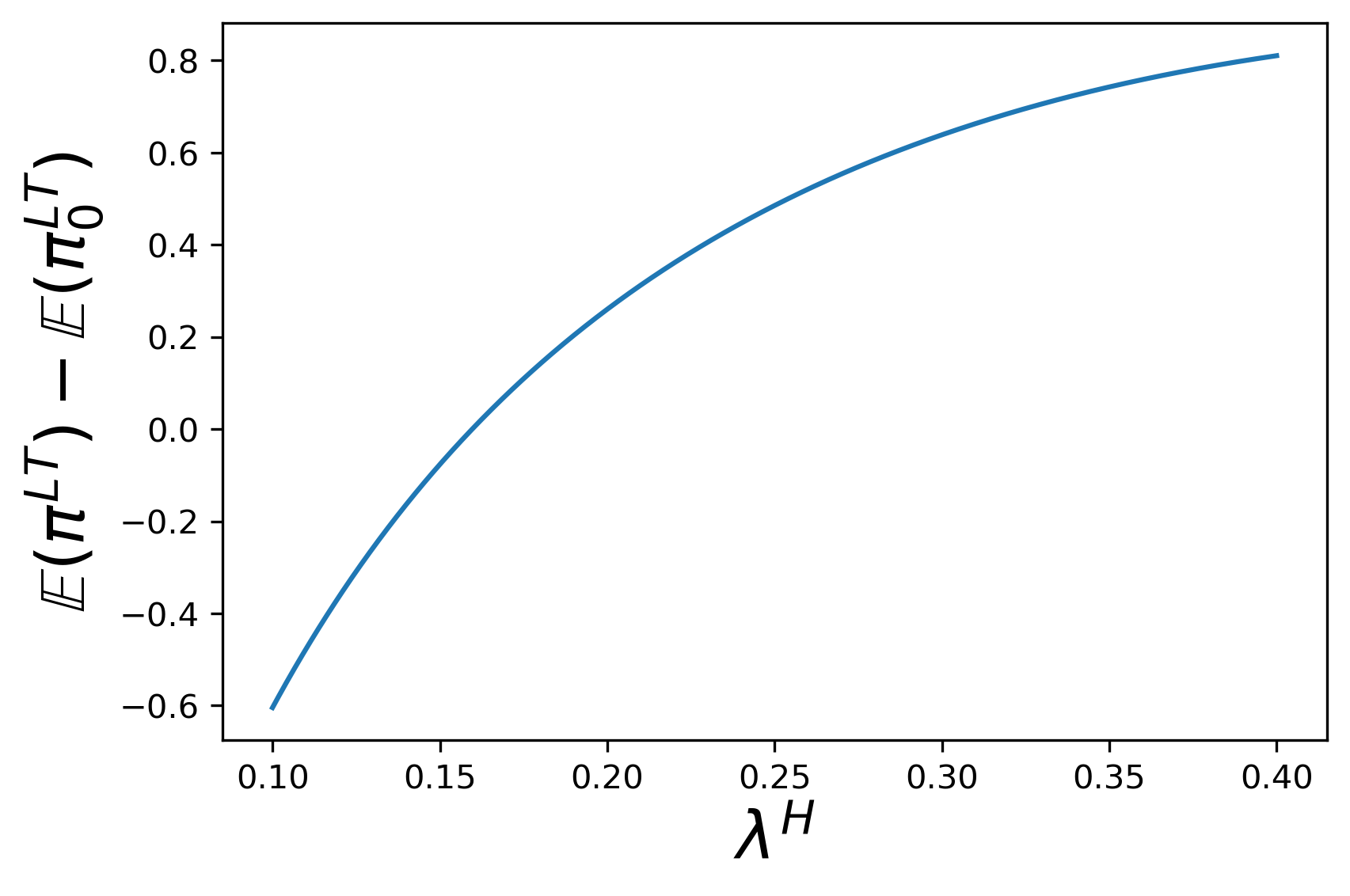}
    \caption{Difference of LT's profit when HFTs act as Round-Tripper ($\phi=10,\Gamma=2$).}
    \label{fignojumpLTnogameLTprofitdiff}
\end{figure}

In the general case, let $\boldsymbol{Q}=\begin{bmatrix}
        -x&x\\
        y&-y
    \end{bmatrix}$ and $\boldsymbol{p(0)}=\begin{bmatrix}
        1/2\\
        1/2
    \end{bmatrix}.$ If $x=y=0,$ the aversions of HFTs do not jump.

When $\phi(1)=\Gamma(1)=0$ and $\phi(2)=10,\Gamma(2)=2,$ HFTs are either with zero inventory aversions or averse to both running and ending positions. Comparing (a) and (c) in Figure \ref{figjumpLTnogameHFT}, where $p_1(t)=p_2(t)=\frac{1}{2},t\in[0,1],$ the crowd of HFTs exhibits different patterns of behavior. When $x=y=0.5,$ HFTs who are at $\phi=\Gamma=0$ should consider the switch to the other state $\phi=10,\Gamma=2,$ therefore control inventories in advance. In contrast, when $x=y=0,$ HFTs who are at $\phi=\Gamma=0$ are always free to accumulate positions. Consequently, although the components of HFTs are same in these two cases, HFTs whose inventory aversions may jump are more sensitive to positions and play the role of Round-Tripper.

When $x,y>0,$ a larger $x$ implies that the proportion of HFTs with higher inventory aversion will grow, consequently, the crowd of HFT tends to perform like a Round-Tripper both throughout and within the execution period. What's more, when HFTs play the role of Round-Tripper, LT is benefited if the temporary impact is relatively large to the permanent impact.

\begin{figure}[!htbp]
    \centering
    \subcaptionbox{$x=0,y=0$}{
    \includegraphics[width = 0.27\textwidth]{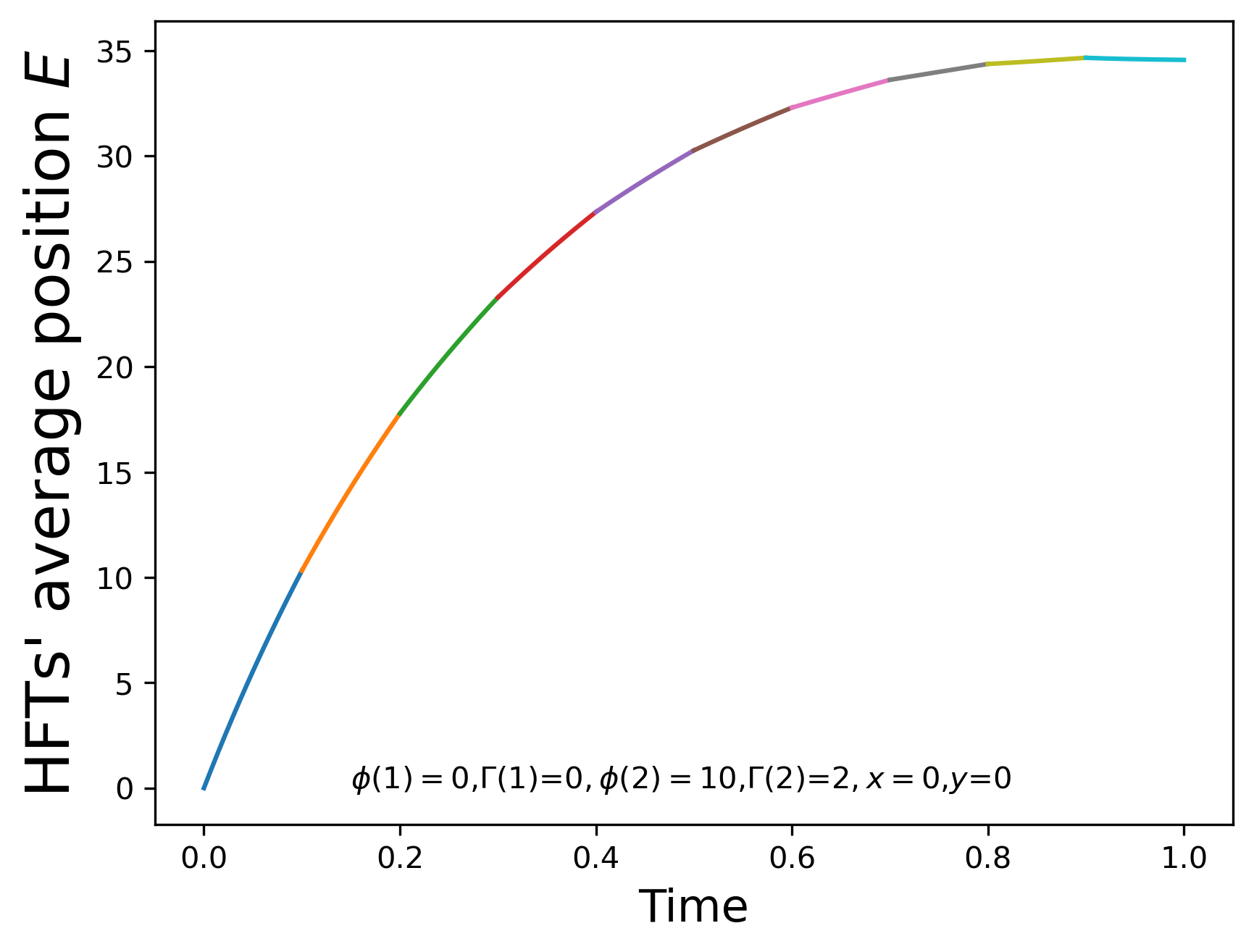}
    }
\subcaptionbox{$x=0.2,y=0.8$}{
    \includegraphics[width = 0.27\textwidth]{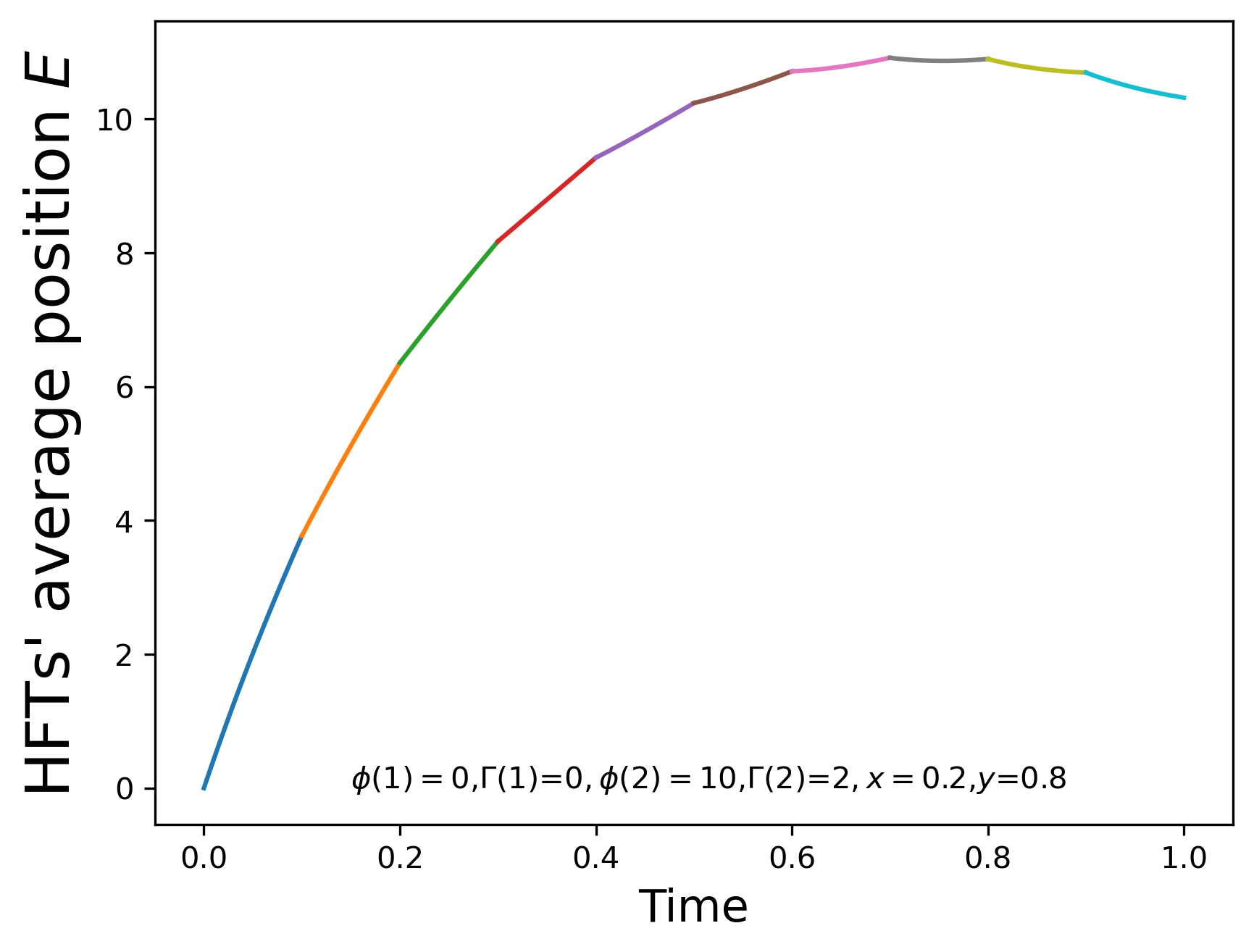}
    }
    
\subcaptionbox{$x=0.5,y=0.5$}{
    \includegraphics[width = 0.27\textwidth]{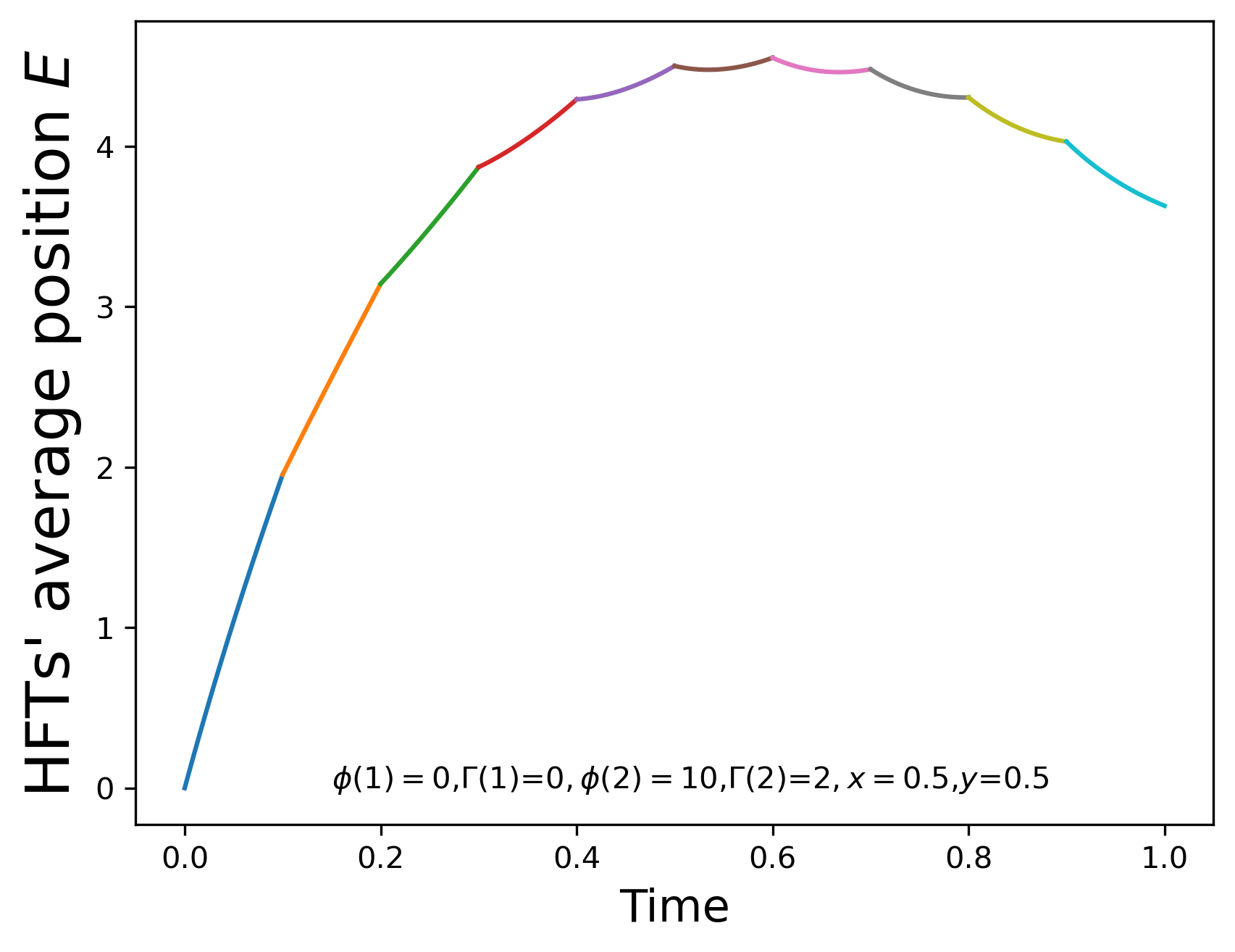}
    }
 \subcaptionbox{$x=0.8,y=0.2$}{
    \includegraphics[width = 0.27\textwidth]{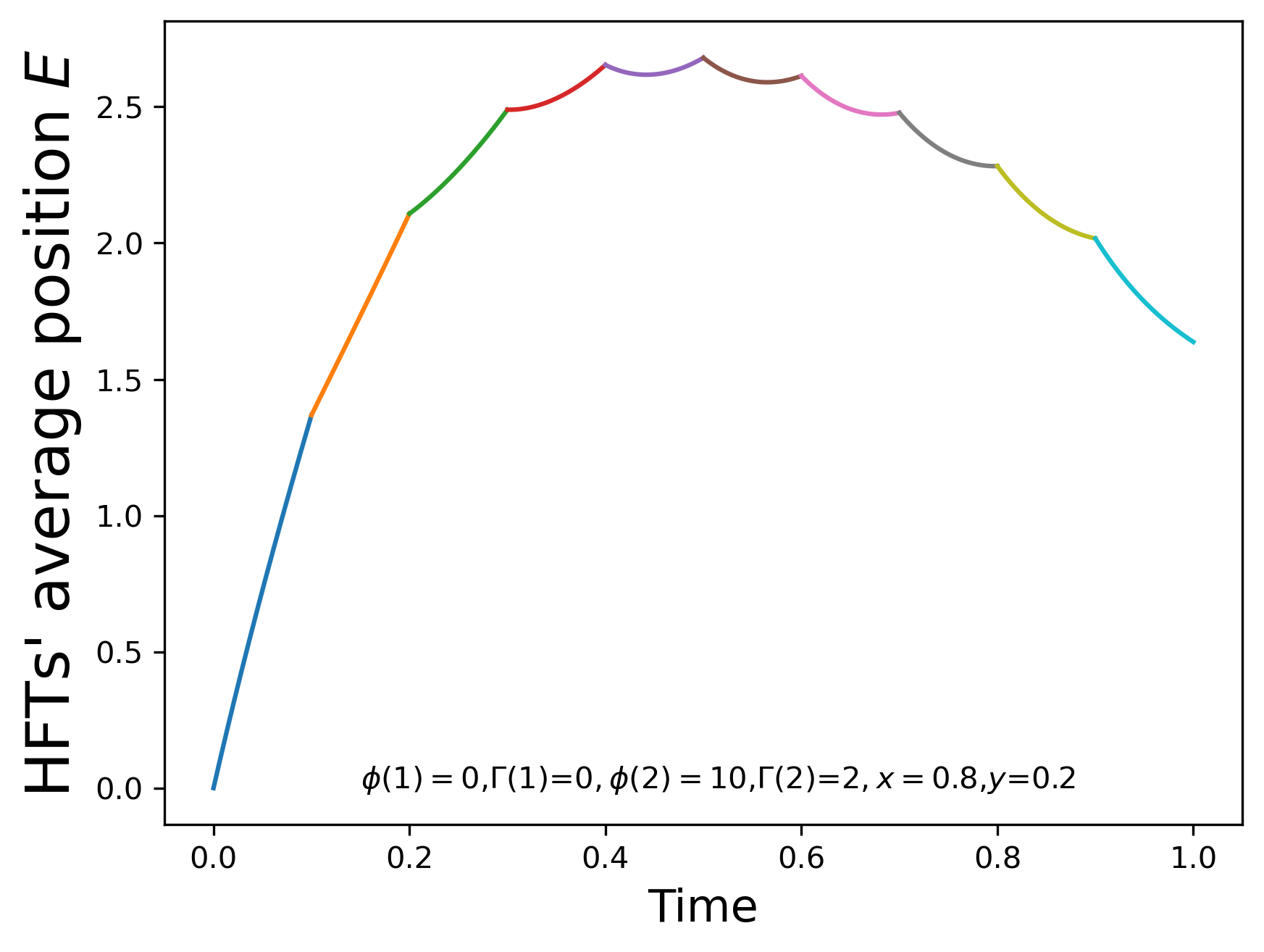}
    }   
\caption{$\phi(1)=0,\Gamma(1)=0,\phi(2)=10,\Gamma(2)=2,$ HFTs' average position $E$.}
\label{figjumpLTnogameHFT}
\end{figure}

When $\phi(1)=0,\Gamma(1)=2$ and $\phi(2)=10,\Gamma(2)=0,$ a larger $x$ implies that the crowd of HFT hates running position more; a larger $y$ implies that the crowd of HFT hates ending position more. As a result, with the increase of $x$, the tendency that HFTs play the role of Round-Tripper at each transaction of the large order is more obvious.

\begin{figure}[!htbp]
    \centering
    \subcaptionbox{$x=0,y=0$}{
    \includegraphics[width = 0.27\textwidth]{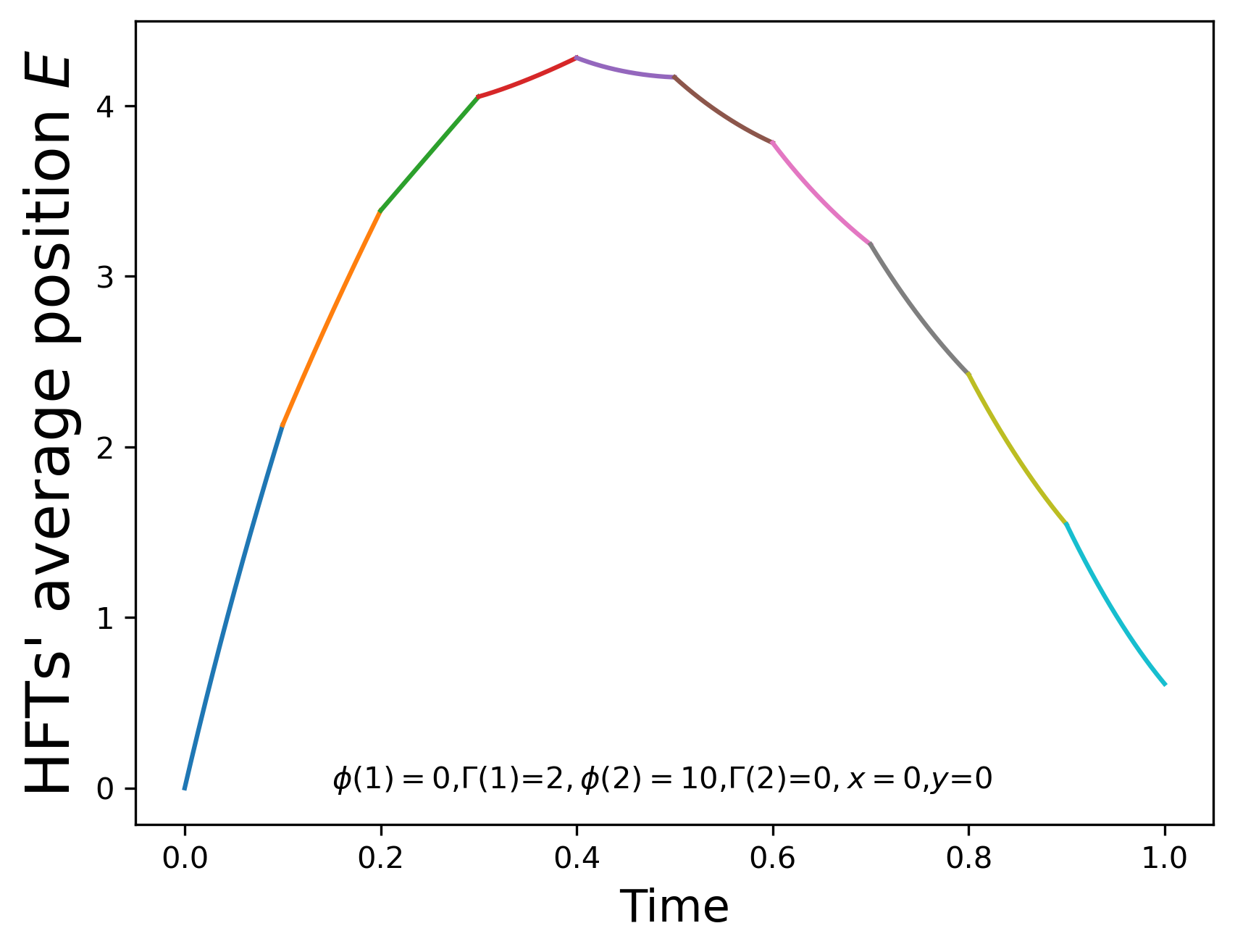}
    }
\subcaptionbox{$x=0.2,y=0.8$}{
    \includegraphics[width = 0.27\textwidth]{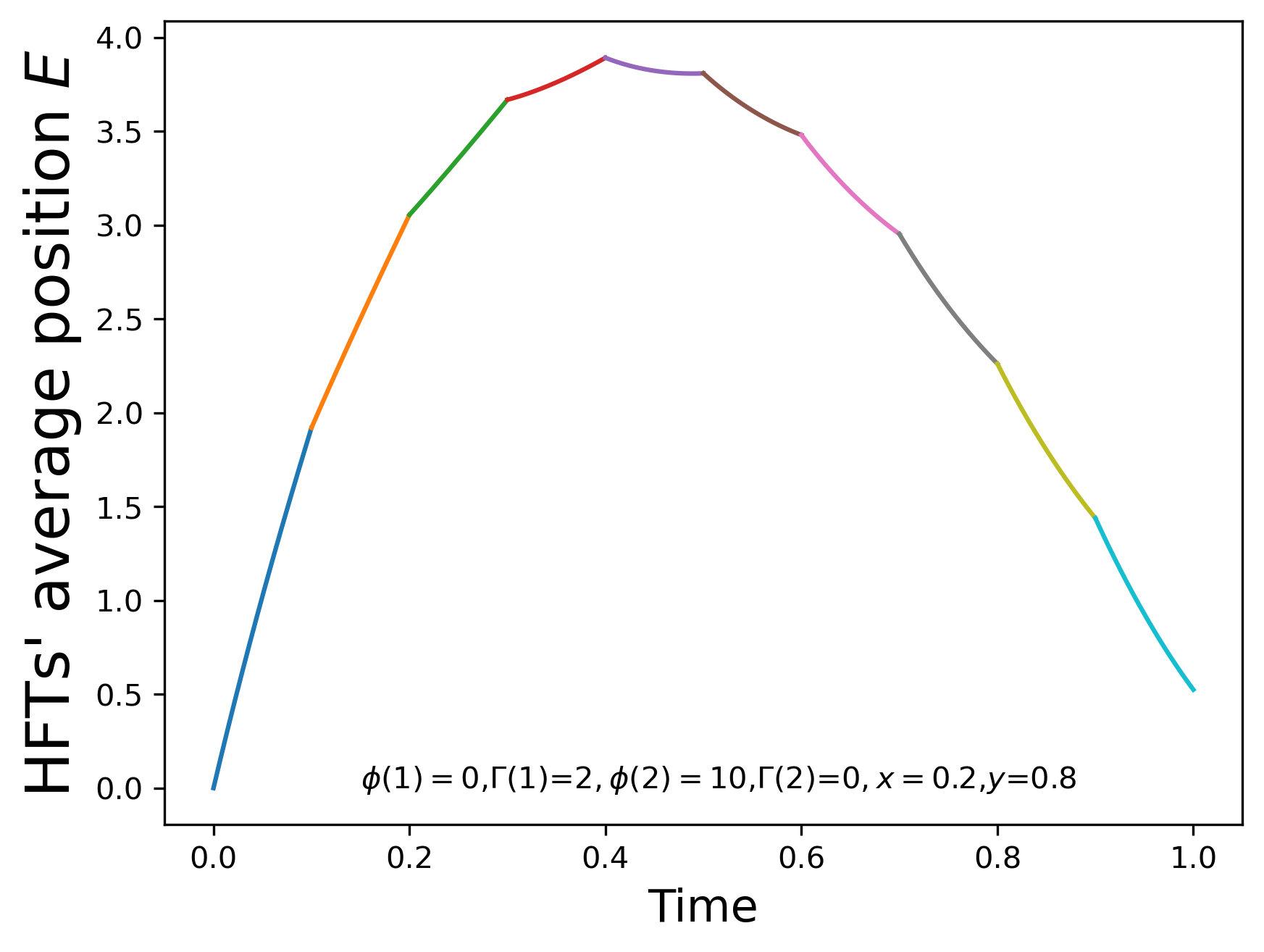}
    }
    
\subcaptionbox{$x=0.5,y=0.5$}{
    \includegraphics[width = 0.27\textwidth]{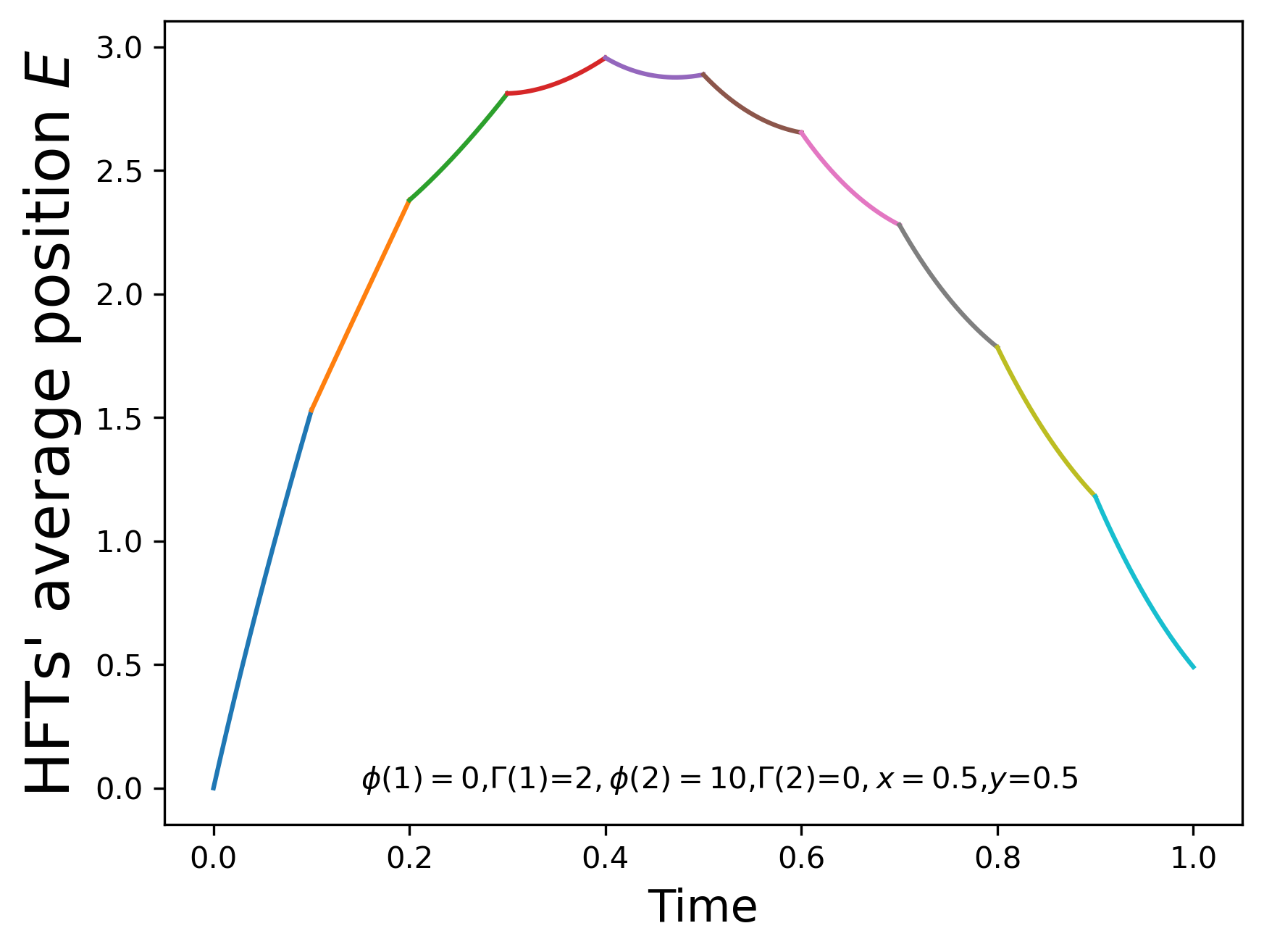}
    }
 \subcaptionbox{$x=0.8,y=0.2$}{
    \includegraphics[width = 0.27\textwidth]{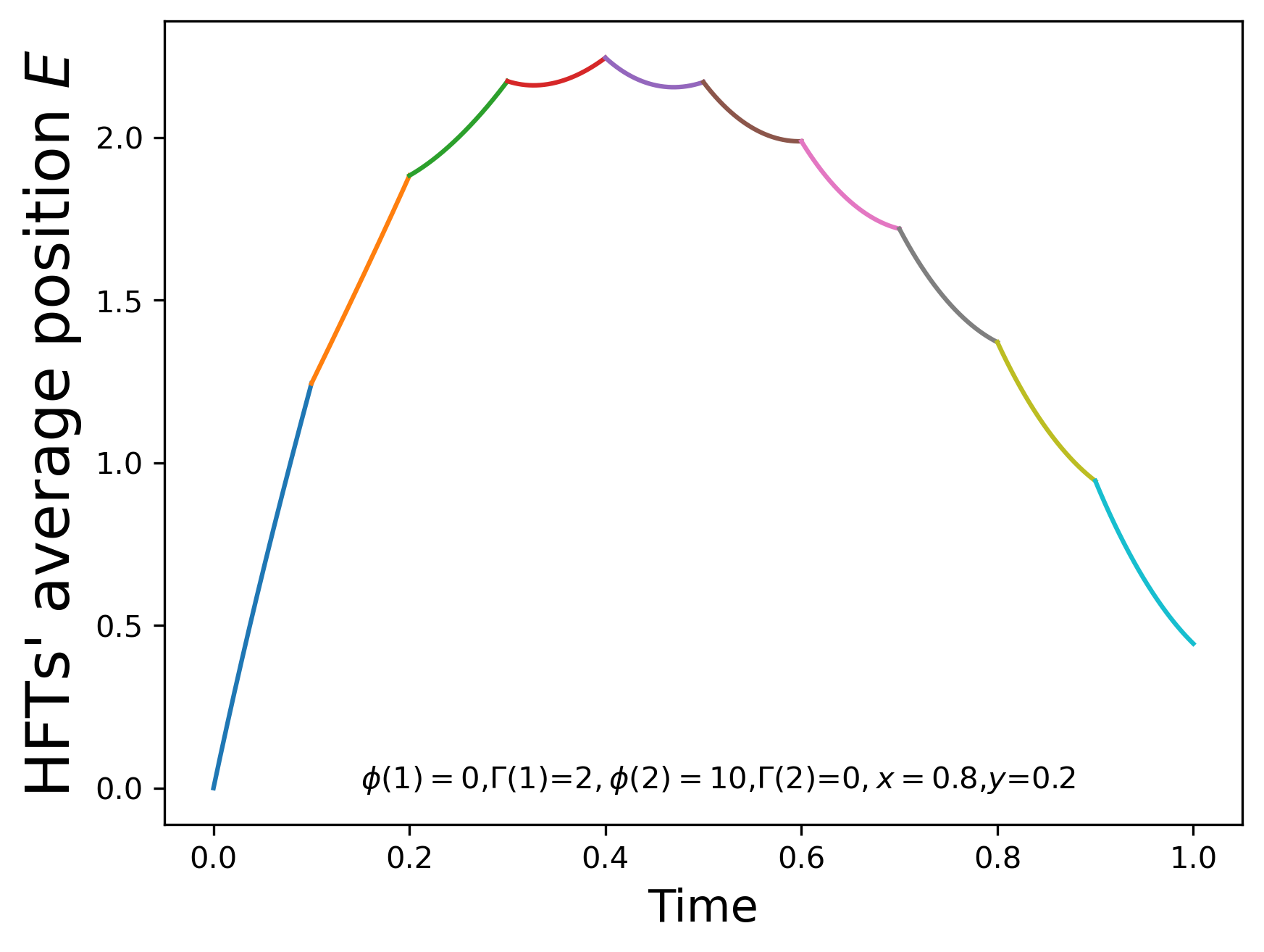}
    }   
\caption{$\phi(1)=0,\Gamma(1)=2,\phi(2)=10,\Gamma(2)=0,$ HFTs' average position $E$.}
\end{figure}

\newpage
\subsection{Overall Nash equilibrium}
In this section, the overall Nash equilibrium between LT and HFTs is studied.
We still first analyze the case where HFTs are of same inventory aversion.  \eqref{mean_field_equation_nojump} and \eqref{LTopt} give the unique Nash equilibrium between LT and HFTs.

LT is going to execute $\xi_0=-9$ during $[0,1]$ and she trades at time points $t_k=\frac{k}{10},k=1,...,9.$ Without HFTs, it is easy to verify that the optimal strategy for LT is
\begin{equation*}
    \xi_k=-\frac{\xi_0}{K},\ k=1,...,K.
\end{equation*}

Figure \ref{fignojumpLTgamephi=0LT} and \ref{fignojumpLTgamephi=0HFT} compare LT and HFTs' strategies in partial and overall equilibrium, when $\phi=0$. $\Gamma$ decides whether HFTs provide liquidity in the latter half of LT's execution period. 
When HFTs always trade along with LT (see (a) of Figure \ref{fignojumpLTgamephi=0HFT}), LT trade more intensely at the beginning, to avoid the cumulative adverse impact caused by HFTs' same-direction orders (see (a) of Figure \ref{fignojumpLTgamephi=0LT}). Correspondingly, it suppress the same-direction trading of HFTs.

When HFTs act as Round-Tripper (see (b) and (c) of Figure \ref{fignojumpLTgamephi=0HFT}), LT buys more when HFTs supply liquidity back (see (b) and (c) of Figure \ref{fignojumpLTgamephi=0LT}). A larger $\Gamma$ leads LT to delay trading more shares.
\begin{figure}[!htbp]
    \centering
\subcaptionbox{$\Gamma=0$}{
    \includegraphics[width = 0.27\textwidth]{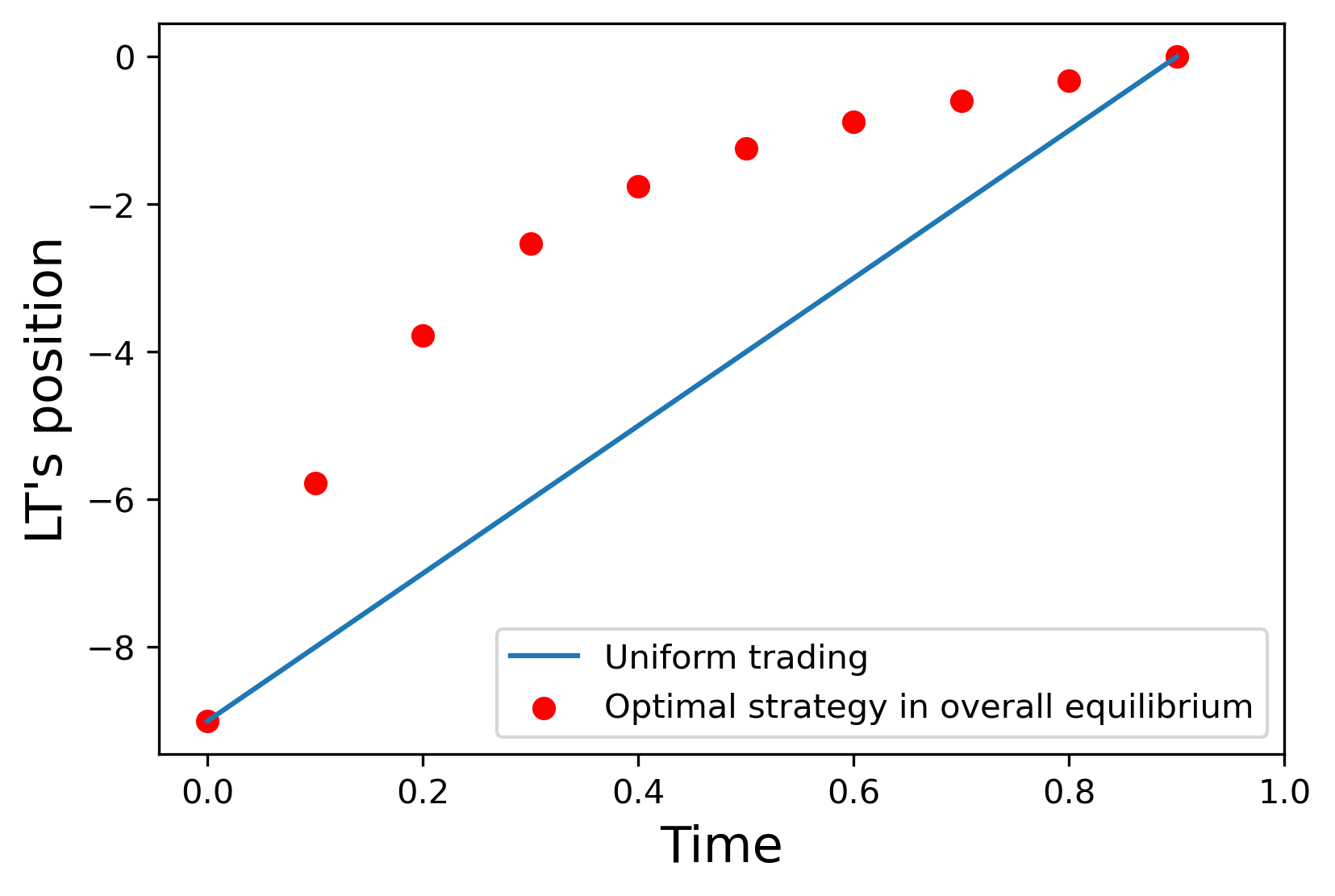}
    }
\subcaptionbox{$\Gamma=0.1$}{
    \includegraphics[width = 0.27\textwidth]{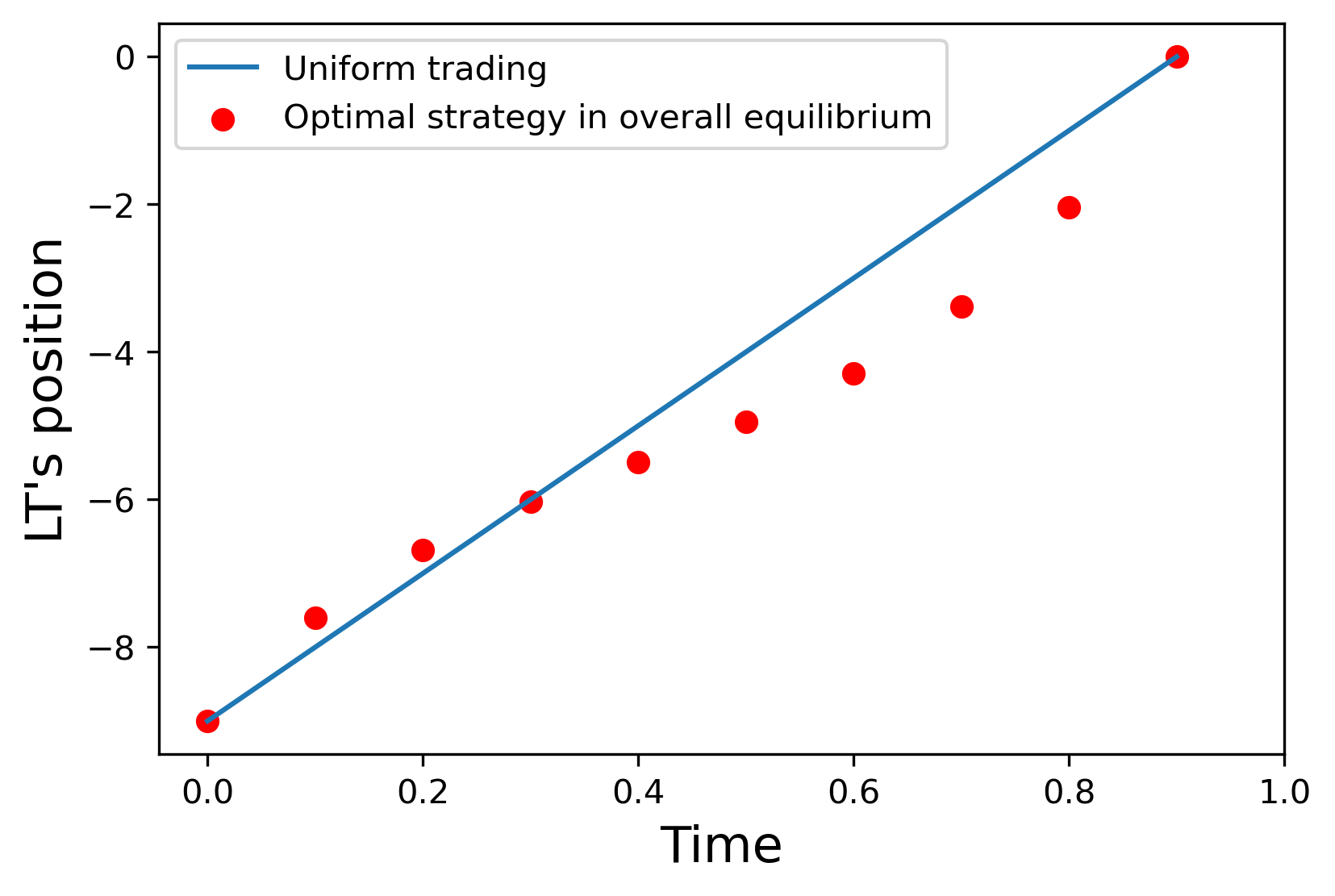}
    }
 \subcaptionbox{$\Gamma=2$}{
    \includegraphics[width = 0.27\textwidth]{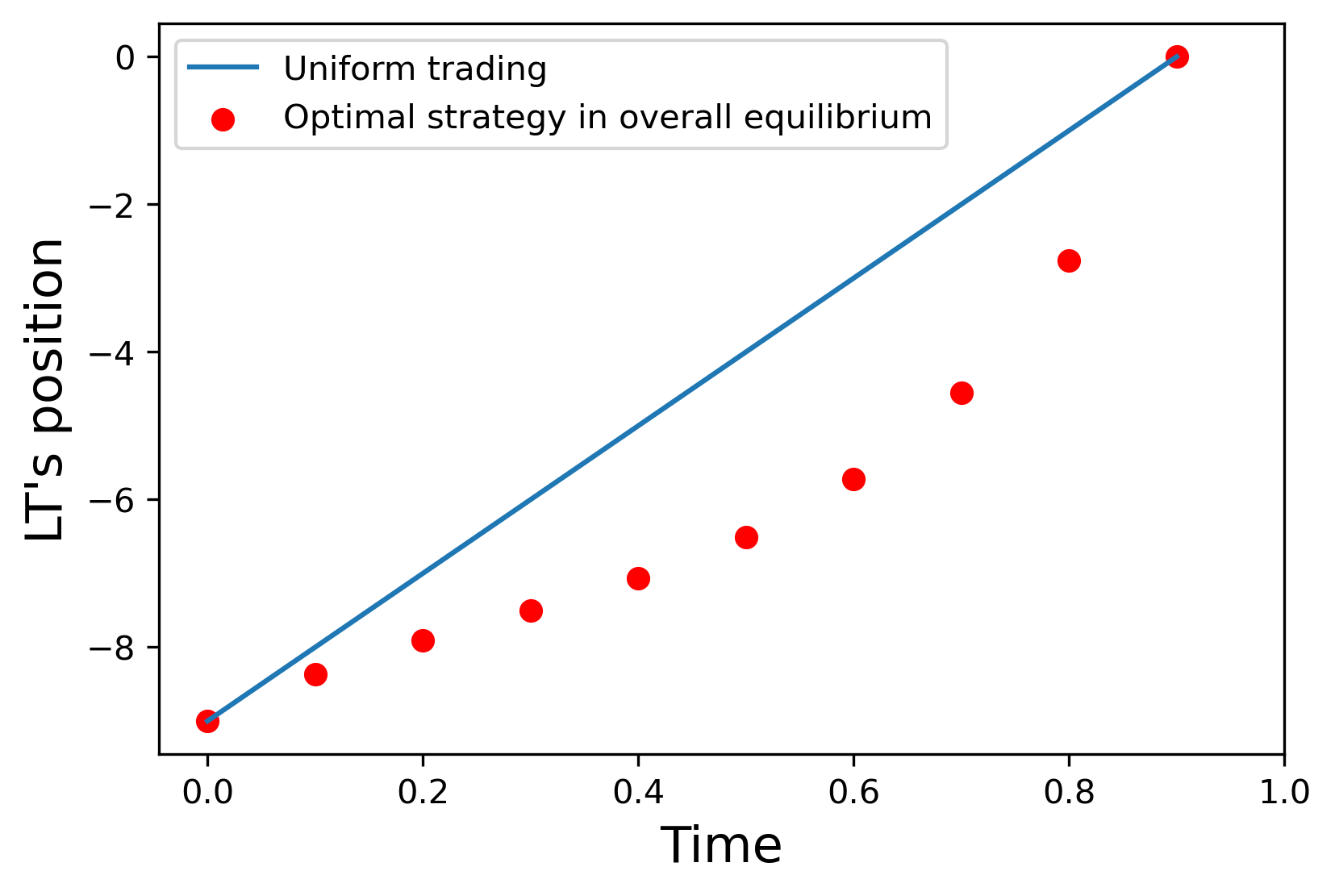}
    }   
    \caption{$\phi=0,$ LT's strategy.}
    \label{fignojumpLTgamephi=0LT}
\end{figure}  

\begin{figure}[!htbp]
    \centering
\subcaptionbox{$\Gamma=0$}{
    \includegraphics[width = 0.27\textwidth]{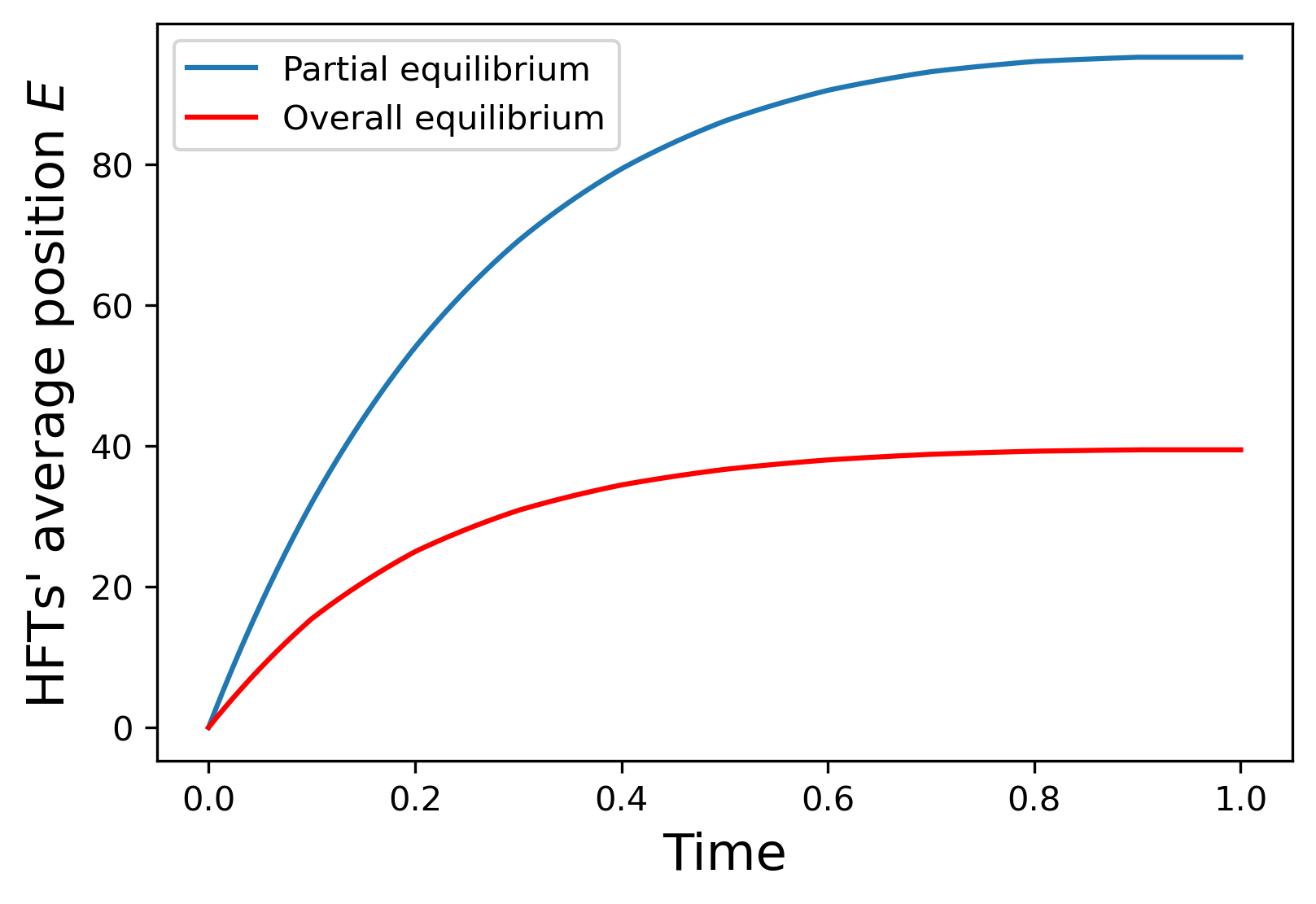}
    }
\subcaptionbox{$\Gamma=0.1$}{
    \includegraphics[width = 0.27\textwidth]{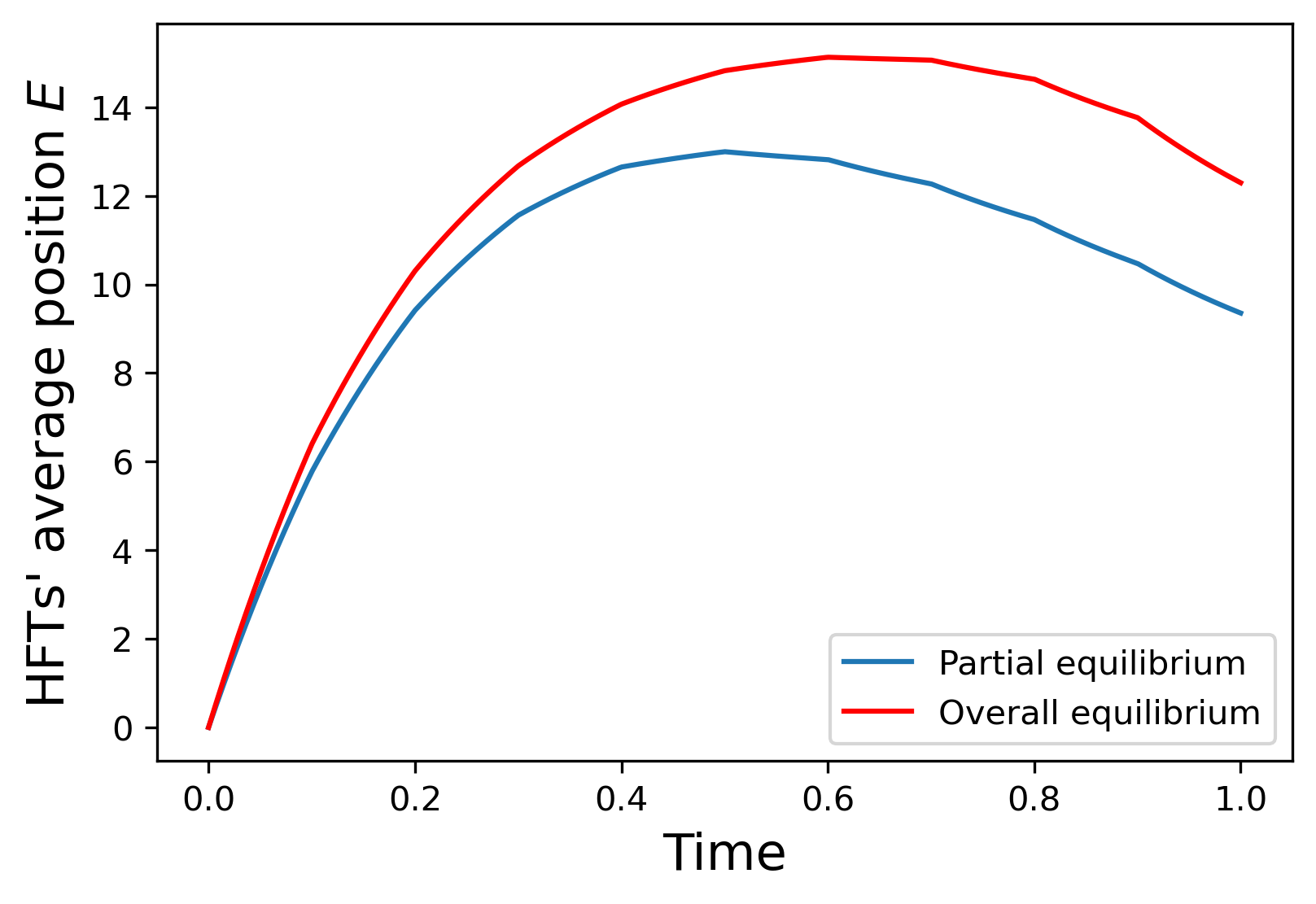}
    }
 \subcaptionbox{$\Gamma=2$}{
    \includegraphics[width = 0.27\textwidth]{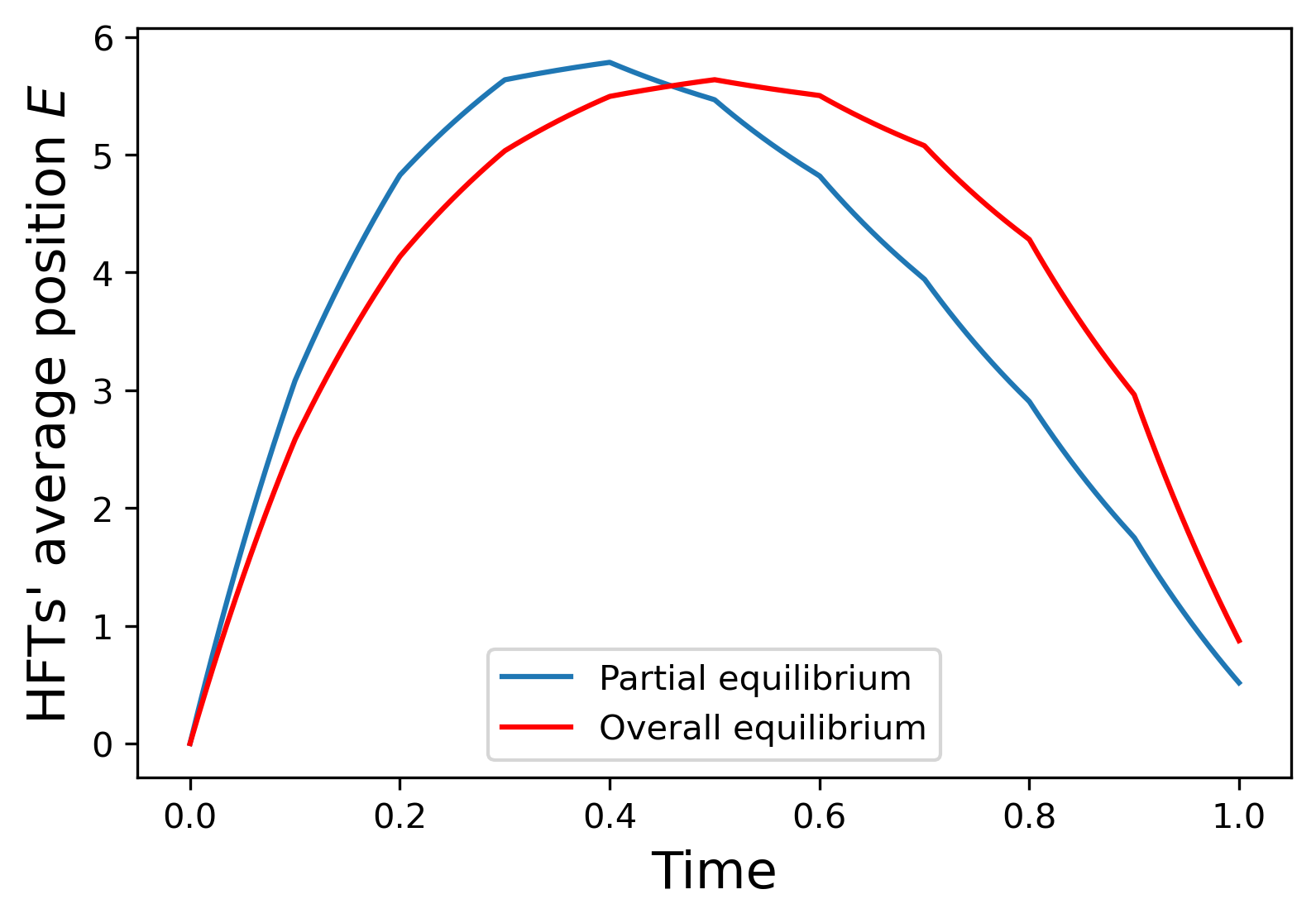}
    }   
    \caption{$\phi=0,$ HFTs' average position $E$.}
    \label{fignojumpLTgamephi=0HFT}
\end{figure}

Comparing LT's strategy with different $\phi,$ as shown in Figure \ref{fignojumpLTgamegam=0LT}, we find that a larger $\phi$ makes LT's strategy closer to uniform trading. It is because a larger $\phi$ leads HFT to be a running Round-Tripper: they constantly buy and sell around large orders. In a sense, it stabilizes the price of each transaction of LT.

\begin{figure}[!htbp]
    \centering
\subcaptionbox{$\phi=0$}{
    \includegraphics[width = 0.27\textwidth]{fig/nojump/LTgame/phi0Gam0LT.png}
    }
\subcaptionbox{$\phi=1$}{
    \includegraphics[width = 0.27\textwidth]{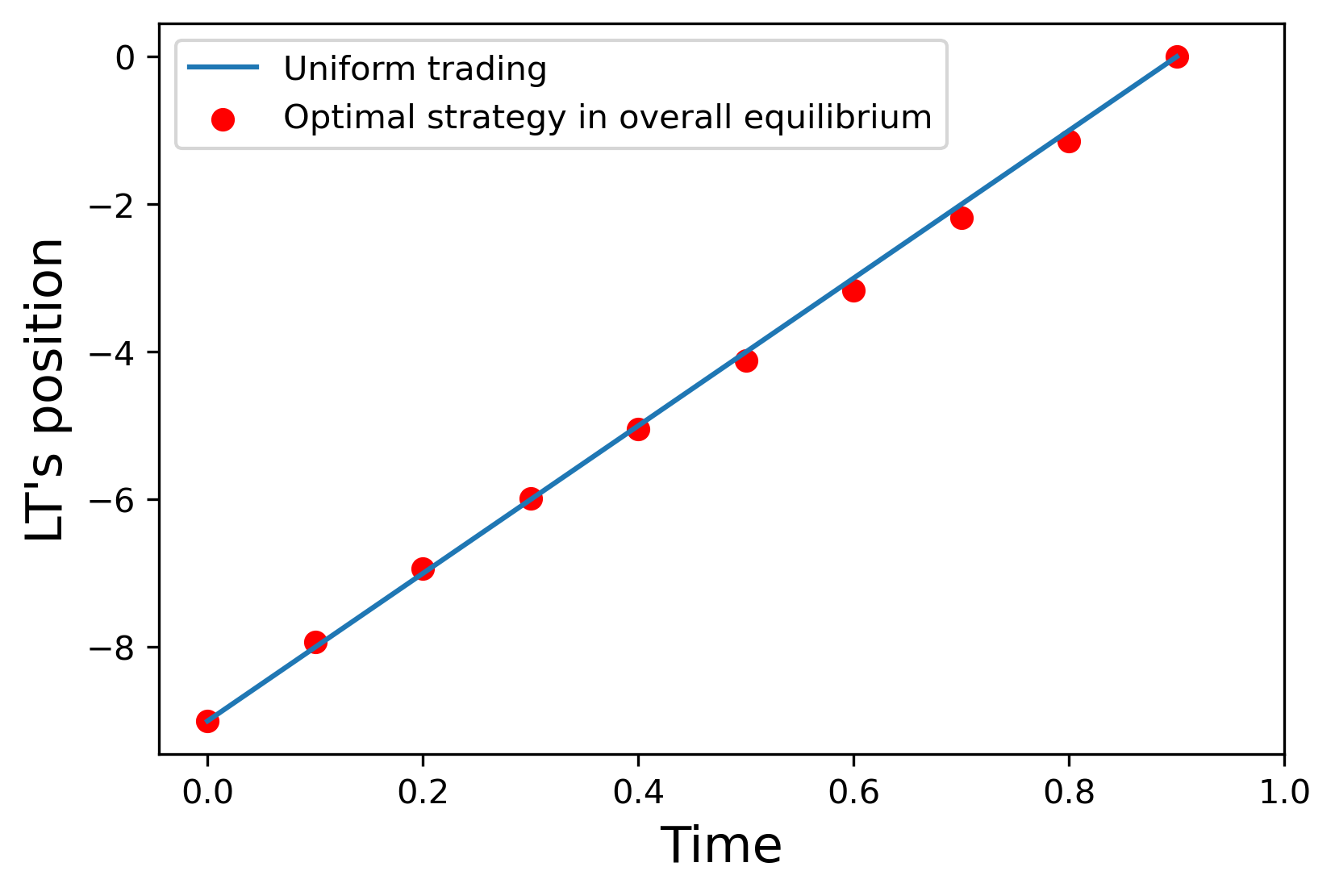}
    }
 \subcaptionbox{$\phi=5$}{
    \includegraphics[width = 0.27\textwidth]{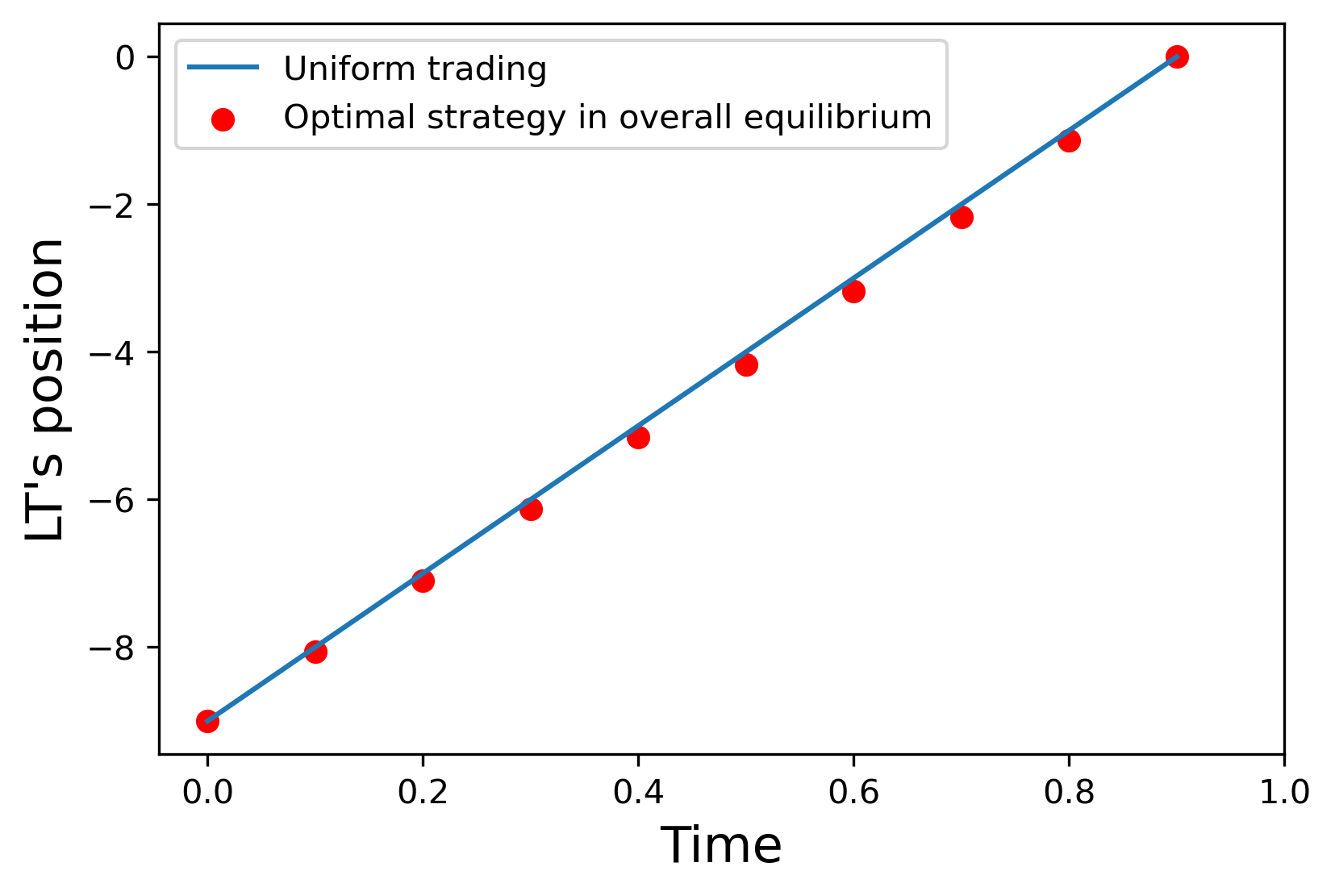}
    }   
    \caption{$\Gamma=0,$ LT's strategy.}
    \label{fignojumpLTgamegam=0LT}
\end{figure}  

\begin{figure}[!htbp]
    \centering
\subcaptionbox{$\phi=0$}{
    \includegraphics[width = 0.27\textwidth]{fig/nojump/LTgame/phi0Gam0HFT.png}
    }
\subcaptionbox{$\phi=1$}{
    \includegraphics[width = 0.27\textwidth]{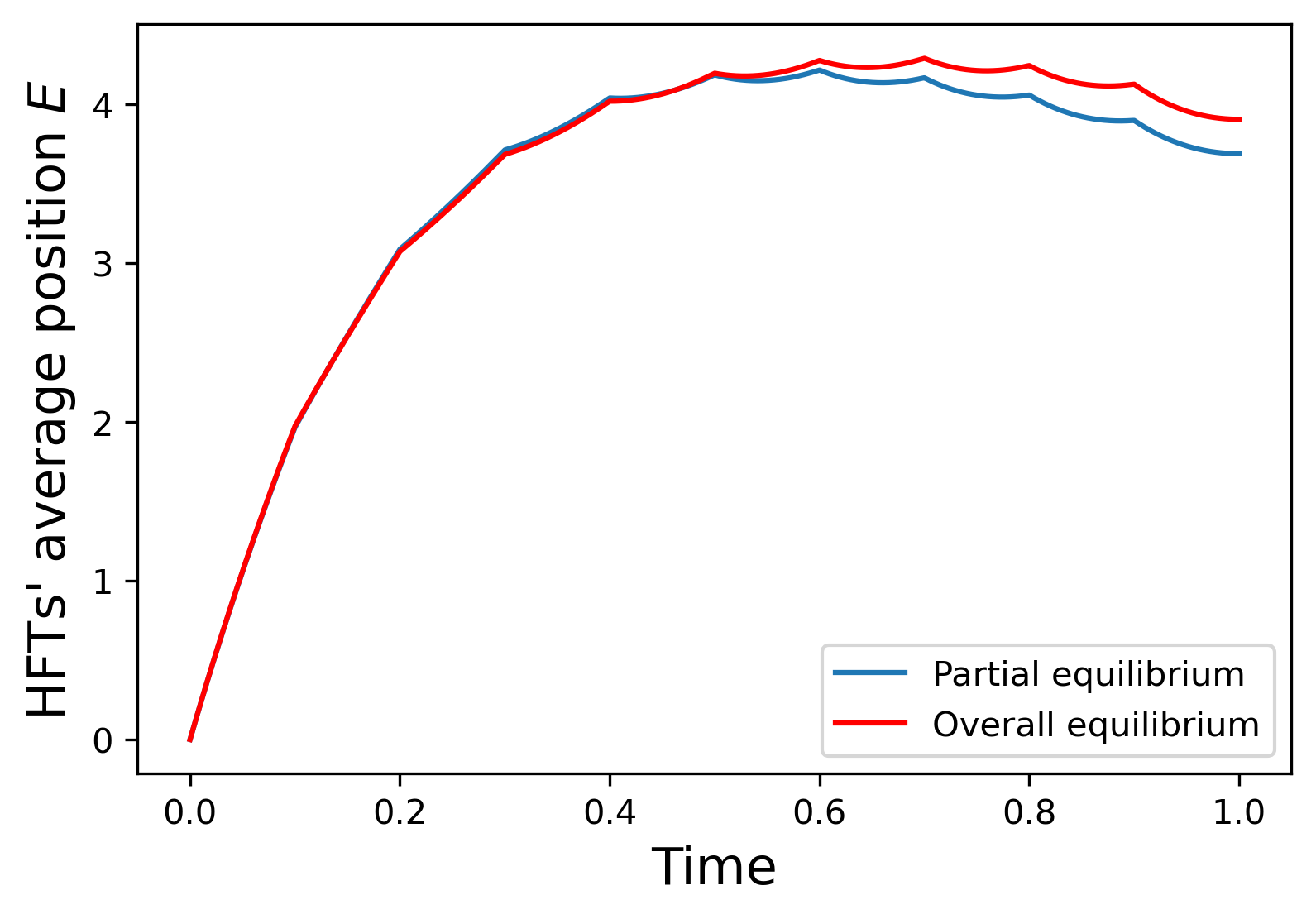}
    }
 \subcaptionbox{$\phi=5$}{
    \includegraphics[width = 0.27\textwidth]{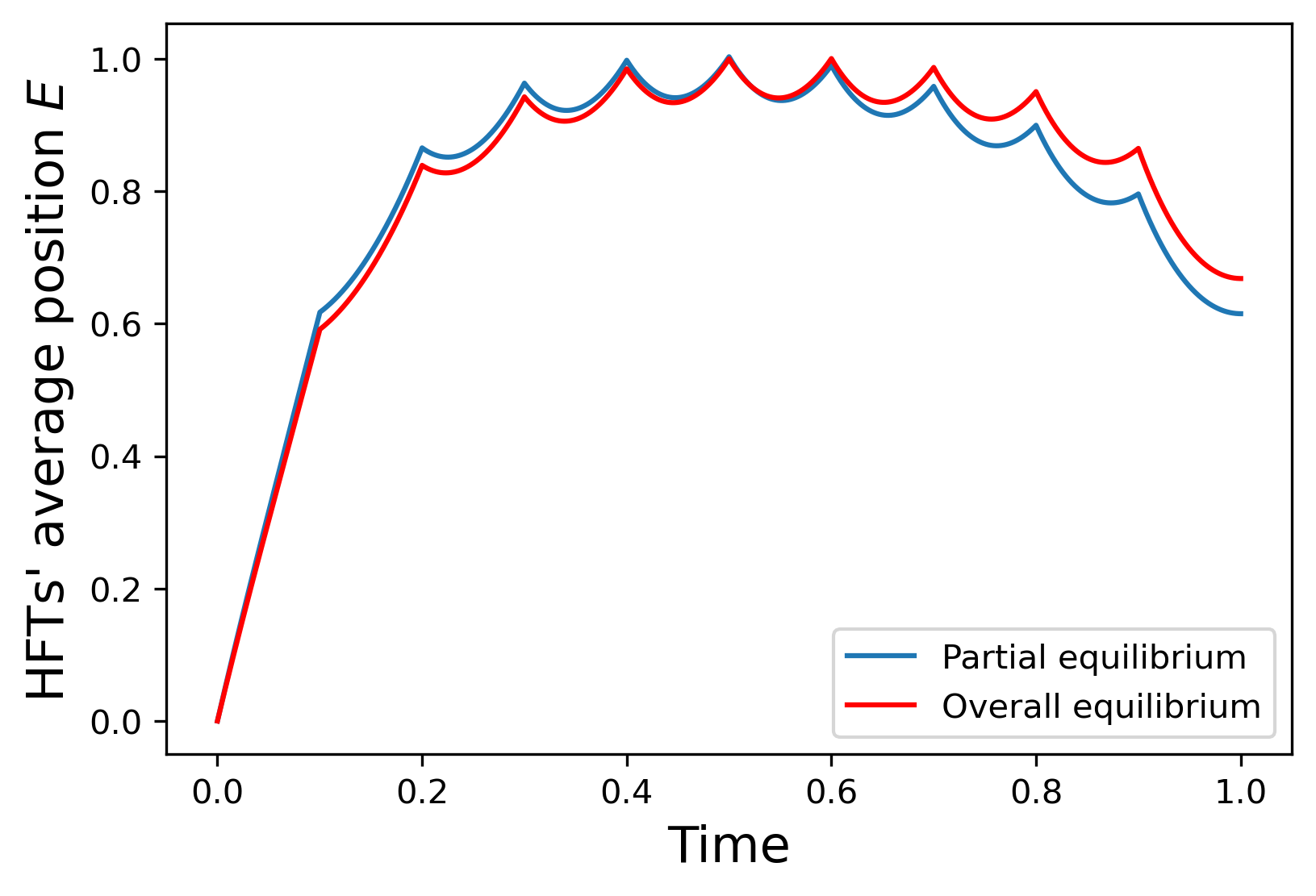}
    }   
    \caption{$\Gamma=0,$ HFTs' average position $E$.}
    \label{fignojumpLTgamegam=0HFT}
\end{figure}

In Figure \ref{fignojumpLTgameLTprofitdiff}, we find that the conclusions about HFTs' impact on LT's profit hold: LT's profit is higher with HFTs if the crowd of HFT plays the role of Round-Tripper and the temporary impact $\lambda^H$ is relatively large to the permanent impact $\gamma^H$.

\begin{figure}[!htbp]
    \centering
    \includegraphics[width = 0.35\textwidth]{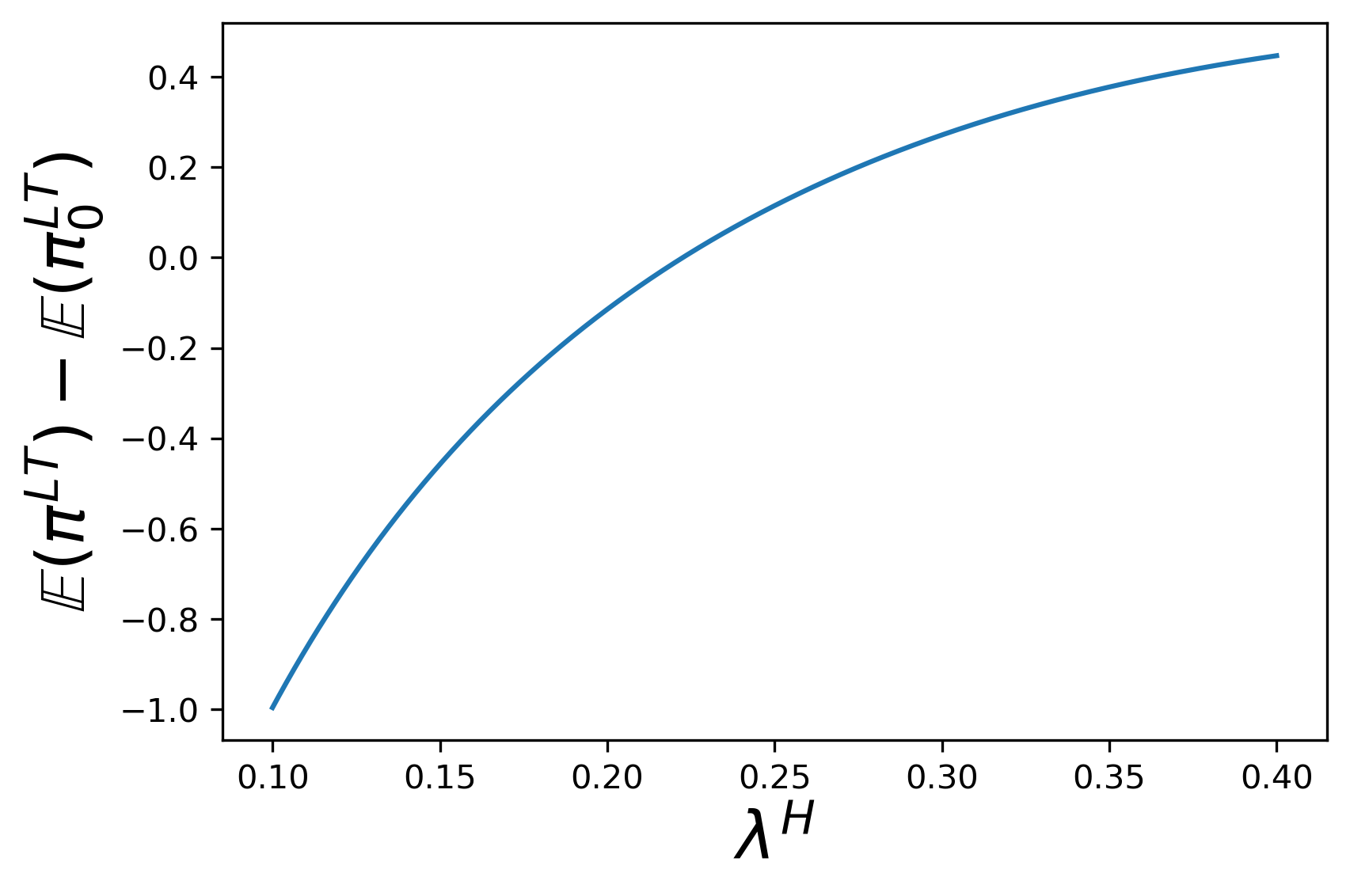}
    \caption{Difference of LT's profit when HFTs act as Round-Tripper ($\phi=10,\Gamma=2$).}
    \label{fignojumpLTgameLTprofitdiff}
\end{figure}

\newpage
In the general case, when $\phi(1)=0,\Gamma(1)=2$ and $\phi(2)=10,\Gamma(2)=0,$ a larger $x$ implies that the crowd of HFT is more averse to running positions. As a result, LT's optimal strategy comes close to uniform trading.

     \begin{figure}[!htbp]
    \centering
    \subcaptionbox{$x=0,y=0$}{
    \includegraphics[width = 0.27\textwidth]{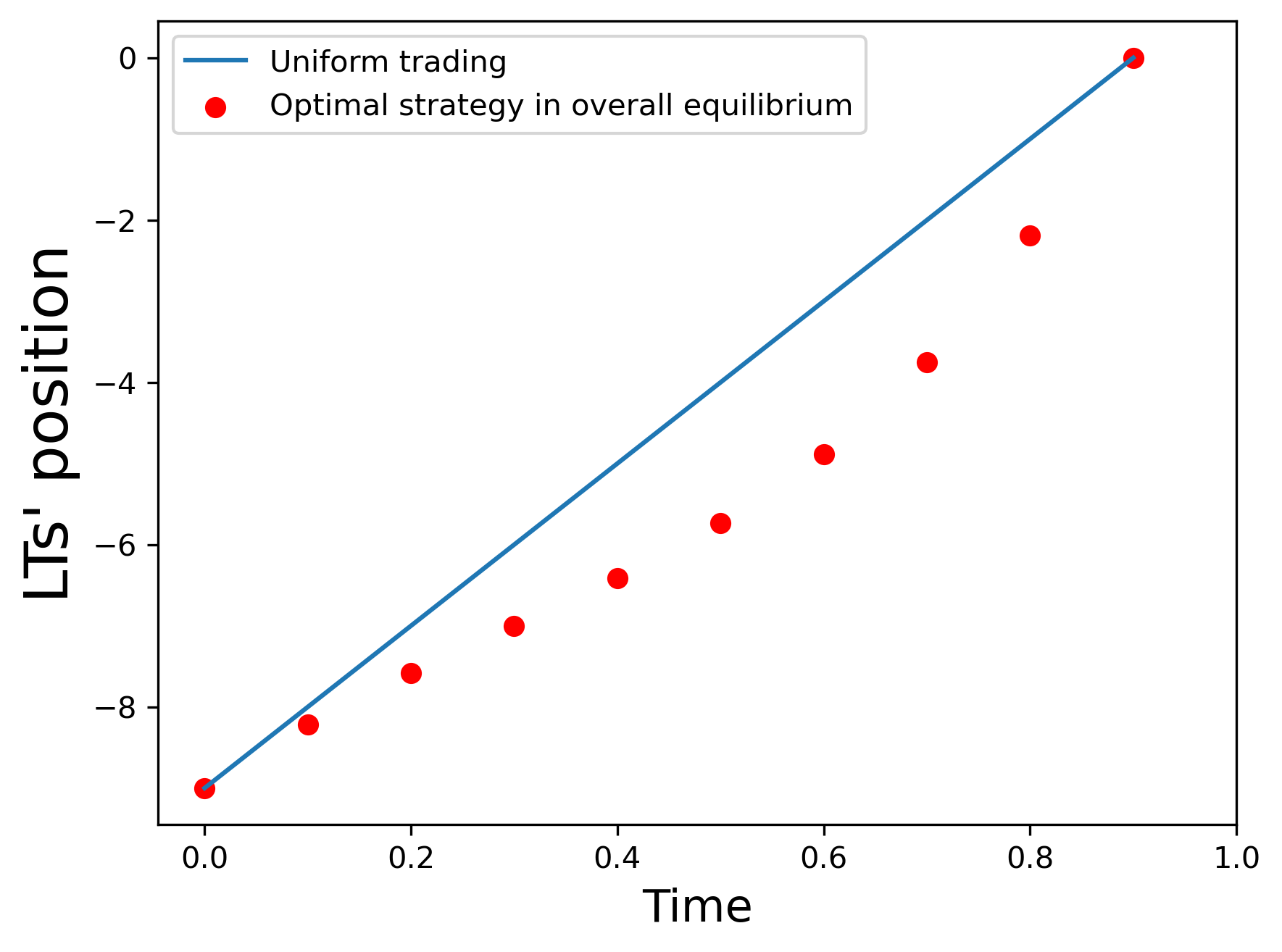}
    }
\subcaptionbox{$x=0.2,y=0.8$}{
    \includegraphics[width = 0.27\textwidth]{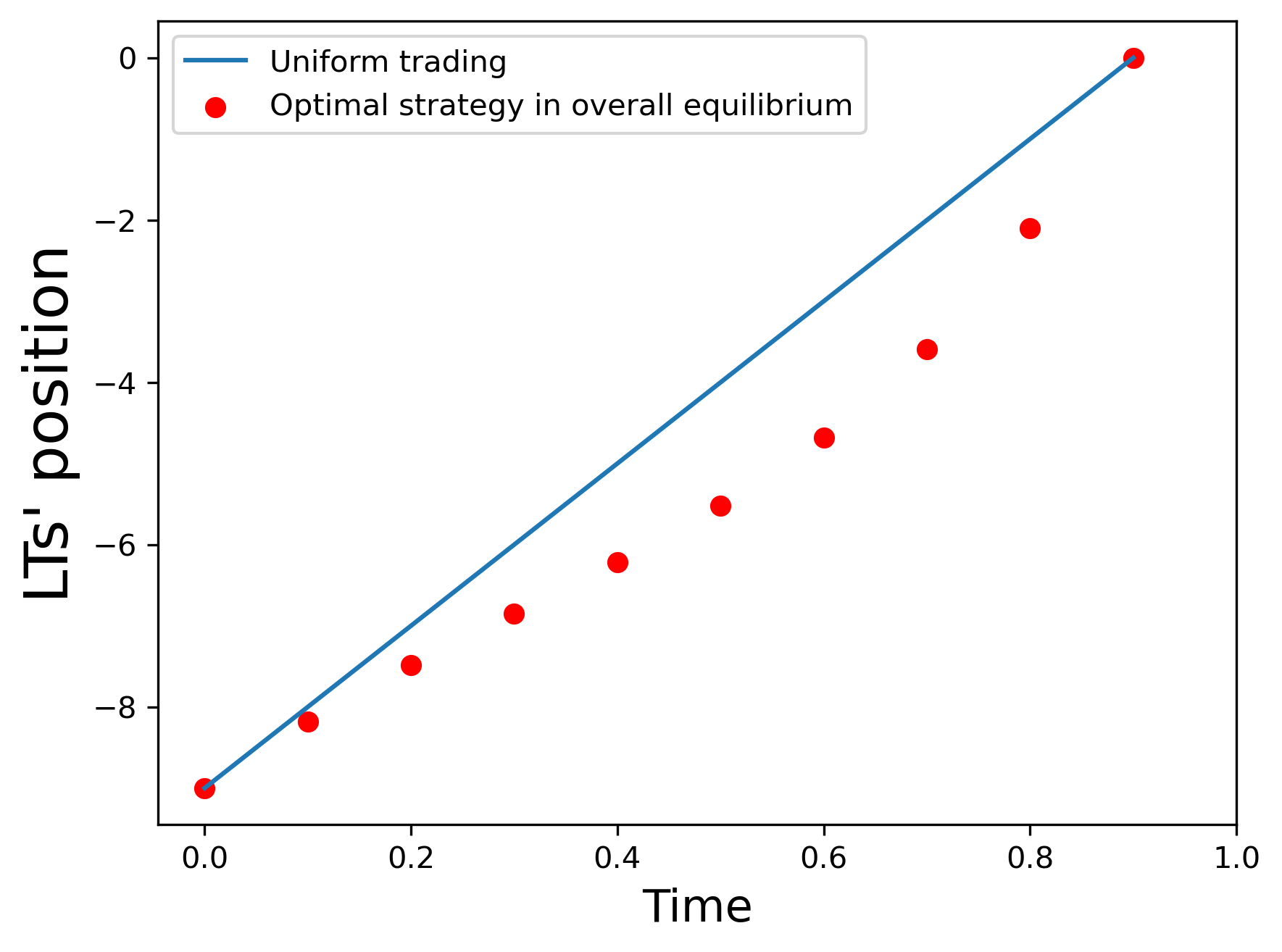}
    }
    
\subcaptionbox{$x=0.5,y=0.5$}{
    \includegraphics[width = 0.27\textwidth]{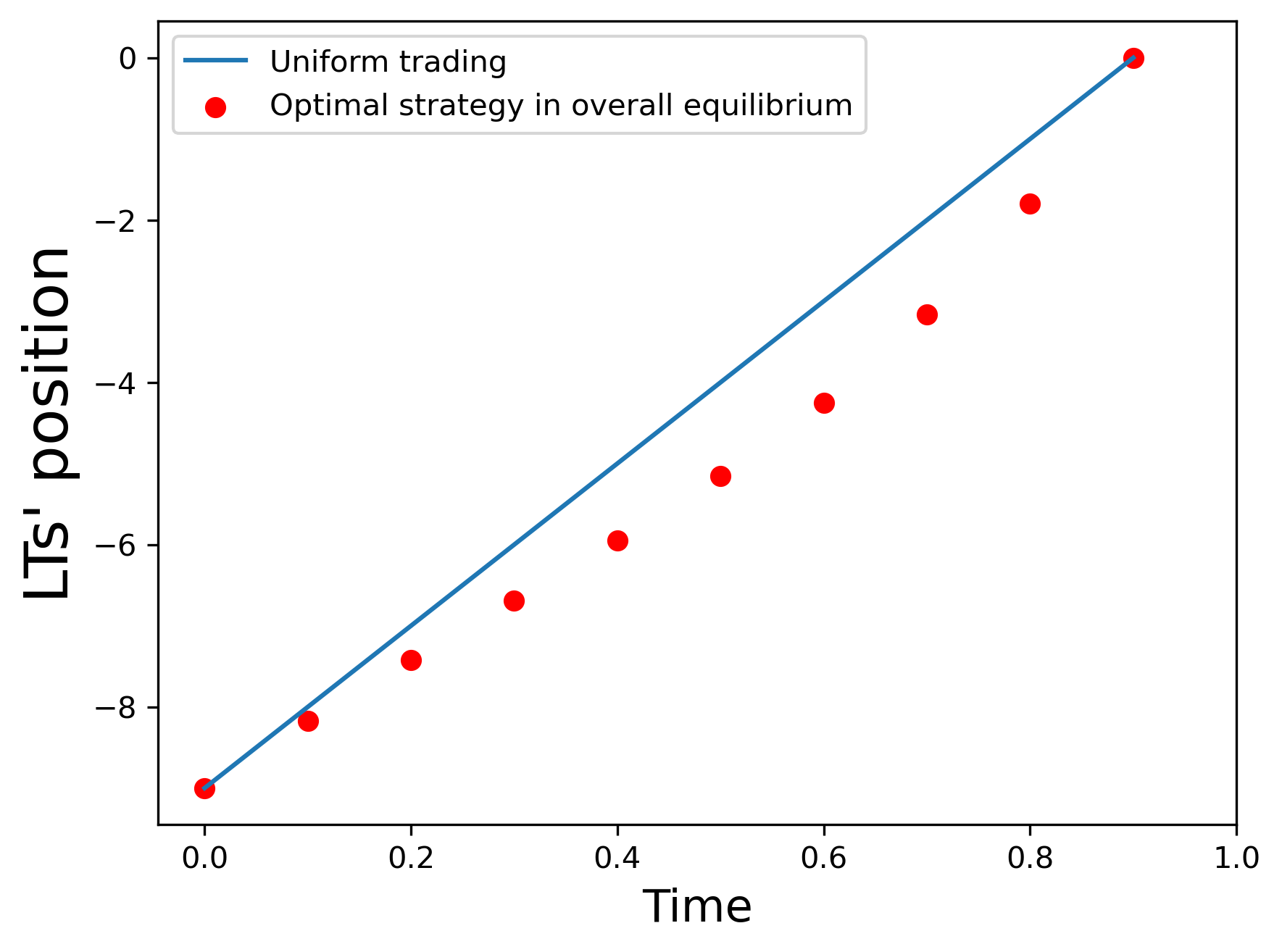}
    }
 \subcaptionbox{$x=0.8,y=0.2$}{
    \includegraphics[width = 0.27\textwidth]{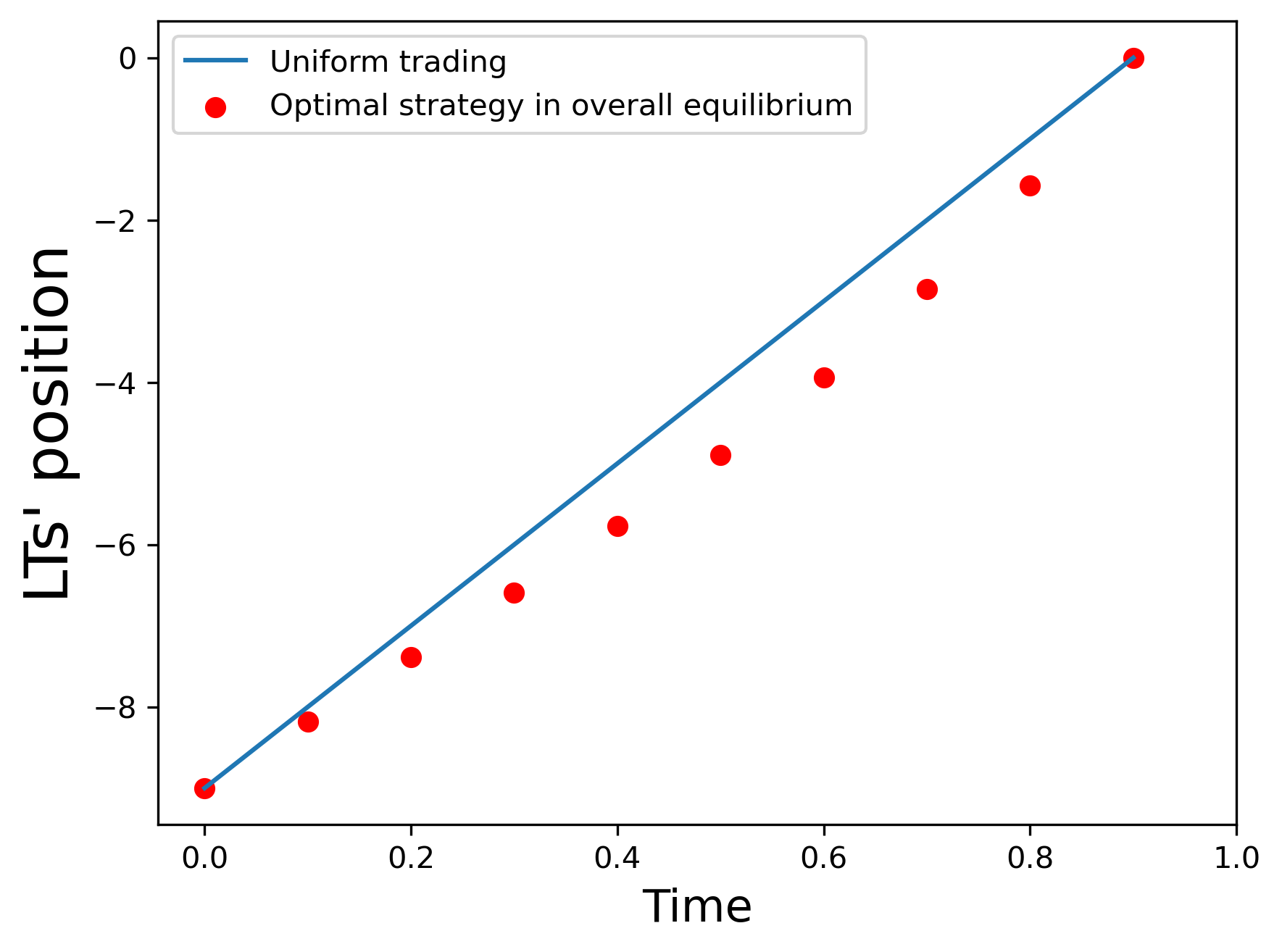}
    }   
\caption{$\phi(1)=0,\Gamma(1)=2,\phi(2)=10,\Gamma(2)=0,$ LT's strategy.}
\label{jumpLTgameLT}
\end{figure}  

    \begin{figure}[!htbp]
    \centering
    \subcaptionbox{$x=0,y=0$}{
    \includegraphics[width = 0.27\textwidth]{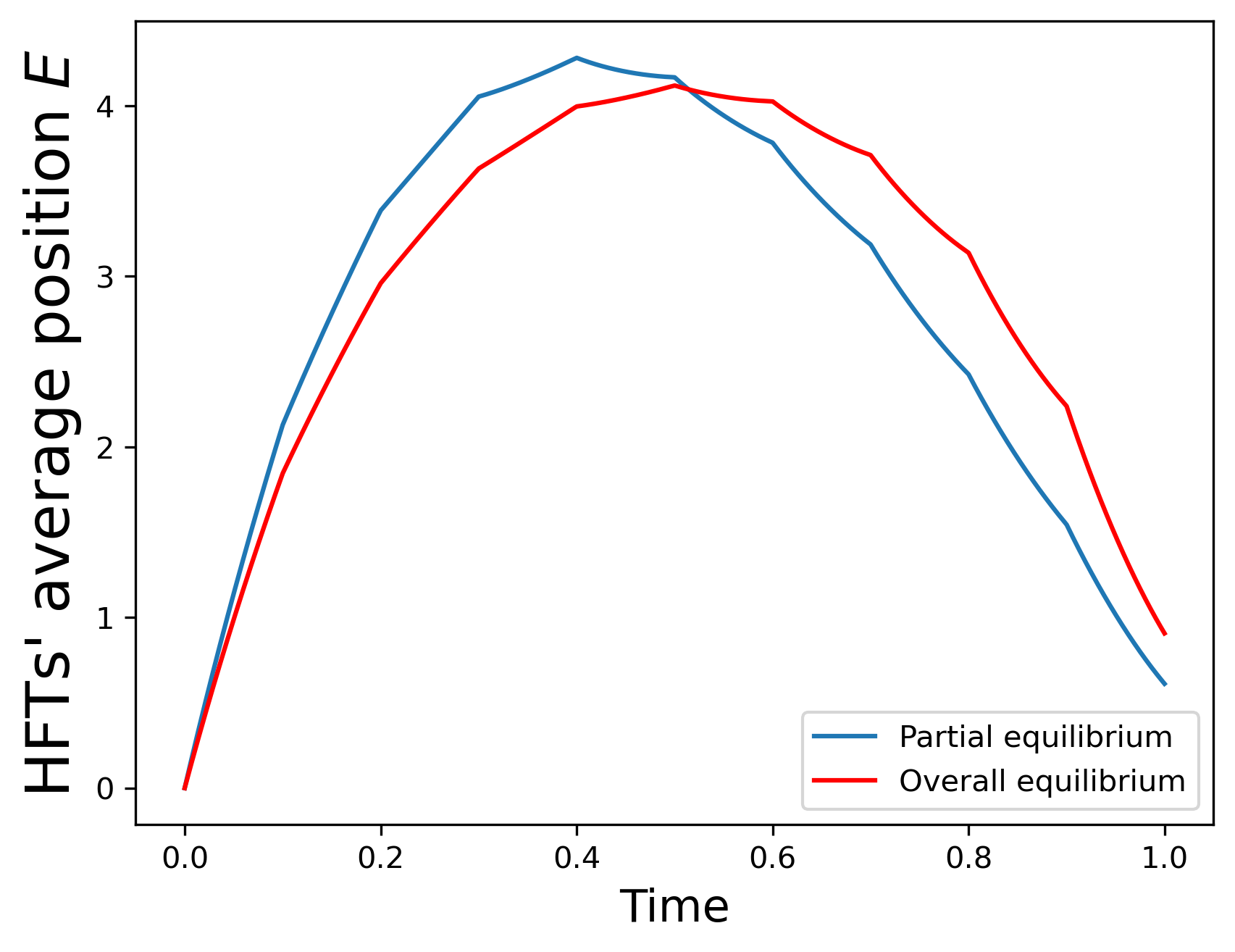}
    }
\subcaptionbox{$x=0.2,y=0.8$}{
    \includegraphics[width = 0.27\textwidth]{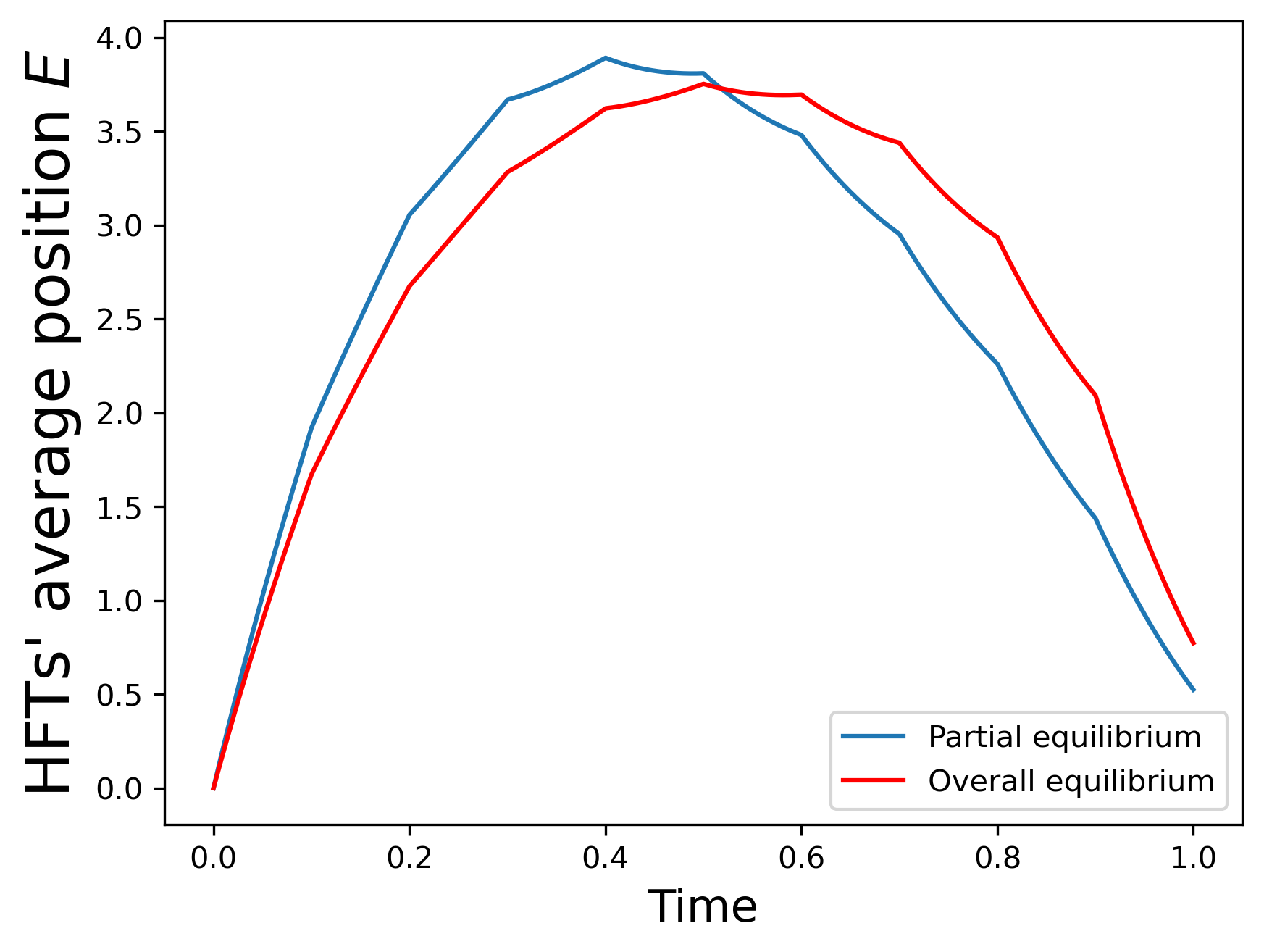}
    }
    
\subcaptionbox{$x=0.5,y=0.5$}{
    \includegraphics[width = 0.27\textwidth]{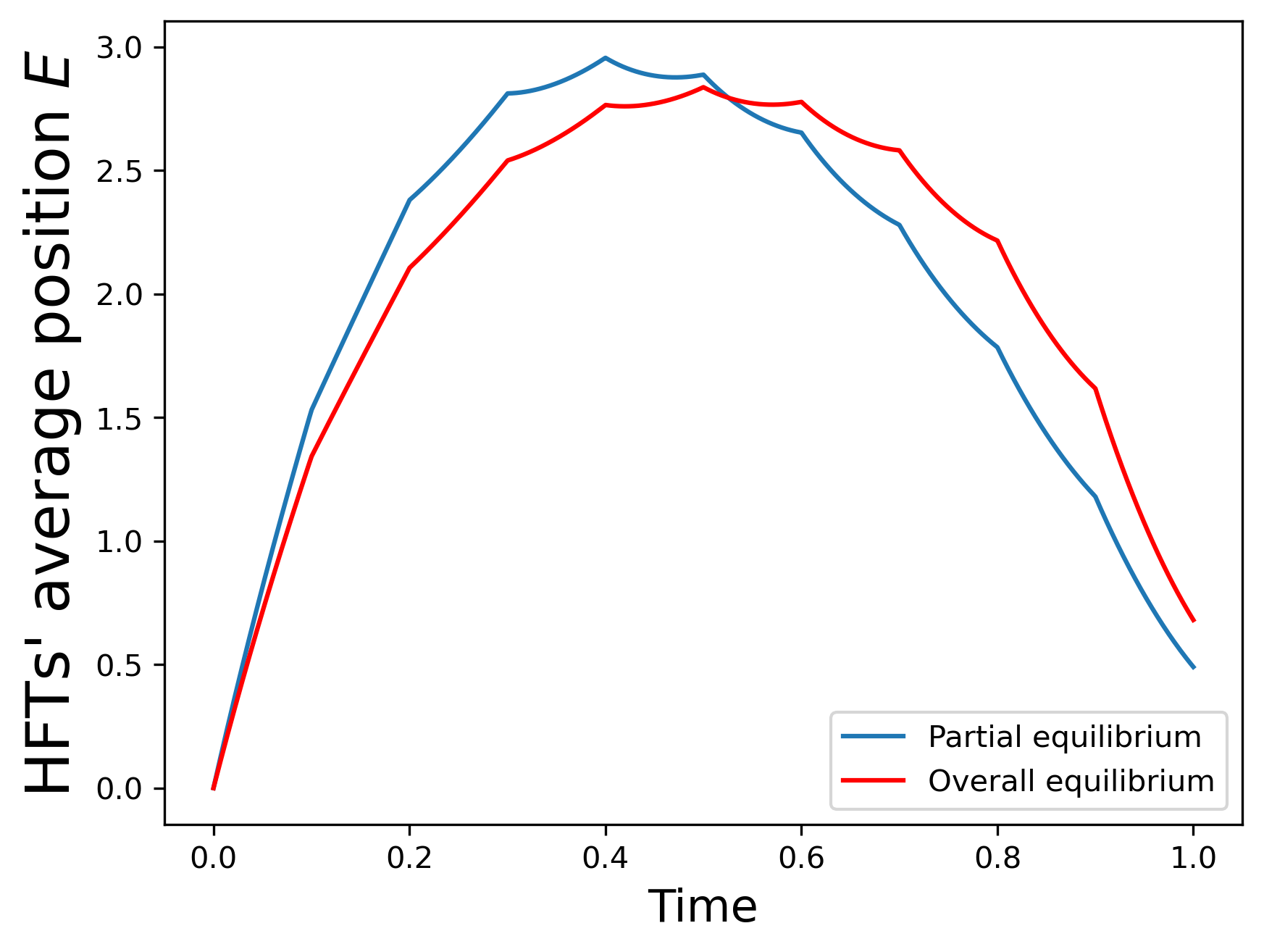}
    }
 \subcaptionbox{$x=0.8,y=0.2$}{
    \includegraphics[width = 0.27\textwidth]{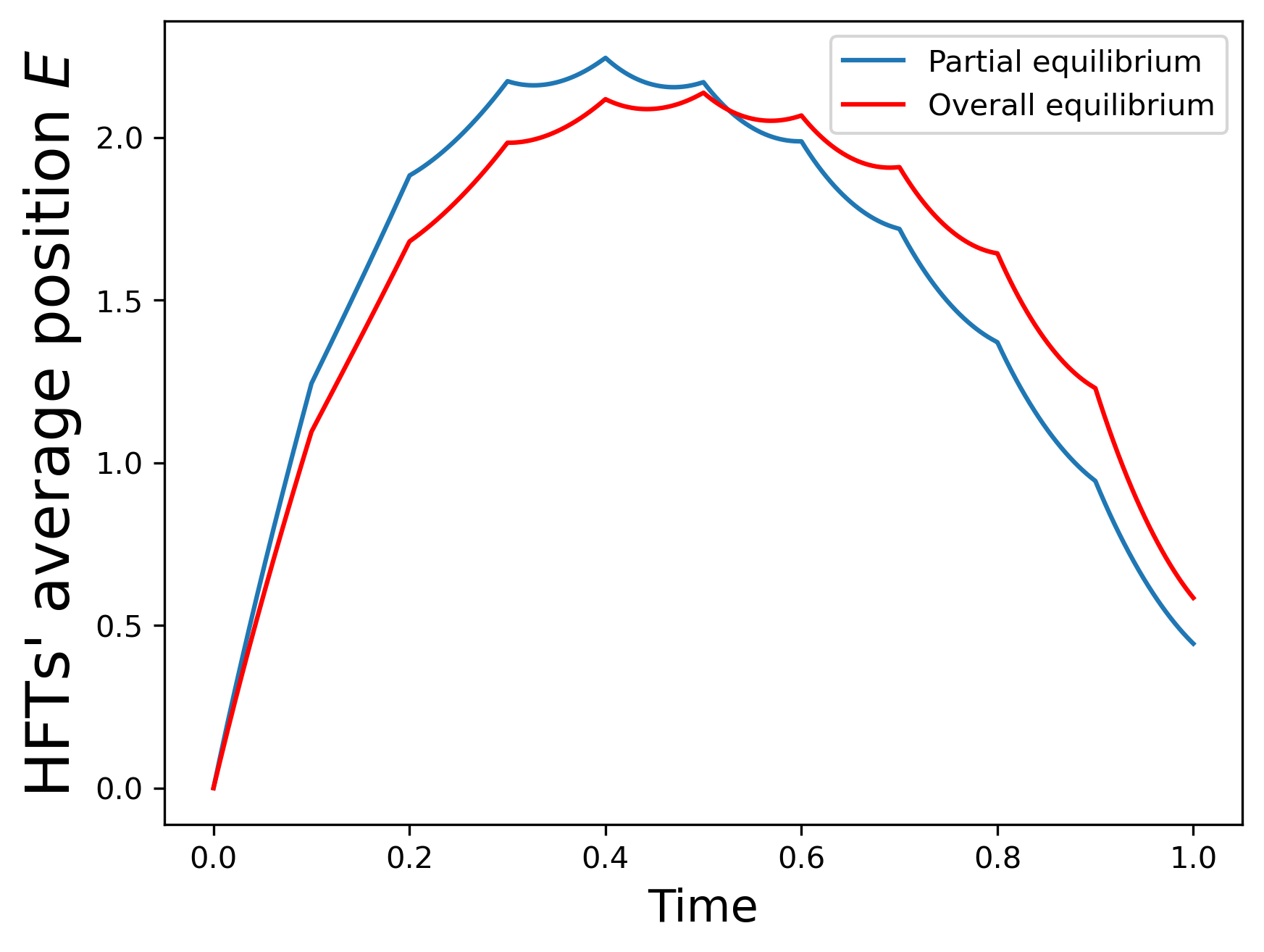}
    }   
\caption{$\phi(1)=0,\Gamma(1)=2,\phi(2)=10,\Gamma(2)=0,$ HFTs' average position $E$.}
\label{jumpLTgameHFT}
\end{figure}

\newpage
\section{Conclusion}
In this paper, we propose a model for studying the interaction between a large trader (LT) and a large population of high-frequency traders (HFTs) who predict and exploit LT's transactions. They affect each other not only through permanent price impact but also temporary one. HFTs are subject to different inventory-averse levels, and the levels may change overtime during the trading period, characterized by independent jump processes. LT maximizes her revenue from the liquidation of her position, while HFTs also care about the inventory risk in addition to the revenue. LT is slower in the sense that instead of trading continuously like HFTs, she only executes orders at discrete time points, and her trading strategy is deterministic.

Since the solution becomes intractable when the population size of HFTs increases, we adopt the idea from mean field game theory such that HFTs affect the market and the LT only through the average trading speed, while the individual trading activity only has an impact on their own trading fee. For a given trading strategy of LT, we solve the partial Nash equilibrium among HFTs by the forward-backward system given by the HJB equation and the Kolmogorov forward equation. To address the discontinuities in the price dynamic caused by LT, we divide the whole trading period into sub-intervals partitioned by the trading time of LT, solve the system on each sub-interval and adjust at the endpoints. When there are more than one inventory-averse levels, by utilizing the linearity of solution with respect to the initial value and the LT's strategy, we are able to efficiently solve the system of coupled equations. The LT's optimal strategy given HFTs' is found by solving a linear system, and we reach the overall Nash equilibrium between LT and HFTs. The property of approximate Nash equilibrium is satisfied when the number of HFTs is large enough. 

Numerical results show that when the HFTs are averse to terminal position, they will act as a Round-Tripper during the whole period, first trade in the same direction as LT and consume liquidity, and then in the opposite direction and provide liquidity, in which case LT will delay the liquidation towards the end of the period. When the HFTs are averse to running inventory, within each sub-interval, they will first trade in the opposite direction and then in the same direction, which is in line with the observation of van Kervel and Menkveld (2019) \cite{van2019high}. Moreover, the repeated liquidity consuming-supply makes the behavior of LT close to uniform trading. When the temporary impact is relatively large compared to the permanent impact, LT benefits from the existence of HFTs.

Further extensions may include introducing execution risk for both LT and HFTs, where LT is not only subject to the uncertainty of executed amount but also execution time, and studying how HFTs learn and predict the behavior of LT.

\newpage
\section*{Appendix}
\label{appendix}
\appendix
\section{Proof of Theorem \ref{thm-verification}}
\label{proof-thm-verification}
We first prove the case $t\in[t_K,T]$. Apply It\^{o}'s lemma on $h(t,X(t),Y(t))$ (Theorem 7.1 in \cite{Hanson2007}),
\begin{equation*}
\begin{aligned}
h(T,X(T),Y(T)) &= h(t,X(t),Y(t))\\
&+\int_t^T \left\{\frac{\partial h}{\partial s} ds + \frac{\partial h}{\partial x}v(s) ds + \sum_{j\in\mathbb{S}, j\neq Y(s)}\left[ h(s, X(s), j) - h(s, X(s),Y(s)  \right]\mathcal{P}(ds,j;Y(s))\right\},
\end{aligned}
\end{equation*}
where $\mathcal{P}(ds,j;Y(s))$ is the Poisson random measure with
\begin{equation*}
\mathcal{P}(ds,j;Y(s))\sim Poisson(Q^{Y(s),j} ds)
\end{equation*}
and
\begin{equation*}
\hat{\mathcal{P}}(ds,j;Y(s)) = \mathcal{P}(ds,j;Y(s))-Q^{Y(s),j}ds
\end{equation*}
is the compensated zero-mean Poisson random measure.

For $v\in\mathcal{A}_{HFT},$ there exists $M_1,M_2$ such that $P(X(t)\in[M_1,M_2])=1.$ So when $t\in[t_K,T],$ $h(t,X(t),Y(t))$ is bounded, which implies that
$$M_t=\int_{t_K}^t \sum_{j\in\mathbb{S}, j\neq Y(s)}\left[ h(s, X(s), j) - h(s, X(s),Y(s)  \right]\hat{\mathcal{P}}(ds,j;Y(s)),t\in[t_K,T]$$
is actually a martingale. Hence,
\begin{equation*}
\mathbbm{E}\left[\left. \int_t^T \sum_{j\in\mathbb{S}, j\neq Y(s)}\left[ h(s, X(s), j) - h(s, X(s),Y(s)  \right]\hat{\mathcal{P}}(ds,j;Y(s)) \right\vert \mathcal{F}_t  \right] = 0
\end{equation*}
and
\begin{equation*}
\begin{aligned}
&\mathbbm{E}\left[h(T,X(T),Y(T)) | \mathcal{F}_t \right] \\
= &h(t,X(t),Y(t)) + \mathbbm{E}\left[\left. \int_t^T  \left\{\frac{\partial h}{\partial s} ds + \frac{\partial h}{\partial x}v(s) ds + \sum_{j\in\mathbb{S}, j\neq Y(s)}\left[ h(s, X(s), j) - h(s, X(s),Y(s)  \right] Q^{Y(s),j} ds\right\}\right \vert \mathcal{F}_t \right] \\
= & h(t,X(t),Y(t)) + \mathbbm{E}\left[\left. \int_t^T  \left\{\frac{\partial h}{\partial s} ds + \frac{\partial h}{\partial x}v(s) ds + \sum_{j\in\mathbb{S}} h(s, X(s), j) Q^{Y(s),j} ds\right\}\right \vert \mathcal{F}_t \right] \\
\leq & h(t,X(t),Y(t)) + \mathbbm{E}\left[\left. \int_t^T \left\{ -\gamma^H \mu(s)X(s) + \phi(Y(s))X(s)^2 + \eta v(s)^2 + \lambda^H v(s)\mu(s)  \right\} ds \right\vert \mathcal{F}_t \right],
\end{aligned}
\end{equation*}
where the last inequality is due to the HJB equation.
So 
\begin{equation*}
h(t,X(t),Y(t))\geq \mathbbm{E}\left[\left. \int_t^T \left\{ \gamma^H \mu(s)X(s) - \phi(Y(s))X(s)^2 - \eta v(s)^2 - \lambda^H v(s)\mu(s)  \right\} ds  -\Gamma(Y(T))X(T)^2 \right\vert \mathcal{F}_t\right]
\end{equation*}
for any admissible control $v$, thus
\begin{equation*}
h(t,X(t),Y(t))\geq \max_{v}\mathbbm{E}\left[\left. \int_t^T \left\{ \gamma^H \mu(s)X(s) - \phi(Y(s))X(s)^2 - \eta v(s)^2 - \lambda^H v(s)\mu(s)  \right\} ds  -\Gamma(Y(T))X(T)^2 \right\vert \mathcal{F}_t\right].
\end{equation*}

On the other side,
\begin{equation*}
\begin{aligned}
h(t,X(t),Y(t))&=\mathbbm{E}\left[\left. \int_t^T \left\{ \gamma^H \mu(s)X(s) - \phi(Y(s))X(s)^2 - \eta v^*(s)^2 - \lambda^H v^*(s)\mu(s)  \right\} ds  -\Gamma(Y(T))X(T)^2 \right\vert \mathcal{F}_t\right] \\
&\leq \max_{v}\mathbbm{E}\left[\left. \int_t^T \left\{ \gamma^H \mu(s)X(s) - \phi(Y(s))X(s)^2 - \eta v(s)^2 - \lambda^H v(s)\mu(s)  \right\} ds  -\Gamma(Y(T))X(T)^2 \right\vert \mathcal{F}_t\right].
\end{aligned}
\end{equation*}
The conclusion follows from the Markovian property.

So we have proved that
\begin{equation*}
\begin{aligned}
h(t_K,x,i)=\max_{v}\mathbbm{E}\Bigg[& \int_{t_K}^T \left\{ \gamma^H \mu(s)X(s) - \phi(Y(s))X(s)^2 - \eta v(s)^2 - \lambda^H v(s)\mu(s)  \right\} ds \\
&-\Gamma(Y(T))X(T)^2 \Bigg\vert X(t_K)=x,Y(t_K)=i\Bigg].
\end{aligned}
\end{equation*}
What's more, \eqref{HJB} tells that $h(t_{K-},x,i)=h(t_K,x,i)+\gamma\xi_K x$, then
\begin{equation*}
\begin{aligned}
  h(t_{K-},x,i)=
  \max_{v}\mathbbm{E}\Bigg[&\gamma\xi_K x(t_K)+ \int_{t_{K-}}^T \left\{ \gamma^H \mu(s)X(s) - \phi(Y(s))X(s)^2 - \eta v(s)^2 - \lambda^H v(s)\mu(s)  \right\} ds \\
  &-\Gamma(Y(T))X(T)^2 \Bigg\vert X(t_{K-})=x,Y(t_{K-})=i\Bigg].   
\end{aligned}
\end{equation*}
$X(t_{K-})=X(t_K)$ since $X$ is continuous; $Y(t_{K-})=Y(t_K)$ since the probability that $Y$ jumps at a certain time point $t_K$ is zero.

During $t\in[t_{K-1},t_{K}),$ by similar arguments used in $[t_K,T]$ (It\^{o}'s lemma and HJB equation), we can prove the desired result.
Remaining intervals $[t_{k-1},t_k),k=1,..,K-1,$ can be proved similarly. Therefore, $H$ is the value function.

\section{Proof of Lemma \ref{lemma-forward-eqn}}
\label{proof-lemma-forward-eqn}
For any smooth testing function $f$,
\begin{equation*}
\frac{d}{dt}\mathbbm{E}[f(X(t),Y(t))]=\lim_{\Delta\rightarrow0}\frac{1}{\Delta}\mathbbm{E}\left[f(X(t+\Delta),Y(t+\Delta))-f(X(t),Y(t))\right].    
\end{equation*}
\begin{equation*}
\begin{aligned}
\text{LHS}&=\frac{d}{dt}\left[\sum_ip_i(t)\int f(x,i)m_t^i(x)dx\right]\\
&=\sum_i\frac{d}{dt}p_i(t)\int f(x,i)m_t^i(x)dx+\sum_ip_i(t)\int f(x,i)\frac{\partial}{\partial t}m_t^i(x)dx.\\
\text{RHS}&=\lim_{\Delta\rightarrow0}\frac{1}{\Delta}\mathbbm{E}\{\mathbbm{E}[f(X(t+\Delta),Y(t+\Delta))-f(X(t),Y(t))|\mathcal{F}_t]\}\\
&=\mathbbm{E}\left \{\lim_{\Delta\rightarrow0}\frac{1}{\Delta}\mathbbm{E}[f(X(t+\Delta),Y(t+\Delta))-f(X(t),Y(t))|\mathcal{F}_t] \right\}\\
&=\mathbbm{E}\left \{\lim_{\Delta\rightarrow0}\frac{1}{\Delta}\mathbbm{E}[f(X(t+\Delta),Y(t+\Delta))-f(X(t),Y(t))|X(t),Y(t)] \right\}\\
&=\mathbbm{E}\left \{\lim_{\Delta\rightarrow0}\frac{1}{\Delta}\left [\frac{\partial f(X(t),Y(t))}{\partial x}v(t, X(t),Y(t))\Delta +\sum_j\left(p_{Y(t),j}(\Delta)-\delta_{Y(t),j}\right) f(X(t),j)\right]\right\}\\
&=\mathbbm{E}\left[\frac{\partial f(X(t),Y(t))}{\partial x}v(t, X(t),Y(t))+\sum_j Q^{Y(t),j}f(X(t),j) \right]\\
&=\sum_ip_i(t)\int \frac{\partial f(x,i)}{\partial x} v(t, x,i)m_t^i(x)dx+\sum_i p_i(t)\sum_j Q^{ij}\int  f(x,j) m_t^i(x)dx\\
&=-\sum_i p_i(t)\int  f(x,i)\partial_x[v(t,x,i)m_t^i(x)]dx+\sum_i\sum_j p_j(t) Q^{ji}\int  f(x,i) m_t^j(x)dx.\\
\end{aligned}
\end{equation*}
$\text{LHS}=\text{RHS}$ for any smooth $f$ implies the result.

\section{Proof of Lemma \ref{lemma-dE}}
\label{proof-lemma-dE}
It should be noted that the equations for densities tell us that
\begin{equation}
\label{Eiandmui}
\begin{aligned}
\frac{dE_i(t)}{dt}&=\int  x\partial_t m_t^i(x)dx\\
&=\int  x\left\{-\partial_x[v(t,x,i)m_t^i(x)]+\frac{1}{p_i(t)}\sum_j p_j(t) Q^{ji}[m_t^j(x)-m_t^i(x)]\right\} dx\\
&=\int  v(t, x,i)m_t^i(x)dx+\frac{1}{p_i(t)}\sum_j p_j(t) Q^{ji}\int  x[m_t^j(x)-m_t^i(x)]dx\\
&=\mu_i(t)+\frac{1}{p_i(t)}\sum_j p_j(t) Q^{ji}[E_j(t)-E_i(t)],
\end{aligned}
\end{equation}
which gives the result.

\section{Proof of Lemma \ref{lemma-Eboundary}}
\label{proof-lemma-Eboundary}
Integrate each side of \eqref{optcontrol2} w.r.t. $m_t^i(x)$,
\begin{equation}
 \label{mui}   
 \mu_i(t)=\frac{1}{2\eta}[h_i^1+2h_i^2 E_i(t)-\lambda^H\mu(t)].
\end{equation}
Then
\begin{equation*}
\begin{aligned}
\mu(t)=&\sum_i p_i(t)\mu_i(t)\\
=&\frac{1}{2\eta}\sum_ip_i(t)[h_i^1+2h_i^2E_i(t)]-\frac{\lambda^H}{2\eta}\mu(t).   
\end{aligned}
\end{equation*}
Hence,
\begin{equation}
\label{mu}
\mu(t)=\frac{1}{\lambda^H+2\eta}\sum_ip_i(t)[h_i^2+2h_i^2E_i(t)].
\end{equation}
Substitute \eqref{hs-boundary} into \eqref{mu}, \begin{equation*}
\begin{aligned}
\mu(t_{k-})&=\frac{1}{\lambda^H+2\eta}\sum_ip_i(t_{k})[h_i^1(t_k)+\gamma\xi_k+2h_i^2(t_{k})E_i(t_k)]\\
&=\mu(t_k)+\frac{\gamma\xi_k}{\lambda^H+2\eta}.
\end{aligned}
\end{equation*}
Then
\begin{equation*}
\begin{aligned}
\mu_i(t_{k-})&=\frac{1}{2\eta}[h_i^1(t_{k-})+2h_i^2(t_{k-})-\lambda^H\mu(t_{k-})]\\
&=\frac{1}{2\eta}[h_i^1(t_{k})+\gamma\xi_k+2h_i^2(t_{k})-\lambda^H(\mu(t_{k})+\frac{\gamma\xi_k}{\lambda^H+2\eta})]\\
&=\mu_i(t_k)+\frac{\gamma\xi_k}{\lambda^H+2\eta}.
\end{aligned}
\end{equation*}

\section{Proof of Theorem \ref{thm-E}}
\label{proof-thm-E}
As for the evolution of $\boldsymbol{E}$ during each interval $\{[t_k,t_{k+1})\}_{k=0}^{K-1}$, multiply each side of \eqref{mui} by $2\eta$ and differentiate by $t$,
\begin{equation*}
\begin{aligned}
2\eta\frac{d\mu_i(t)}{dt}=&\frac{dh_i^1(t)}{dt}+2E_i(t)\frac{dh_i^2(t)}{dt}+2h_i^2(t)\frac{dE_i(t)}{dt}-\lambda^H\frac{d\mu(t)}{dt}\\
=&-\gamma^H\mu(t)-\sum_j Q^{ij}h_j^1(t)-\frac{h_i^2(h_i^1-\lambda^H\mu(t))}{\eta}+2E_i(t)\frac{dh_i^2(t)}{dt}+2h_i^2(t)\frac{dE_i(t)}{dt}-\lambda^H\frac{d\mu(t)}{dt}\\
=&-\gamma^H\mu(t)-\sum_j Q^{ij}h_j^1(t)-\frac{h_i^2}{\eta}[2\eta \mu_i(t)-2h_i^2E_i(t)]+2E_i(t)\frac{dh_i^2(t)}{dt}+2h_i^2(t)\frac{dE_i(t)}{dt}-\lambda^H\frac{d\mu(t)}{dt}\\
=&-\gamma^H\mu(t)-\sum_j Q^{ij}h_j^1(t)-\frac{h_i^2}{\eta}[2\eta \mu_i(t)-2h_i^2E_i(t)]+2E_i(t)[\phi(i)-\sum_jQ^{ij}h_j^2(t)-\frac{(h_i^2)^2}{\eta}]\\
&+2h_i^2\frac{dE_i(t)}{dt}-\lambda^H\frac{d\mu(t)}{dt}\\
=&-\lambda^H\frac{d\mu(t)}{dt}-\gamma^H\mu(t)+2h_i^2[\frac{dE_i(t)}{dt}-\mu_i(t)]-\sum_j Q^{ij}[2\eta\mu_j-2h_j^2 E_j]+2E_i[\phi(i)-\sum_j Q^{ij}h_j^2]\\
=&-\lambda^H[\sum_jp_j(t)\frac{d\mu_j(t)}{dt}+\sum_j\frac{dp_j(t)}{dt}\mu_j(t)]-\gamma^H\sum_jp_j(t)\mu_j(t)+2h_i^2[\dot{E}_i(t)-\mu_i(t)]\\
&-\sum_j Q^{ij}(2\eta\mu_j-2h_j^2E_j)+2E_i[\phi(i)-\sum_j Q^{ij}h_j^2].\\
0=&2\eta\frac{d\mu_i(t)}{dt}+\sum_j[\gamma^Hp_j(t)+\lambda^H\sum_kp_k(t)Q^{kj}+2\eta Q^{ij}]\mu_j(t)+\lambda^H\sum_jp_j(t)\frac{d\mu_j(t)}{dt}\\
&-2h_i^2[\frac{dE_i(t)}{dt}-\mu_i(t)]-2\sum_jQ^{ij}h_j^2E_j-2E_i[\phi_i-\sum_jQ^{ij}h_j^2].\\
\end{aligned}
\end{equation*}
The boundary condition at $T$ is:
\begin{equation*}
\begin{aligned}
&0=h_i^1(T)=2\eta\mu_i(T)-2h_i^2 E_i(T)+\lambda^H\mu(T)\\
\Rightarrow&0=2\eta \mu_i(T)+2\Gamma(i)E_i(T)+\lambda^H\sum_jp_j(t)\mu_j(T).
\end{aligned}
\end{equation*}
Then
\begin{equation*}
\begin{cases}
\boldsymbol{0}=&[2\eta\boldsymbol{I}+\lambda^H\boldsymbol{e}\boldsymbol{p}^T]\boldsymbol{\dot{\mu}(t)}+[\gamma^H\boldsymbol{e}\boldsymbol{p}^T+2\eta\boldsymbol{Q}+\lambda^H\boldsymbol{e}\boldsymbol{p}^T\boldsymbol{Q}]\boldsymbol{\mu(t)}\\
&-2\boldsymbol{H(t)}[\boldsymbol{\dot{E}(t)}-\boldsymbol{\mu(t)}]-2[\boldsymbol{\Phi(t)}+\boldsymbol{Q}\boldsymbol{H(t)}]\boldsymbol{E(t)},\ t\in[t_k,t_{k+1}),\ k=0,...,K,\\
\boldsymbol{0}=&[2\eta\boldsymbol{I}+\lambda^H\boldsymbol{ep}^T]\boldsymbol{\mu(T)}+2\boldsymbol{\Gamma}\boldsymbol{E(T)}.
\end{cases}
\end{equation*}

\section{Proof of Theorem \ref{thm-epsilonNash}}
\label{proof-thm-epsilonNash}
For the ease of notation, we denote $\mathbbm{E}\left[\cdot\vert \mathcal{F}_t\right]$ by $\mathbbm{E}_t[\cdot]$, and we assume there are $(M+1)$ HFTs in the market. We define a functional $\Phi$ on $v, \delta, \overline{v}$ as:
\begin{equation*}
\begin{aligned}
&\Phi(v,\delta,\overline{v}) \\
=& X(0)P(0) + \mathbbm{E}_0\left\{ \gamma^H \int_0^T X(t) \left[\delta v(t) + (1-\delta) \overline{v}(t)  \right] dt + \gamma \sum_{k=1}^K \xi_k X(t_k)\right.\\
&\left.- \lambda^H \int_0^T v(t)\left[\delta v(t) + (1-\delta) \overline{v}(t)  \right] dt - \eta \int_0^T v(t)^2 dt -\Gamma(Y(T))X(T)^2 -  \int_0^T \phi(Y(t))X(t)^2 dt \right\},
\end{aligned}
\end{equation*}
where $v$ and $\overline{v}$ are $\mathcal{F}_t$-progressively measurable and square integrable processes, $\delta\geq 0$, and $dx(t)=v(t)dt$. 

W.L.O.G. we choose the 1st HFT as the representative. Note that for the 1st HFT, 
\begin{equation*}
J^{M+1}_{HFT}(v_1) = \Phi\left(v_1, \frac{1}{M+1},\frac{1}{M}\sum_{m=2}^{M+1} v_m \right), \, J_{HFT}(v_1) = \Phi(v_1, 0, \mu). 
\end{equation*}

Given $\delta$ and $\overline{v}$, denote the maximizer of $\Phi$ by $\Tilde{v}(\delta, \overline{v})$ and the maximum is denoted by $\Tilde{\Phi}$. We give the definition of $\Phi$ and $\Tilde{\Phi}$'s continuity in $v$ and $\overline{v}$.
Take $\Phi$ and $v$ as an example: given $\delta$ and $\overline{v}$, if
$\forall\varepsilon>0,$ there exists a $\delta>0,$ when $\mathbbm{E}_0\int_0^T |v(t)-v^{\delta}(t)|^2 dt<\delta,\ \text{a.s.},$ we have
$$|\Phi(v,\delta,\overline{v})-\Phi(v^{\delta},\delta,\overline{v})|<\varepsilon,\ \text{a.s.},$$
then we say that $\Phi$ is continuous in $v.$

It can be verified that $\Phi$ is jointly continuous in $(v,\delta,\overline{v})$ and strictly concave in $v$ when $\delta$ is sufficiently small. $\Tilde{v}$ is jointly continuous in $(\delta,\overline{v})$ and $\Tilde{\Phi}$ is jointly continuous in $(\delta,\overline{v})$. 

Let $v_1^{M+1,*}$ denote the optimal strategy in the $(M+1)$-HFT game where the other HFTs adopt the MFG strategy, then
\begin{equation*}
\begin{aligned}
& \vert J^{M+1}_{HFT}(v_1^*) -J^{M+1}_{HFT}\left( v_1^{M+1,*} \right)  \vert \\
= & \vert \Phi\left(v_1^*,\frac{1}{M+1},\frac{1}{M}\sum_{m=2}^{M+1} v_m^* \right) - \Phi\left(v_1^{M+1,*}, \frac{1}{M+1},\frac{1}{M}\sum_{m=2}^{M+1} v_m^* \right)   \vert \\
\leq & \vert \Phi\left(v_1^*,\frac{1}{M+1},\frac{1}{M}\sum_{m=2}^{M+1} v_m^* \right) - \Phi(v_1^*,0,\mu) \vert + \vert \Tilde{\Phi}(0,\mu) - \Tilde{\Phi}\left( \frac{1}{M+1},\frac{1}{M}\sum_{m=2}^{M+1} v_m^* \right) \vert,
\end{aligned}
\end{equation*}
and it's sufficient to show that $\mathbbm{E}_0 \int_0^T \vert \frac{1}{M}\sum_{m=2}^{M+1} v_m^*(t) - \mu(t) \vert^2 dt  \stackrel{\text{a.s.}}{\rightarrow} 0$ as $M\rightarrow \infty$.



%
\textbf{Step 1}: For $M$ players with $N$ types, let $m_i(t)$ denote the number of players in type $i$ at time $t$, and $\theta_i(t) = \frac{m_i(t)}{M}$, we show that
$\mathbbm{E}_0 \lVert \theta(t) - p(t) \rVert^2  \stackrel{\text{a.s.}}{\rightarrow} 0$ uniformly on $[0,T]$ as $M \rightarrow \infty$.

Apply It\^{o}'s lemma on $(\theta_i(t) - p_i(t))^2$ and take the conditional expectation,
\begin{equation*}
\begin{aligned}
&\mathbbm{E}_0\lvert \theta_i(t) - p_i(t) \rvert^2 \\
=& \lvert \theta_i(0) - p_i(0) \rvert^2 - 2\mathbbm{E}_0\left\{ \int_0^t \left[ \theta_i(s) - p_i(s) \right] \sum_j p_j(s) Q^{ji} ds \right\} \\
&-\mathbbm{E}_0\left\{ \int_0^t m_i(s) Q^{ii} \left[\left( \frac{m_i(s)-1}{M} -p_i(s) \right)^2 - \left(\frac{m_i(s)}{M} -p_i(s) \right)^2 \right] ds \right\}\\
&+\mathbbm{E}_0\left\{ \int_0^t \sum_{j\neq i} m_j(s) Q^{ji}\left[ \left(\frac{m_i(s)+1}{M} - p_i(s)  \right)^2 - \left(\frac{m_i(s)}{M} - p_i(s)  \right)^2 \right] ds \right\} \\
=& \lvert \theta_i(0) - p_i(0) \rvert^2 - 2\mathbbm{E}_0\left\{ \int_0^t \left[ \theta_i(s) - p_i(s) \right] \sum_j p_j(s) Q^{ji} ds \right\} \\
&+ \mathbbm{E}_0\left\{ \int_0^t \theta_i(s) Q^{ii} \left[2\left[ \theta_i(s) - p_i(s) \right] -\frac{1}{M} \right] ds  \right\} + \mathbbm{E}_0\left\{ \int_0^t \sum_{j\neq i} \theta_j(s) Q^{ji} \left[2\left[\theta_i(s) - p_i(s)  \right] +\frac{1}{M} \right] ds \right\} \\
=& \lvert \theta_i(0) - p_i(0) \rvert^2 + 2\mathbbm{E}_0\left\{ \int_0^t \left[\theta_i(s) - p_i(s) \right] \sum_j Q^{ji} \left[\theta_j(s) - p_j(s)  \right]  ds \right\} +\frac{1}{M} \mathbbm{E}_0\left\{ \int_0^t \sum_{j\neq i} Q^{ji}\left[ \theta_j(s) + \theta_i(s) \right] ds \right\}\\
\leq & \lvert \theta_i(0) - p_i(0) \rvert^2 + C_1 \mathbbm{E}_0\left\{ \int_0^t \lvert \theta_i(s) - p_i(s) \rvert \sum_j  \lvert \theta_i(s) - p_i(s)  \rvert  ds \right\} +\frac{C_1}{M},
\end{aligned}
\end{equation*}
where $C_1>0$ only depends on $Q$.

Sum from $i=1$ to $N$, apply Fubini's theorem and Cauchy-Schwarz inequality,
\begin{equation*}
\begin{aligned}
\mathbbm{E}_0 \lVert \theta(t) - p(t) \rVert^2 & \leq  \lVert \theta(0) - p(0)\rVert^2 + C_1 \int_0^t \mathbbm{E}_0 \left[\sum_j \lvert \theta_j(s) - p_j(s) \rvert  \right]^2 ds + \frac{C_1 N}{M}\\
& \leq \lVert \theta(0) - p(0)\rVert^2 + C_1 N \int_0^t \mathbbm{E}_0 \lVert \theta(s) - p(s) \rVert^2 ds + \frac{C_1 N}{M}.
\end{aligned}
\end{equation*}

By Gronwall's inequality,
\begin{equation*}
\mathbbm{E}_0 \lVert \theta(t) - p(t) \rVert^2  \leq \left( \lVert \theta(0) - p(0)\rVert^2 +\frac{C_1 N}{M} \right) e^{C_1 NT}.
\end{equation*}

By the strong law of large numbers, $ \lVert \theta(0) - p(0)\rVert^2\stackrel{\text{a.s.}}{\rightarrow}0$ as $M \rightarrow \infty$, and we reach the conclusion.
\smallskip

\textbf{Step 2}: We show that if the $j$-th HFT adopts the MFG strategy, $\{X_j(t)\}_{j=1}^{M+1}$ are a.s. uniformly bounded on $[0,T]$.
Omitting the subscript $j,$
\begin{equation*}
\begin{aligned}
\lvert X(t) \rvert &= \lvert X(0) + \int_0^t v^*(s,X(s),Y(s)) ds \rvert\\
&=\lvert X(0) + \int_0^t \frac{1}{2\eta}\left[h_{Y(s)}^1(s) + 2h_{Y(s)}^2(s) X(s) -\lambda^H \mu(s)   \right]ds \rvert\\
&\leq \lvert X(0) \rvert + C_2 \int_0^t \lvert X(s) \rvert ds + C_2,
\end{aligned}
\end{equation*}
where $C_2$ does not depend on $X(0)$ or $Y(0)$.

By Gronwall's inequality,
\begin{equation*}
\lvert X(t) \rvert \leq \left[ \lvert X(0) \rvert + C_2\right] e^{C_2 T}.
\end{equation*}
The conclusion follows that $\{X_j(0)\}_{j=1}^{M+1}$ are a.s. uniformly bounded.





\textbf{Step 3}: If the 2nd to the $(M+1)$-th HFTs adopt the MFG strategy, define $Z_i(t) = \frac{1}{M}\sum_{m=2}^{M+1} 1_{[Y_m(t) = i]} X_m(t)$, and $\nu_i(t) = p_i(t) E_i(t)$, we show that $\mathbbm{E}_0 \lVert  Z(t) - \nu(t)  \Vert ^2 \stackrel{\text{a.s.}}{\rightarrow}0$ uniformly on $[0,T]$ as $M\rightarrow \infty$.

Apply It\^{o}'s lemma on $\left[ Z_i(t) - \nu_i(t)\right]^2$ and take the conditional expectation,
\begin{equation*}
\begin{aligned}
&\mathbbm{E}_0\left[Z_i(t) - \nu_i(t)  \right]^2 \\
=& \left[ Z_i(0) - \nu_i(0) \right]^2 + 2\mathbbm{E}_0\left\{ \int_0^t \left[Z_i(s) - \nu_i(s)  \right]\left[ \frac{1}{M} \sum_{m=2}^{M+1} 1_{[Y_m(s)=i]}v_m^*(s,X_m(s),Y_m(s)) -p_i(s)\mu_i(s) - \sum_{j}Q^{ji} \nu_j(s) \right] ds \right\} \\
& - \mathbbm{E}_0\left\{ \int_0^t \sum_{m=2}^{M+1} 1_{[Y_m(s)=i]} Q^{ii}\left[\left( Z_i(s)-\nu_i(s)-\frac{1}{M}X_m(s)  \right)^2 - \left( Z_i(s) - \nu_i(s) \right)^2  \right]ds \right\} \\
&+ \mathbbm{E}_0\left\{\int_0^t \sum_{j\neq i} \sum_{m=2}^{M+1} 1_{[Y_m(s)=j]} Q^{ji}  \left[\left( Z_i(s)-\nu_i(s)+\frac{1}{M}X_m(s)  \right)^2 - \left( Z_i(s) - \nu_i(s) \right)^2  \right]ds \right\} \\
=& \left[ Z_i(0) - \nu_i(0) \right]^2 + 2\mathbbm{E}_0\left\{ \int_0^t \left[Z_i(s) - \nu_i(s)  \right]\left[ \frac{1}{M} \sum_{m=2}^{M+1} 1_{[Y_m(s)=i]}v_m^*(s,X_m(s),Y_m(s)) -p_i(s)\mu_i(s) - \sum_{j}Q^{ji}\nu_j(s) \right] ds \right\} \\
&+ \mathbbm{E}_0\left\{\int_0^t Q^{ii}\left( 2 Z_i(s) \left [ Z_i(s) - \nu_i(s) \right] - \frac{1}{M^2}\sum_{m=2}^{M+1} 1_{[Y_m(s)=i]}X_m(s)^2  \right) ds \right\} \\
&+ \mathbbm{E}_0\left\{\int_0^t \sum_{j\neq i} Q^{ji} \left(2Z_j(s)\left[Z_i(s) - \nu_i(s)  \right]   + \frac{1}{M^2}\sum_{m=2}^{M+1} 1_{[Y_m(s)=j]} X_m(s)^2\right) ds  \right\} \\
\leq & \left[ Z_i(0) - \nu_i(0) \right]^2 + C_3 \mathbbm{E}_0\left\{\int_0^t \vert Z_i(s)-\nu_i(s)\vert \left(\sum_j \vert Z_j(s)-\nu_j(s)\vert\right) ds \right\} \\
& + C_3 \mathbbm{E}_0\left\{ \int_0^t \lvert Z_i(s) - \nu_i(s)  \rvert \left(  \lvert \theta_i(s) - p_i(s) \rvert + \lvert Z_i(s) - \nu_i(s) \rvert \right) ds \right\} + \frac{C_3}{M^2}\int_0^t \sum_{m=2}^{M+1}\mathbbm{E}_0 X_m(s)^2 ds.
\end{aligned}
\end{equation*}

Sum from $i=1$ to $N$ and apply Cauchy-Schwarz inequality,
\begin{equation*}
\begin{aligned}
& \mathbbm{E}_0 \lVert Z(t) - \nu(t) \rVert ^2 \\
\leq & \lVert Z(0) - \nu(0) \rVert^2 + C_3 (N+1) \int_0^t  \mathbbm{E}_0 \lVert Z(s)-\nu(s) \rVert^2 ds  \\
& + C_3 \int_0^t \sqrt{\mathbbm{E}_0 \lVert Z(s) - \nu(s)\rVert^2 \cdot \mathbbm{E}_0 \lVert \theta(s) - p(s)\rVert^2} ds + \frac{C_3 NT}{M^2} \sum_{m=2}^{M+1} \left[X_m(0)^2 + 2C_2 \vert X_m(0) \vert + C_2^2  \right] e^{2C_2 T}.
\end{aligned}
\end{equation*}
In step 2, we have proved that $\{X_j(t)\}_{j=1}^{M+1}$ is uniformly bounded. As a result, $\mathbbm{E}_0 \lVert Z(t) - \nu(t)\rVert^2$ is also uniformly bounded. What's more, $\mathbbm{E}_0 \lVert \theta(t) - p(t) \rVert^2\stackrel{\text{a.s.}}{\rightarrow}0$ uniformly, so for $M$ that is large enough,
\begin{equation*}
\begin{aligned}
 \mathbbm{E}_0 \lVert Z(t) - \nu(t) \rVert ^2 
< & \lVert Z(0) - \nu(0) \rVert^2 + C_3 (N+1) \int_0^t  \mathbbm{E}_0 \lVert Z(s)-\nu(s) \rVert^2 ds \\
&+ \frac{C_3 NT}{M^2} \sum_{m=2}^{M+1} \left[X_m(0)^2 + 2C_2 \vert X_m(0) \vert + C_2^2  \right] e^{2C_2 T}, \text{ a.s.}
\end{aligned}
\end{equation*}

By Gronwall's inequality,
\begin{equation*}
\mathbbm{E}_0 \lVert Z(t) - \nu(t) \rVert ^2 < \alpha e^{C_3(N+1)T}, \text{ a.s.}
\end{equation*}
where
\begin{equation*}
\begin{aligned}
\alpha &= \lVert Z(0) - \nu(0) \rVert^2 + \frac{C_3 NT}{M^2} \sum_{m=2}^{M+1} \left[X_m(0)^2 + 2C_2 \vert X_m(0) \vert + C_2^2  \right] e^{2C_2 T}.
\end{aligned}
\end{equation*}

By the strong law of large numbers, $\alpha\stackrel{\text{a.s.}}{\rightarrow} 0$ as $M\rightarrow \infty$, we reach the conclusion.

\textbf{Step 4}: If the 2nd to the $(M+1)$-th HFTs adopt the MFG strategy, denote $\hat{v}^M(t) = \frac{1}{M}\sum_{m=2}^{M+1} v_m^*(t)$, we show that $\mathbbm{E}\lvert \hat{v}^M(t) - \mu(t) \rvert^2 \stackrel{\text{a.s.}}{\rightarrow} 0$ uniformly on $[0,T]$.

\begin{equation*}
\begin{aligned}
&\mathbbm{E}_0\lvert \hat{v}^M(t) - \mu(t) \rvert^2 \\
\leq& \mathbbm{E}_0\left[ \sum_j \frac{\lvert h_j^1(t) - \lambda^H \mu(t)\rvert}{2\eta} \lvert \theta_j(t) - p_j(t)\rvert + \frac{\vert h_j^2(t)\rvert}{\eta}\lvert Z_j(t) - \nu_j(t) \rvert \right]^2 \\
\leq &C_5 \mathbbm{E}_0\lVert \theta(t) - P(t) \rVert^2 + C_5 \mathbbm{E}_0 \lVert Z(t) - \nu(t)\rVert^2.
\end{aligned}
\end{equation*}
Then the result follows Step 1 and Step 3.

Steps 1-4 imply that $\mathbbm{E}_0 \int_0^T \vert \frac{1}{M}\sum_{m=2}^{M+1} v_m^*(t) - \mu(t) \vert^2 dt \stackrel{\text{a.s.}}{\rightarrow} 0$ as $M\rightarrow \infty$, thus we have proved the theorem.






\section{Proof of Theorem \ref{thm-epsilonNash2}}
\label{proof-thm-epsilonNash2}
Define $\Psi$ as a function of $\xi=(\xi_1,...,\xi_K)$, $\overline{X}=(\overline{X}_0,\overline{X}_1,...,\overline{X}_K)$ and $\overline{v}=(\overline{v}_1,...,\overline{v}_K)$ as:
\begin{equation*}
\Psi(\xi,\overline{X},\overline{v})=\mathbbm{E}\left\{\sum_{k=1}^K(-\xi_k)\left[ \gamma\sum_{j=1}^k\xi_j+\gamma^H (\overline{X}_k-\overline{X}_0) dt +\lambda^H\overline{v}_k+(\lambda+\eta_0)\xi_k \right] \right\}
\end{equation*}
where $\xi\in\mathbbm{R}^K$ is deterministic, $\overline{X}_k$ and $\overline{v}_k$ are $\mathcal{F}_{t_k}$-measurable. Note that (by omitting constant terms)
\begin{equation*}
J_{LT}^M(\xi)=\Psi(\xi,\overline{X}^M, \overline{v}^M),\,J_{LT}(\xi)=\Psi(\xi,E,\mu).
\end{equation*}
where $\overline{X}^M=(\overline{X}^M(0),\overline{X}^M(t_1),...,\overline{X}^M(t_K))$, $\overline{v}^M=(\overline{v}^M(t_1),...,\overline{v}^M(t_K))$, $E=(E(0),E(t_1),...,E(t_K))$ and $\mu=(\mu(t_1),...,\mu(t_K))$.

It can be verified that $\Psi$ is jointly continuous in $(\xi,\overline{X},\overline{v})$ and strictly concave in $\xi$. For given $(\overline{X},\overline{v})$, denote $\Tilde{\xi}$ as the maximizer, since it is just solving a linear system in order to find $\Tilde{\xi}$, $\Tilde{\xi}(\overline{X},\overline{v})$ is jointly continuous in $(\overline{X},\overline{v})$. Furthermore, if we denote $\Tilde{\Psi}$ as the maximum, then $\Tilde{\Psi}$ is jointly continuous in $(\overline{X},\overline{v})$.

Denote the optimal strategy of LT by $\xi^{M,*}$ when there are $M$ HFTs and each adopting the MFG strategy. The goal is to show that $\vert J^M_{LT}(\xi^*;v_1^*,...,v_M^*)-J^M_{LT}(\xi^{M,*};v_1^*,...,v_M^*)\vert \rightarrow 0$ as $M\rightarrow \infty$.

Observing that
\begin{equation*}
\vert J^M_{LT}(\xi^*;v_1^*,...,v_M^*)-J^M_{LT}(\xi^{M,*};v_1^*,...,v_M^*)\vert \leq \vert \Psi(\xi^*,\overline{X}^M,\overline{v}^M) - \Psi(\xi^*,E,\mu)\vert + \vert \Tilde{\Psi}(E,\mu) - \Tilde{\Psi}(\overline{X}^M,\overline{v}^M) \vert,
\end{equation*}
by the weak law of large numbers, $\mathbbm{E}\vert \overline{X}^M(t_k)-E(t_k)\vert^2 \rightarrow 0, k=0,1,...,K$ and $\mathbbm{E}\vert \overline{v}^M(t_k)-\mu(t_k)\vert^2 \rightarrow 0, k=1,...,K$, and we reach the conclusion.

\newpage
\bibliographystyle{unsrt}
\bibliography{bibfile}
\addcontentsline{toc}{section}{Reference} 
\end{document}